\pdfoutput=1
 \documentclass[11pt,letterpaper]{article}
 \def\arXiv{1} 
\newcommand{\notarxiv}[1]{foo}
\newcommand{\arxiv}[1]{ba}
\ifdefined \arXiv
	\renewcommand{\arxiv}[1]{#1}%
	\renewcommand{\notarxiv}[1]{\ignorespaces}%
\else%
	\renewcommand{\arxiv}[1]{\ignorespaces}%
	\renewcommand{\notarxiv}[1]{#1}%
\fi%

\usepackage{thm-restate}
\usepackage{bm}
\usepackage{bbm}

\arxiv{
\usepackage{vmargin,fancyhdr}   %
\usepackage{tikz}
\usetikzlibrary{positioning,chains,fit,shapes,calc}
}

\usepackage{subcaption}
\notarxiv{
\usepackage{algorithm}
\usepackage{algorithmic}
}
\arxiv{
\usepackage[linesnumbered,ruled,vlined]{algorithm2e}
\usepackage{xcolor}

\SetKwComment{Comment}{$\triangleright$\ }{}
\SetCommentSty{mycommfont}
\SetKwInput{KwInput}{Input}                %
\SetKwInput{KwOutput}{Output}              %
\SetKwInput{KwReturn}{Return}              %
}
\usepackage{enumerate}

\usepackage{amsmath,amssymb, amsthm}    %
\usepackage{verbatim}           %
\usepackage{xspace}             %
\usepackage{graphicx,float}     %
\usepackage{ifthen,calc}        %
\usepackage{textcomp}           %
\usepackage{fancybox}           %
\usepackage{hhline}             %
\usepackage{float}              %

\arxiv{
\usepackage[vflt]{floatflt}     %
\usepackage[small,compact]{titlesec}
\usepackage{setspace}
}

\usepackage{color}
\definecolor{ForestGreen}{rgb}{0.1333,0.5451,0.1333}
\usepackage{thumbpdf}
\arxiv{\usepackage[letterpaper,
colorlinks,linkcolor=ForestGreen,citecolor=ForestGreen,
bookmarks,bookmarksopen,bookmarksnumbered]
{hyperref}}

\arxiv{
\setpapersize{USletter}
\setmarginsrb{1in}{.5in}        %
{1in}{.5in}        %
{.25in}{.25in}     %
{.25in}{.5in}      %
\setlength{\marginparwidth}{.75in}
\setlength{\marginparsep}{.05in}

\newcommand{\showccc}[0]{0}
\newcommand{\ccc}[2][nothing]{%
	\ifthenelse{\showccc=0}{}{
		\ensuremath{^{\Lsh\Rsh}}\marginpar{\raggedright\tiny\textsf{%
				\ifthenelse{\equal{#1}{nothing}}{}{\textbf{#1}\\}#2}}}}
\pagestyle{fancy}
\newcounter{hours}\newcounter{minutes}
\newcommand{\hhmm}{%
	\setcounter{hours}{\time/60}%
	\setcounter{minutes}{\time-\value{hours}*60}%
	\ifthenelse{\value{hours}<10}{0}{}\thehours:%
	\ifthenelse{\value{minutes}<10}{0}{}\theminutes}
\lhead{}
\chead{}
\ifthenelse{\showccc=0}{\rhead{}}{\rhead{\today \ [\hhmm]}}
\lfoot{}
\cfoot{\thepage}
\rfoot{}
}

\arxiv{

\newtheorem{proposition}{Proposition}
\newtheorem{corollary}{Corollary}
\newtheorem{definition}{Definition}

\newtheorem{remark}{Remark}
\newtheorem{lemma}{Lemma}
}

\arxiv{}

\usepackage{float}
\floatstyle{plain}\newfloat{myfig}{t}{figs}[section]
\floatname{myfig}{\textsc{Figure}}
\floatstyle{plain}\newfloat{myalg}{H}{algs}[section]
\floatname{myalg}{}
\newcommand{\defeq}{:=}
\newcommand{\norm}[1]{\left\lVert#1\right\rVert}
\newcommand{\inprod}[2]{\left\langle#1, #2\right\rangle}
\newcommand{\eps}{\epsilon}
\newcommand{\lam}{\lambda}

\newcommand{\argmin}{\textup{argmin}} 
\newcommand{\R}{\mathbb{R}}

\newcommand{\N}{\mathbb{N}}
\newcommand{\diag}[1]{\textbf{\textup{diag}}\left(#1\right)}
\newcommand{\half}{\frac{1}{2}}
\newcommand{\thalf}{\tfrac{1}{2}}
\newcommand{\1}{\mathbf{1}}
\newcommand{\E}{\mathbb{E}}
\newcommand{\Var}{\textup{Var}}

\newcommand{\opt}{\textup{OPT}}
\newcommand{\xset}{\mathcal{X}}
\newcommand{\yset}{\mathcal{Y}}

\newcommand{\ma}{\mathbf{A}}
\newcommand{\mb}{\mathbf{B}}
\newcommand{\mr}{\mathbf{R}}
\newcommand{\mx}{\mathbf{X}}
\newcommand{\mm}{\mathbf{M}}

\newcommand{\ai}{\ma_{i:}}
\newcommand{\aj}{\ma_{:j}}
\newcommand{\id}{\mathbf{I}}
\newcommand{\tO}{\widetilde{O}}
\newcommand{\nnz}{\textup{nnz}}

\newcommand{\mw}{\mathbf{W}}
\newcommand{\Par}[1]{\left(#1\right)}
\newcommand{\Brack}[1]{\left[#1\right]}
\newcommand{\Brace}[1]{\left\{#1\right\}}

\newcommand{\mzero}{\mathbf{0}}
\newcommand{\tx}{\tilde{x}}

\newcommand{\bu}{\bar{u}}

\newcommand{\bx}{\bar{x}}
\newcommand{\by}{\bar{y}}
\newcommand{\bv}{\bar{v}}
\newcommand{\blam}{\bar{\lam}}

\newcommand{\Sherman}{\mathsf{ACMirrorProx}}
\newcommand{\LSSherman}{\mathsf{LowSpaceACMirrorProx}}
\newcommand{\AltMin}{\mathsf{AltMin}}

\newcommand{\RedSize}{\mathsf{VertexReduction}}
\newcommand{\Feasible}{\mathsf{RoundToFeasible}}

\newcommand{\x}{^\mathsf{x}}
\newcommand{\y}{^\mathsf{y}}
\newcommand{\0}{\mathbf{0}}
\newcommand{\rsum}{r_{\textup{sum}}}
\newcommand{\tmx}{\tilde{\mx}}
\newcommand{\lsum}{l_{\textup{sum}}}
\newcommand{\lscale}{l_{\textup{scale}}}
\newcommand{\rscale}{r_{\textup{scale}}}
\newcommand{\bM}{\bar{M}}

\newcommand{\RandomSample}{\mathsf{RandomSample}}
\newcommand{\hx}{\hat{x}}

\newcommand{\cmax}{C_{\max}}
\newcommand{\tG}{\widetilde{G}}
\newcommand{\tmb}{\widetilde{\mb}}
\newcommand{\tV}{\widetilde{V}}
\newcommand{\tE}{\widetilde{E}}
\newcommand{\tw}{\tilde{w}}

\newcommand{\oracle}{\mathcal{O}}
\newcommand{\dmcm}{d_{\textup{MCM}}}
\newcommand{\Overflow}{\textup{Overflow}}

\newcommand{\Cut}{\mathsf{Cut}}
\newcommand{\Link}{\mathsf{Link}}
\newcommand{\LCA}{\mathsf{LCA}}
\newcommand{\ChangeRoot}{\mathsf{ChangeRoot}}

\newcommand{\Min}{\mathsf{Min}}
\newcommand{\Add}{\mathsf{Add}}
\newcommand{\Sum}{\mathsf{Sum}}

\newcommand{\otilde}{\widetilde{O}}
\newcommand{\poly}{\mathrm{poly}}

  \usepackage{nth}
  \usepackage{intcalc}

\arxiv{
\usepackage{url}

\setlength\parindent{0pt}
\setlength{\parskip}{4pt}
}

\arxiv{
\begin{document}

	\begin{titlepage}
		\def\thepage{}
		\thispagestyle{empty}
		
		\title{Semi-Streaming Bipartite Matching \\in Fewer Passes and Optimal Space}

		\date{}
		\author{
			Sepehr Assadi\thanks{Rutgers University, {\tt sepehr.assadi@rutgers.edu}}
			\and
			Arun Jambulapati\thanks{Stanford University, {\tt \{jmblpati, yujiajin, sidford, kjtian\}@stanford.edu}}
			\and
			Yujia Jin\footnotemark[2]
			\and
			Aaron Sidford\footnotemark[2]
			\and
			Kevin Tian\footnotemark[2]
		}
		
		\maketitle

\arxiv{
\abstract{We provide $\otilde(\epsilon^{-1})$-pass semi-streaming algorithms for computing $(1-\epsilon)$-approximate maximum cardinality matchings in bipartite graphs. Our most efficient methods are deterministic and use optimal, $O(n)$, space, improving upon the space complexity of the previous state-of-the-art $\otilde(\epsilon^{-1})$-pass 
algorithm of \cite{AhnG18}. To obtain our results we provide semi-streaming adaptations of more general continuous optimization tools. Further, we leverage these techniques to obtain improvements for streaming variants of approximate linear programming, 
optimal transport, exact matching, transshipment, and shortest path problems.
}
}

	\end{titlepage}
}

\notarxiv{

\title{Semi-Streaming Bipartite Matching in Fewer Passes and Optimal Space} 

\titlerunning{Semi-Streaming Bipartite Matching in Fewer Passes and Optimal Space} 

\author{Sepehr Assadi}{Rutgers University, United States}{sepehr.assadi@rutgers.edu}{}{}

\author{Arun Jambulapati}{Stanford University, United States}{jmblpati@stanford.edu}{}{}

\author{Yujia Jin}{Stanford University, United States}{yujiajin@stanford.edu}{}{}

\author{Aaron Sidford}{Stanford University, United States}{sidford@stanford.edu}{}{}

\author{Kevin Tian}{Stanford University, United States}{kjtian@stanford.edu}{}{}

\authorrunning{S. Assadi, A. Jambulapati, Y. Jin, A. Sidford and K. Tian} 

\Copyright{Sepehr Assadi, Arun Jambulapati, Yujia Jin, Aaron Sidford and Kevin Tian}



\relatedversion{A previous version whose results are subsumed by this version is hosted on arXiv.} \relatedversiondetails[linktext={arXiv:2011.03495}]{Outdated version}{https://arxiv.org/abs/2011.03495} %

\acknowledgements{We thank Alireza Farhadi for helpful conversations and for introducing the authors to the literature on semi-streaming matching. Researchers are supported in part by a Microsoft Research Faculty Fellowship, NSF CAREER Award CCF-1844855, NSF Grant CCF-1955039, a PayPal research gift, a Sloan Research Fellowship, and a Stanford Graduate Fellowship.}%

\SeriesVolume{00}
\ArticleNo{00}

\begin{document}

\maketitle

\arxiv{
\abstract{We provide $\otilde(\epsilon^{-1})$-pass semi-streaming algorithms for computing $(1-\epsilon)$-approximate maximum cardinality matchings in bipartite graphs. Our most efficient methods are deterministic and use optimal, $O(n)$, space, improving upon the space complexity of the previous state-of-the-art $\otilde(\epsilon^{-1})$-pass 
algorithm of \cite{AhnG18}. To obtain our results we provide semi-streaming adaptations of more general continuous optimization tools. Further, we leverage these techniques to obtain improvements for streaming variants of approximate linear programming, 
optimal transport, exact matching, transshipment, and shortest path problems.
}
}

}

\arxiv{
\cleardoublepage
\pagenumbering{gobble}
\hypersetup{linkcolor=black}
\tableofcontents
\hypersetup{linkcolor=ForestGreen}
\cleardoublepage
\pagenumbering{arabic}
}
\newpage 

\section{Introduction}
\label{sec:intro}

We study the fundamental problem of designing semi-streaming algorithms for  maximum cardinality matching in bipartite graphs (MCM). We consider the insertion-only problem where an unknown $n$-vertex $m$-edge undirected bipartite graph is presented as a stream of edge insertions. Our goal is to compute an $\eps$-approximate MCM (a matching with value $\ge 1 - \eps$ times the optimum) in the semi-streaming model~\cite{FeigenbaumKMSZ05}, i.e.\ using $\otilde(n)$ space,\footnote{Throughout, we measure space in machine words of size $\Theta(\log{n})$ and use $\otilde$ to hide factors polylogarithmic in $n$, $\eps^{-1}$, and $\cmax$, the largest edge weight of the relevant problem. In all our applications, without loss of generality, $\eps^{-1}$ and $\cmax$ are $O(\text{poly}(n)$). Consequently, in our applications an $\otilde(n)$ space algorithm is a  semi-streaming, i.e.\ $O(n \cdot \text{poly}(\log(n))$, algorithm for our applications.}  and as few passes as possible.

Due to its canonical and prevalent nature, MCM is  well-studied in the graph streaming literature. In a single pass, a simple greedy algorithm achieves a $\frac12$-approximation in $O(n)$ space (cf.\ Lemma~\ref{lem:greedy}) and obtaining better than a $(\frac1{1+\ln{2}})\approx 0.59$ approximation requires $n^{1+\Omega(1/\log\log{(n)})}$ space~\cite{Kapralov21} (see also \cite{GoelKK12,Kapralov13}). Correspondingly, there is a line of research on designing multi-pass algorithms for  $\eps$-approximating MCM. In this paper, we focus on  $\otilde(\poly(\eps^{-1}))$-pass semi-streaming algorithms. Here, the state-of-the-art includes the $O(\eps^{-2})$-pass $O(n)$-space algorithm of~\cite{AssadiLT21} (see also~\cite{AhnG13})
and the $O({\eps^{-1}}\log{n})$-pass $\otilde(n \cdot \text{poly}(\eps^{-1}))$-space algorithm of~\cite{AhnG18}.\footnote{Streaming MCM has been studied under several other variants and parameter settings; see Section~\ref{sec:prev} for details.}

For different ranges of $\epsilon$, these semi-streaming MCM algorithms provide non-trivial trade-offs between space and pass complexity. While~\cite{AhnG18} provides an improved $\otilde(\eps^{-1})$ passes
compared to $O(\eps^{-2}$) in~\cite{AssadiLT21}, the space complexity of~\cite{AhnG18} is larger by  $\text{poly}(\eps^{-1})$ factors\footnote{We remark that the introduction and main theorem statements of \cite{AhnG18} are written for constant $\eps$. However, for subconstant $\eps$ it is straightforward to see their space requirement incurs a $\text{poly}(\eps^{-1})$ dependence.} and therefore only adheres to the semi-streaming restriction of using $\otilde(n)$ space when 
$\eps \geq (\text{poly}\!\log{(n)})^{-1}$. The central question we address in this paper is whether this trade-off is necessary:

\begin{quote}
\centering{\emph{Is it possible to obtain an $\otilde(\epsilon^{-1})$-pass semi-streaming algorithm for $\eps$-approximating maximum bipartite matchings for the entire range of $\epsilon > 0$?}}
\end{quote}

In this paper we answer this question in the affirmative. We provide multiple ways to leverage recent advances in continuous and combinatorial optimization for this goal. For example, in Appendix~\ref{sec:box} we show how to obtain this result by a careful application of recent results on the box-constrained Newton method~\cite{CMTV17} and streaming sparsification~\cite{McGregor14}.

Our main result is a simple algorithm that further improves the space complexity. We provide an  $O(n)$-space and $O(\eps^{-1}\cdot\log{(\eps^{-1})\cdot \log{n})}$-pass algorithm based on recent advances in area convexity and extragradient methods \cite{Sherman17, JambulapatiST19, CohenST21}. This space-dependency is optimal, as $\Omega(n)$ space is needed to simply output the final matching. Further, this algorithm is deterministic, in contrast to the previous $\otilde(\eps^{-1})$-pass
algorithm of~\cite{AhnG18}, and the method's runtime is $\otilde(m \cdot \eps^{-1})$. %

To obtain this result we introduce a more general algorithmic framework of independent interest. We design semi-streaming algorithms which can approximately determine the value of and produce low-space \emph{implicit} fractional solutions to linear programs, given in the form of box-simplex games, when rows of the constraint matrix are presented in a stream. We obtain our MCM results by applying this optimization method to a linear programming representation of the problem and show that this implicit solution can be converted into a sparse solution in low space.

We believe our results demonstrate the power of recent optimization advances for solving streaming problems. Beyond  resolving the space complexity of $\otilde(\epsilon^{-1})$-pass algorithms, we show that our techniques extend to yield the following additional results.

\newpage
\begin{itemize}
\item \textbf{Exact MCM in $o(n)$ passes for all densities.} Recent work~\cite{LiuSZ20} asked whether the $O(n\log{n})$-pass semi-streaming algorithm for \emph{exact} MCMs following from a careful implementation of the classical Hopcroft-Karp algorithm~\cite{HopcroftK73} is improvable. The authors designed an $\otilde(n)$-space algorithm based on interior-point methods using $\otilde(\sqrt{m})$ passes, which is $o(n)$ passes except for in dense graphs. By combining our approximate MCM method with recent advances in streaming reachability \cite{LiuJS19}, we obtain an $O(n^{\frac34+o(1)})$-pass, $\otilde(n)$-space algorithm for exact MCM, thereby answering the main open problem of~\cite{LiuSZ20} for all densities. Interestingly, this improvement is a direct byproduct of resolving our motivating question as the space-pass trade-offs in~\cite{AhnG18} and~\cite{AssadiLT21} prevent either method from obtaining nontrivial semi-streaming exact MCM algorithms even when combined with~\cite{LiuJS19} (see Remark~\ref{rem:space-pass-efficiency}).

\item \textbf{Optimal transport.} Closely-related to computing MCMs, the discrete optimal transport problem has received widespread recent interest due to applications in machine learning~\cite{Cuturi13, AltschulerWR17}. Several recent works have designed different optimization algorithms (influenced by developments for matching)~\cite{Quanrud19, blanchet2018towards,JambulapatiST19} solving the problem on a support size of $n$ using $\widetilde{O}(n^2\eps^{-1})$ total work. We provide the first semi-streaming algorithm in this setting within $\widetilde{O}(\eps^{-1})$ passes, $O(n)$ space, and comparable work.  

\item \textbf{Transshipment and shortest path.} Yet another application of our method is to solving the transshipment problem, a type of uncapacitated minimum cost flow problem, on undirected graphs. 
Several recent works have focused on this problem in streaming and parallel models~\cite{BeckerKKL17,Li20,AndoniSZ20} and used this to obtain approximate shortest path algorithms. 
We provide a semi-streaming algorithm that achieves an $\eps$-approximation to transshipment within $\widetilde{O}(\eps^{-1})$ passes, $\widetilde{O}(n)$ space, and $\widetilde{O}(m\eps^{-1})$ work. 
As a direct corollary of this result, we obtain $\eps$-approximate semi-streaming algorithms for the $s$-$t$ shortest path problem with same complexities. Our results on transshipment and shortest path improve upon the previous best semi-streaming algorithms of~\cite{BeckerKKL17} by $\widetilde{O}(\eps^{-1})$ factors in passes and work. 
\end{itemize}

\subsection{Problem setup}
\label{sec:setup}

The basic (and motivating) problem we consider in this work is that of computing an approximate MCM in bipartite graph, $G = (V, E)$,  given as an insertion-only stream defined as follows.

\begin{definition}[Semi-streaming graph model]\label{def:semistream}
In the \emph{semi-streaming graph model}, a graph $G = (V, E)$ with vertex set $V = [n]$ is presented to the algorithm as an arbitrarily ordered stream of edges $(u,v)$ for $u, v \in V$ (the tuple contains the weight for weighted graphs). 
The algorithm can read this stream in sequential \emph{passes} and is constrained to use $\otilde(n)$ space.
\end{definition}
We typically let $n = |V|$ and $m = |E|$. $M^*$ denotes the the size of the MCM in $G$, and $L, R \subseteq V$ with $L \cap R = \emptyset$ and $L \cup R = V$ denotes the vertex bipartition. A key goal of this paper is to efficiently compute approximate MCMs in the semi-streaming graph model. Formally, we refer to  any flow $x \in [0, 1]^E$ such that the total flow adjacent to any $v \in V$ is $\le 1$ as a \emph{fractional matching} and call the matching integral if $x \in \{0, 1\}^E$; we say such a matching is a $\eps$-approximate MCM if $\norm{x}_1\geq (1 - \epsilon) M^*$. 
The outputs of our algorithms typically include both \emph{the approximate problem value} and \emph{an approximate solution}. For MCM problems with value $M^*$, the algorithm returns this value within a $(1-\eps)$ factor and an integral matching that achieves this approximate value.

To obtain our main result, we develop streaming algorithms for
  a more general set of problems. These problems are linear programming in the form of \emph{box-simplex games}, i.e. bilinear minimax games between a \emph{box}, or $\ell_\infty$-constrained player, and \emph{simplex}, or nonnegative $\ell_1$-constrained player:
\begin{equation}\label{eq:boxsimplex}
	\min_{x \in \Delta^m} \max_{y \in [-1, 1]^n} y^\top\ma^\top x + c^\top x -b^\top y.
\end{equation}
Maximizing over $y$ for a fixed $x$, problem \eqref{eq:boxsimplex} is equivalent to the following variant of  $\ell_1$-regression:
\begin{equation}\label{eq:boxsimplexprimal}
\min_{x \in \Delta^m} c^\top x + \norm{\ma^\top x - b}_1.
\end{equation}
We call $x \in \Delta^m$ an $\eps$-approximate minimizer for \eqref{eq:boxsimplexprimal} if it satisfies $
c^\top x + \norm{\ma^\top x - b}_1 \le \min_{x' \in \Delta^m} c^\top x' + \norm{\ma^\top x' - b}_1 + \eps$.
We relate \eqref{eq:boxsimplex}, \eqref{eq:boxsimplexprimal} to MCM via reduction: we construct an appropriate instance of \eqref{eq:boxsimplex}, \eqref{eq:boxsimplexprimal} such that any approximate minimizer to \eqref{eq:boxsimplexprimal} can be efficiently converted into an approximate MCM. We note the semi-streaming solver we develop in Section~\ref{sec:value} attains an approximate solution for \eqref{eq:boxsimplex}, but in applications we only ever use that the $x$ block approximately minimizes \eqref{eq:boxsimplexprimal}.
 
 Interestingly, we provide a general streaming algorithm for approximating the value of \eqref{eq:boxsimplex}, \eqref{eq:boxsimplexprimal} (and implicitly representing an optimal solution) in the following access model. 

\begin{definition}[Semi-streaming matrix model]\label{def:semistream_matrix}
	In the \emph{semi-streaming matrix model}, a matrix $\ma \in \R^{m \times n}$ and a vector $c \in \R^m$ are presented to the algorithm respectively as an arbitrarily ordered stream 
	of rows $\{\ai\}_{i \in [m]}$, and a similarly ordered stream of coordinates $\{c_i\}_{i \in [m]}$ (the ordering on $[m]$ is arbitrary, but each $\{\ai, c_i\}$ pair is given together). The algorithm can read this stream in sequential \emph{passes} and is constrained to use $\otilde(n)$ space.
\end{definition}

By choosing $\ma$ to be the incidence matrix of a graph and setting $c$ to be the weights if applicable, Definition~\ref{def:semistream_matrix} generalizes Definition~\ref{def:semistream}. Note that obtaining semi-streaming algorithms for this problem is nontrivial only when $m = \omega(n)$. 

\medskip
\noindent
\textbf{Further applications.} In \notarxiv{Appendix~\ref{sec:app}}\arxiv{Section~\ref{ssec:further}}, we study additional combinatorial optimization problems which are also reducible to box-simplex games. In the problem of \emph{discrete optimal transportation}, a complete bipartite graph $G = (V, E)$ where $E = L \times R$ has associated demands $\ell \in \Delta^L$, $r \in \Delta^R$, and edge costs $c\in \R^E$. The demands $\ell$, $r$ are probability distributions on domains of equal size, and the goal is to compute the minimum cost transport plan which attains the prescribed marginals on the vertex sets $L$ and $R$. The semi-streaming access model gives the cost of each edge as it is presented; note that the marginals $\ell$, $r$ can be stored explicitly in $O(n)$ space. In this setting, our designed algorithm gives an approximate value of the optimal transportation, as well as a fractional transportation plan $x\in\Delta^{E}$ that meets the demands exactly (has the correct marginals at every vertex) and achieves an approximately optimal value.
Using our techniques we also give a result for computing maximum weight matchings (MWMs), which is competitive with the state-of-the-art when the optimal matching is highly saturated (e.g.\ when the MWM value is within a constant factor of $\norm{w}_\infty n$ in a graph with weights $w$, the largest possible value). 

Finally, in~Section~\ref{sec:tran} we go beyond ``bipartite matching-type'' applications of our framework and consider the transshipment problem. Given a demand vector $d \in \R^V$ on vertices of an undirected graph with non-negative edge weights, the goal of transshipment is to route a flow satisfying this demand while minimizing the sum of weighted flow magnitude over all edges. Among other applications, 
setting demands of any pairs of vertices $s,t$ in the transshipment problem to $1$ and $-1$ respectively, and other vertices to zero, reduces the $s$-$t$ shortest path problem to transshipment. Combining our techniques 
with standard tools, we give a semi-streaming algorithm for transshipment improving the state-of-the-art by roughly an $\eps^{-1}$-factor in the number of passes and work.

\subsection{Our results}
\label{sec:results}

Here we state several key results of our paper. The following is our main result.

\begin{restatable}[Approximate MCM] {theorem}{restatemcmdown}\label{thm:match-rounding}
	There is a deterministic semi-streaming algorithm which given any bipartite $G = (V, E)$ with $|V|=n, |E|=m$, finds a $\eps$-multiplicatively approximate MCM in $O\!\Par{\log n \cdot \log(\eps^{-1}) \cdot \eps^{-1}}$ passes, $O({n})$ space, and $O({ m\log^2 n  \cdot \eps^{-1})}$ total work. 
\end{restatable}

The result most closely related to our Theorem~\ref{thm:match-rounding} is the semi-streaming algorithm of~\cite{AhnG18} that achieves a $\eps$-approximate MCM using $O(\eps^{-1} \cdot \log{n})$ passes and $O(n \cdot \poly{(\log{n},\eps^{-1})})$ space.
Theorem~\ref{thm:match-rounding} improves the space dependency to the optimal bound of $O(n)$ at a cost of a $\log \eps^{-1}$ factor in the number of passes. Additionally, 
our algorithm is arguably more straightforward than~\cite{AhnG18} and is deterministic (the randomized algorithm of~\cite{AhnG18} works with high probability). 

An interesting feature of our algorithm in Theorem~\ref{thm:match-rounding} its space complexity has no dependence on the parameter $\eps$. Consequently, this algorithm can obtain a \emph{very accurate} approximation of MCM in $O(n)$ space, albeit at a cost of a large pass complexity. We leverage this feature and complement the algorithm with standard augmenting path approaches for MCM, implemented efficiently using recent advances on PRAM (and semi-streaming) algorithms for directed reachability problem in~\cite{LiuJS19}. This
yields the following result for exactly computing an \emph{exact} MCM, resolving an open problem of~\cite{LiuSZ20} on obtaining $o(n)$-pass semi-streaming algorithms for exact MCM (note however that~\cite{LiuSZ20} handles weights, while we primarily focus on MCM). 

\begin{restatable}[Exact MCM]{theorem}{restateexact}\label{thm:mcm-exact}
There is a randomized semi-streaming algorithm which given any bipartite $G=(V,E)$ with $|V|=n$, finds an exact MCM with high probability in $O(n^{\frac{3}{4}+o(1)})$ passes. 
\end{restatable}

Our techniques in obtaining Theorem~\ref{thm:match-rounding} yields results for several other semi-streaming combinatorial optimization problems, such as the following.

\begin{restatable}[Optimal transport] {theorem}{restateotthm}\label{thm:OT-rounding}
There is a deterministic semi-streaming algorithm which given any optimal transport instance on a complete bipartite graph on $V = L \cup R$, costs $c \in \R^{E}_{\ge 0}$, and two sets of demands $\ell \in \Delta^L$, $r \in \Delta^R$, finds an $\eps \norm{c}_\infty$-additive approximate optimal transport plan using $O\Par{\eps^{-1} \log n\log \eps^{-1}}$ passes, $O\Par{n}$ space, and $O\Par{n^2 \eps^{-1} \log^2 n }$ work.
\end{restatable}

To our knowledge, this is the first non-trivial semi-streaming algorithm for optimal transport (however, it is plausible one can use 
	existing semi-streaming algorithms for $b$-matching such as~\cite{AhnG13} in conjunction with the reduction of~\cite{blanchet2018towards} to obtain non-trivial semi-streaming algorithms.) We further give an algorithm for  MWM using similar techniques (see \notarxiv{Appendix~\ref{sec:app}.}\arxiv{Section~\ref{ssec:randsam-ot}}).

Finally, yet another application of our techniques is to the transshipment problem. We define transshipment and prove the following result in Section~\ref{sec:tran}.

\begin{restatable}[Approximate transshipment and shortest path]{theorem}{restatetransshipthm}
\label{thm:l1main}
There is a randomized semi-streaming algorithm which given weighted undirected $G=(V,E,w)$ and  demand vector $d \in \R^n$, finds an $\eps$-multiplicatively approximate minimizer to the minimum transshipment cost in $O(\log^{O(1)}n \cdot \eps^{-1})$ passes, $O(n \log^{O(1)} n)$ space, and $O(m\log^{O(1)}n \cdot \eps^{-1})$ total work with high probability in $n$. This, in particular, yields a semi-streaming algorithm that $\eps$-multiplicative approximates the $s$-$t$ shortest path problem with the same pass, space, and work complexities.
\end{restatable}

\subsection{Our techniques} 
\label{sec:approach}

In this section, we overview our approach. We first give an overview of a natural reformulation of MCM as an appropriate $\ell_1$-regression problem. We then discuss the two main components of our semi-streaming MCM implementation, which are representative of our overall framework in other applications. The first component is a low-space solver for general $\ell_1$ regression problems in the form of box-simplex games, which returns an implicit representation of an approximate (possibly dense) solution. The second component is a rounding procedure which sparsifies implicit approximate solutions and rounds them to feasibility and  (in certain contexts) integrality. 

Our main contribution is to give efficient semi-streaming implementations and applications of each of these components. Our semi-streaming methods for optimal transport, MWM, and transshipment similarly follow from our framework, via appropriately formulating these applications as regression problems, applying our semi-streaming tools to approximate the regression problem, and rounding the approximate solution to an explicit sparse solution of no worse objective value.

\textbf{MCM as $\ell_1$-regression.} 
 The first piece of our MCM algorithm is to formulate an appropriate $\ell_1$-regression problem, whose approximate solution also implies an approximately optimal matching. Variants of this (fairly natural) reduction have found previous use, but its proof of correctness is suggestive of our overall approach, so we describe it here for completeness.  To obtain this reduction, we follow a ``penalizing overflow'' approach which has also found recent applications to approximate maximum flow \cite{Sherman13, KelnerLOS14, Peng16, Sherman17, SidfordT18} and optimal transport \cite{JambulapatiST19}. 
 
 Recall that the (fractional) MCM problem asks to find a flow $x \in \R^E_{\ge 0}$ with maximal $\ell_1$ norm, such that no vertex has more than $1$ unit of flow adjacent to it (the ``matching constraints''). In linear-algebraic terms, letting $\mb \in \{0, 1\}^{E \times V}$ be the (unsigned) incidence matrix for the graph, the problem asks to maximize $\norm{x}_1$ subject to the matching constraints, $\mb^\top x \le \1$ entrywise. The penalizing overflow approach relaxes this problem to an appropriate \emph{unconstrained} optimization problem, by appropriately penalizing matching constraint violations (rather than treating them as hard constraints), and using combinatorial structure to argue solving the unconstrained problem suffices. We now describe how to instantiate this reduction for the MCM problem.
 
 We begin by obtaining $M$, a $2$-approximation to the MCM value $M^*$, in a single pass using a greedy algorithm (Lemma~\ref{lem:greedy}), which is used to parameterize our $\ell_1$-regression problem. The next step in our reduction is a rounding procedure adapted from \cite{AltschulerWR17}, which gives an algorithm showing that any flow $\hx \in \R^E_{\ge 0}$ (not necessarily a feasible matching) can be converted into a feasible (fractional) matching $\tx$ on the same support, with size loss proportional to the total amount of constraint violations $\Overflow_d(\hx) \defeq \norm{(\mb^\top \hx- d)_+}_1$, where subscripting $+$ denotes the nonnegative part and $\mb^\top \hx - d$ encodes constraint violations (cf.\ Lemma~\ref{lem:overflowremoval}). Using this observation, we demonstrate that to find a large fractional matching it suffices to solve the following $\ell_1$-regression problem
\[\min_{x \in \Delta^{\tE}} -\inprod{\1_E}{2Mx} + 2M\norm{(\mb^\top \hx - d)_+}_1\]
to $\eps M$ additive accuracy, where the graph $\tG \defeq (\tV, \tE)$ modifies our original $G = (V, E)$ by adding one dummy edge, and the restriction of $2Mx$ to the original edges $E$ plays the role of our desired matching; these complications are to fix the $\ell_1$ norm of our $x$ variable, to put it in a form suitable for canonical optimization techniques. Further, to explicitly return an approximate MCM, we show that we can take an implicit representation of our approximate optimizer $x$, and use dynamic data structures in $O(n)$ space to return $x' \in \Delta^{\tE}$ on a sparse support satisfying $\mb^\top x = \mb^\top x'$. We can then explicitly compute an approximately optimal MCM on this sparse support.

We now discuss the technical tools which go into developing both of these pieces: our semi-streaming $\ell_1$-regression solver, and obtaining an approximately optimal matching on a sparse support.
\newpage

\textbf{Iterative methods in low space.} Our semi-streaming algorithm for approximately solving  \eqref{eq:boxsimplex}, \eqref{eq:boxsimplexprimal} follows recent improved algorithms solving these box-simplex problems in the standard (non-streaming) setting. Classical first-order methods for solving these problems, such as smoothing \cite{Nesterov05} or decoupled extragradient methods \cite{Nemirovski04, Nesterov07} have iteration counts incurring either a quadratic dependence on the error ratio $\norm{\ma}_\infty \eps^{-1}$ or growing polynomially with the dimension (each iteration running in time linear in the sparsity of the matrix $\ma$). A line of work, beginning with a breakthrough by \cite{Sherman17} (which solved a more general problem), has designed improved first-order methods bypassing both of these bottlenecks, with an iteration count of $\tO(\norm{\ma}_\infty \eps^{-1})$. The \cite{Sherman17} algorithm was analyzed in the box-simplex setting we study using a different regularizer by \cite{JambulapatiST19}, and later \cite{CohenST21} demonstrated that the classical methods of \cite{Nemirovski04, Nesterov07} could also be analyzed using the \cite{JambulapatiST19} regularizer and obtain a matching rate. 

Our contribution to this line of work is the crucial observation that the \cite{JambulapatiST19, CohenST21} algorithms can deterministically be implemented implicitly in $O(n)$ space under the semi-streaming matrix model of Definition~\ref{def:semistream_matrix}. This observation is perhaps surprising, since the simplex variable throughout the algorithm is typically fully-dense, and we do not subsample to preserve sparsity. We instead utilize the recursive structure of the box-simplex algorithms to give a low-space representation of simplex iterates, such that the actual iterate can be accessed in $O(1)$ passes. 

We begin by giving a high-level overview of the algorithm of \cite{CohenST21}; we remark that we focus on this algorithm (as opposed to \cite{Sherman17, JambulapatiST19}) because its analysis is convenient for obtaining tighter dependencies on problem parameters for our applications. The \cite{CohenST21} algorithm runs in $\tO(\norm{\ma}_\infty \eps^{-1})$ iterations, each computing a box-simplex pair $(x_t, y_t) \in \Delta^m \times [-1, 1]^n$. Each box-simplex pair is an approximate solution to a regularized linear objective of the form (hereinafter we let $|\cdot|$ applied to matrix or vector arguments act entrywise)
\begin{equation}\label{eq:regsubproblem}\min_{x \in \Delta^m, y \in [-1, 1]^n} \inprod{g\x}{x} + \inprod{g\y}{y} + r(x, y),\text{ where } r(x, y) \defeq x^\top |\ma |(y^2) + 10\norm{\ma}_\infty \sum_{i \in [m]} x_i \log x_i.\end{equation}
Typically, closed-form solutions to the above displayed problem are difficult to compute, because the regularizer $r$ couples the two blocks. However, \cite{JambulapatiST19} provided an alternating minimization subroutine which rapidly converges to an approximate solution, which suffices for our purposes.

Our first observation is that in every iteration, the simplex variable $x_t$, is always proportional to a vector $\exp(\ma v_t + |\ma|u_t + \lam_t c)$ for $n$-dimensional vectors $v_t$, $u_t$ and a scalar $\lam_t$. This is because throughout the alternating subroutine for solving \eqref{eq:regsubproblem}, every simplex iterate is the minimizer to a entropy-regularized linear objective, where the linear term is with respect to a linear combination of two pieces. The first is a gradient operator of the box-simplex problem, of the form $\Par{\ma y + c,\; b - \ma^\top x}$ for iterate $x, y$, and the second is the gradient of our regularizer (which on the $x$ side is either a linear combination of $\ma$ and $|\ma|$ applied to $n$-dimensional vectors or an entropic gradient; the latter yields a logarithm of previous iterates, which have an implicit representation inductively). Combining shows all linear terms have the form $\ma v + |\ma| u+ \lam c$, enabling our recursion.

Our second observation is that for any $(v_t, \lam_t)$ pair, computing the vectors $\ma^\top \exp(\ma v_t + |\ma|u_t+ \lam_t c)$ and $|\ma|^\top \exp(\ma v_t + |\ma|u_t+ \lam_t c)$ (required in the gradient operator) can be performed in a single semi-streaming pass, and $O(n)$ space, because the resulting matrix-vector product decouples by edge. Combining these pieces implies a deterministic semi-streaming iterative method for approximating the value of box-simplex games, as stated in the following (cf.\ Section~\ref{sec:prelims} for definitions).
\begin{restatable}{theorem}{restatelss}\label{thm:lss}
	There is a deterministic algorithm obtaining an $\eps$-additive approximation to the value of \eqref{eq:boxsimplex}, \eqref{eq:boxsimplexprimal} in $O(\frac{\norm{\ma}_\infty} \eps \cdot  \log(m) \log(\max\Par{\norm{\ma}_\infty, \norm{b}_1} \cdot \eps^{-1}))$ passes and $O(n)$ space.
	Moreover, the algorithm outputs a stream of nonnegative $1$-sparse vectors in $\R^m$ of length $O(m \cdot \frac{\norm{\ma}_\infty \log m}{\eps})$
	whose sum is $\bx$, an $\eps$-approximate optimizer to \eqref{eq:boxsimplexprimal}.
	\end{restatable}

Theorem~\ref{thm:lss} gives an $O(n)$ space, $\tO\Par{\norm{\ma}_\infty\eps^{-1}}$-pass algorithm for obtaining $\eps$-additive approximations to the value of box-simplex games. The remainder of our work is in applying this method to matching problems and providing semi-streaming implementations for rounding the solution.

\textbf{Detecting a sparse support via cycle cancelling.} So far, we have given a way of running a solver for the problems \eqref{eq:boxsimplex}, \eqref{eq:boxsimplexprimal} entirely in $O(n)$ space. However, it remains to discuss how to extract an MCM from these iterates; although we have a low-space implicit representation, the iterates themselves remain fully dense, and to write them down explicitly for postprocessing is too expensive. To get around this obstacle, we give a data structure which is compatible with our implicit representation, and can take a flow $x$ given as a sequence of edge additions and obtain an $O(n)$-sparse $x'$ such that $\norm{x}_1 = \norm{x'}_1$ and $\mb^\top x = \mb^\top x'$ (i.e.\ the ``quality of the matching'' is unaffected), but we can afford to explicitly write down $x'$. Our data structure is based on cycle cancelling via link/cut trees, a standard technique in the literature following e.g.\ \cite{SleatorT83, GoldbergT89}.

Ultimately, we show our framework of implicitly solving a box-simplex game and feeding the implicit representation of its solution into a link/cut tree data structure to detect a sparse support containing a good-quality approximate matching is quite flexible. We leverage this framework to also solve semi-streaming variants of optimal transport, weighted matching, and transshipment problems.

\subsection{Previous work}\label{sec:prev}

\smallskip
\noindent
\textbf{Maximum matching.} MCM and its many variants are arguably the most extensively studied problems in the graph streaming literature; see, e.g.~\cite{FeigenbaumKMSZ05,McGregor05,GoelKK12,KonradMM12,EggertKMS12,Kapralov13,AhnG13,KapralovKS14,AssadiKLY16,PazS17,AssadiKL17,KaleT17,Konrad18,Tirodkar18,AhnG18,AssadiBBMS19,GamlathKMS19,FarhadiHMRR20,Bernstein20,KapralovMNT20,Kapralov21} and references therein.

The first multi-pass MCM algorithm was given in~\cite{FeigenbaumKMSZ05}---alongside the introduction of the semi-streaming model itself---achieving a $(\frac23-\eps)$-approximation in $O(\eps^{-1} \log\eps^{-1})$ passes. 
The first $(1-\eps)$-approximation was then achieved by~\cite{McGregor05}, for both bipartite and non-bipartite graphs, using $(\eps^{-1})^{O(\eps^{-1})}$ passes.
The pass-complexity of this result for bipartite graphs was improved in~\cite{EggertKMS12,AhnG13,Kapralov13,AhnG18,AssadiLT21}, 
leading to the aforementioned state-of-the-art results of ~\cite{AhnG18, AssadiLT21}. More recently,~\cite{LiuSZ20} designed a semi-streaming algorithm for this problem with
$\otilde(\sqrt{m} \cdot \log{(\eps^{-1})})$ passes, which achieves better pass-dependence on $\eps$, at a cost of a polynomial dependence on $n,m$ in the number of passes. Finally,~\cite{FischerMU21} 
very recently gave the first $\poly(\eps^{-1})$-pass algorithm for this problem for non-bipartite graphs.

The aforementioned works focus on insertion-only streams. When allowing deletions to the stream, i.e.\ in the turnstile streaming model, the problem becomes significantly harder. In particular, 
the best approximation ratio possible to MCM in a single-pass is provably only $\Omega(n^{1/3})$~\cite{AssadiKLY16,DarkK20} which is achieved by the algorithms of~\cite{AssadiKLY16,ChitnisCEHMMV16} (see also~\cite{Konrad15}). Nevertheless, many multi-pass streaming algorithms for MCM can be extended to turnstile streams at the cost of an additional $O(\log{n})$ multiplicative factor in their number of passes~\cite{AhnGM12}. 
It is worth noting that the algorithm of~\cite{AhnG18}  directly works in the turnstile streaming model.

Alongside MCM, maximum weight matching (MWM)  has also been studied extensively. Some key results include $(\frac12-\eps)$-approximation in a single pass~\cite{PazS17}, 
and $(1-\eps)$-approximation in $(\eps^{-1})^{O(\eps^{-2})}$~\cite{GamlathKMS19}, $O(\eps^{-2} \log\eps^{-1})$~\cite{AhnG13}, and $O(\eps^{-1} \log n)$~\cite{AhnG18} passes. The exact
algorithm of~\cite{LiuSZ20} in~$\otilde(\sqrt{m})$ also works for the more general MWM problem.

We also note that several other streaming variants of MCM and MWM have also been studied, including random-order streams~\cite{KonradMM12,FarhadiHMRR20,Konrad18,KapralovKS14,AssadiBBMS19,Bernstein20,AssadiB21}, $n^{1+\Omega(1)}$-space algorithms \cite{AssadiBBMS19,AhnG18,BehnezhadDETY17}, or \emph{size} estimation in $o(n)$-space~\cite{KapralovKS14,EsfandiariHLMO15,AssadiKL17,McGregorV18}. Reviewing this vast literature  is beyond the scope of this paper and we refer the interested reader to these references. 

As in many computational models, MCM and MWM have been a testbed for development and adaptation of various algorithmic tools for streaming
including local-ratio algorithms~\cite{PazS17,GamlathKMS19}, augmenting paths~\cite{FeigenbaumKMSZ05,McGregor05,EggertKMS12}, water-filling or auction algorithms~\cite{Kapralov13,AssadiLT21}, 
iterative methods e.g.\ multiplicative weights~\cite{AhnG13,AhnG18} and interior-point methods~\cite{LiuSZ20}.

\smallskip
\noindent
\textbf{Optimal transport.} The design of efficient algorithms for discrete optimal transportation has received widespread attention in recent years~\cite{Cuturi13,AltschulerWR17,AltschulerBRW19,LinHJ19}, due in large part to its many applications in modern statistical modeling and machine learning. To our knowledge, ours is the first result which applies to this problem in the semi-streaming setting (though as we remark earlier, it may be possible to combine prior-work to achieve non-trivial results). Our approach for optimal transport follows straightforwardly from our algorithms for solving the fractional maximum matching problem, via known reductions in the literature.

\smallskip
\noindent
\textbf{Transshipment and shortest path.} The transshipment problem in the semi-streaming was previously studied in~\cite{BeckerKKL17} which designed an $\widetilde{O}(\eps^{-2})$ pass semi-streaming algorithm for this problem using 
gradient descent. This algorithm in turn allowed~\cite{BeckerKKL17} to obtain $\eps$-approximate semi-streaming algorithms for the (single-source) shortest path problem in $\widetilde{O}(\eps^{-2})$ passes, improving 
upon the $n^{o(1)}$-pass algorithm of~\cite{HenzingerKN16} that required $n^{1+o(1)}$-space (more than the restriction of semi-streaming algorithms). More recently,~\cite{ChangFHT20} gave a semi-streaming algorithm for finding exact single-source 
shortest paths in $\widetilde{O}(\sqrt n)$ passes. On the lower bound front,~\cite{GuruswamiO13} showed that $\eps$-approximate semi-streaming algorithms for the $s$-$t$ shortest path problem require $\Omega(\eps^{-1})$ passes (see also~\cite{AssadiR20,KPS0Y21} for more  advances on this front).

\section{Preliminaries}
\label{sec:prelims}

\medskip
\noindent
\textbf{General notation.} We use $[n]$ to denote $\{1, ... , n\}$ and  $\0_n$ and $\1_n$ to denote the all-zeros and all-ones vectors in $\R^n$. An event holds with high probability if it holds with probability $\ge 1 - n^{-c}$ for $c > 0$ (in our results, $c$ can be chosen to be arbitrarily large affecting guarantees by constant factors). We use $\nnz(\ma)$ to denote the number of nonzero entries in a matrix $\ma$, and its row $i$ and column $j$ are denoted $\ai$ and $\aj$. We say scalar $\alpha > 0$ is an $\eps$-multiplicative approximation to scalar $\beta > 0$ if $|\tfrac{\alpha}{\beta} - 1| \le \eps$; similarly, we say it is an $\eps$-additive approximation if $|\alpha - \beta| \le \eps$. For a vector $v$, we use $v_+$ or $\Par{v}_+$ to denote the vector which entrywise has $\Brack{v_+}_i = \max\Par{v_i, 0}$, and use $v_S\in\R^{n}$ to denote the vector which zeroes out entries of $v$ in $[n] \setminus S$ for $S \subseteq [n]$.

\medskip
\noindent
\textbf{Norms.} We let $\norm{v}_p$ denote the $\ell_p$ norm of vector $v$, and $\norm{\mm}_p$ denote the $\ell_p$ operator norm of a matrix. In particular, $\norm{\mm}_\infty$ is the maximum $\ell_1$ norm over rows of $\mm$, and $\norm{\mm}_1$ is the maximum $\ell_1$ norm over columns. We let $\Delta^m \subset \R^m_{\ge 0}$ denote the nonnegative simplex, so $x \in \Delta^m \iff \norm{x}_1 = 1$ and entrywise $x\ge 0$. Finally, we let $\norm{v}_0$ be the number of nonzero entries of a vector $v$.

\medskip
\noindent
\textbf{Graphs.} We denote a graph by $G = (V, E)$, and define $n = |V|$ and $m = |E|$  when context is clear. We refer to the independent vertex subsets of a bipartite graph by $L$ and $R$. We denote the (unsigned edge-vertex) incidence matrix of a graph by $\mb \in \R_{\ge 0}^{E \times V}$, where $\mb_{ev} = 1$, if and only if $v \in V$ is an endpoint of $e$ and is zero otherwise. We refer to any assignment of values to edges $f \in \R^E$ as a \emph{flow}. When the graph is weighted (with associated edge weights $w \in \R_{\ge 0}^E$), we refer to the graph by $G = (V, E, w)$. 

\medskip
\noindent
\textbf{Computation model.} Throughout, we operate in the standard word RAM model of computation, where basic arithmetic operations on $O(\log n)$-bit words can be performed in constant time. For weighted graphs, we assume all weights can be stored in $O(1)$ words, and all results concerning additional memory overhead count the number of additional words necessary for algorithms. We also assume always that $\eps^{-1} = O(\text{poly}(n))$. We measure space overhead complexity by the number of words used. In other related computational models, this may increase our space or work complexities by a logarithmic factor. 	%

\section{Box-simplex games in low space}
\label{sec:value}

In this section, we provide semi-streaming algorithms (cf.\ Definition~\ref{def:semistream_matrix}) for approximately solving box-simplex bilinear games of the form \eqref{eq:boxsimplex}, which will yield approximate minimizers for the $\ell_1$ regression problem. 
 The complexity of our algorithms depend on the quantity $\norm{\ma}_\infty$. We prove in this section that we can implement recent improved algorithms for problems of the form \eqref{eq:boxsimplex} based on \emph{area-convex regularization} \cite{Sherman17, JambulapatiST19, CohenST21} in the semi-streaming model in low memory, and give the guarantees as Theorem~\ref{thm:lss}.

Our algorithm for Theorem~\ref{thm:lss} is a low-space implementation of a primal-dual algorithm inspired by \cite{Sherman17}, specialized to box-simplex games (as analyzed in \cite{JambulapatiST19, CohenST21}). Specifically, the algorithms and analyses of \cite{Sherman17, JambulapatiST19, CohenST21} are examples of \emph{extragradient methods}, and are based on creating custom analyses for the classical algorithms of \cite{Nemirovski04, Nesterov07} tailored to the specific structure of box-simplex games. Our main algorithm, Algorithm~\ref{alg:sherman}, follows the framework of \cite{CohenST21}, who analyzed the convergence of a mirror prox scheme for solving \eqref{eq:boxsimplex}. Lines 5 and 6 of the algorithm are implemented with $\AltMin$ (Algorithm~\ref{alg:altmin}), an alternating minimization subroutine. Finally, we note that we peform pre-processing and post-processing of the problem in Lines 1 and 2, so that our final guarantees only depend on properties of $\ma$ (and not $b$ or $c$). The result of these steps only affect the problem by a constant scalar shift, and can be implemented in one semi-streaming pass (a more detailed statement and discussion is given as Lemma~\ref{lem:truncateb} in Appendix~\ref{app:deferredsherman}).
 
\notarxiv{
\begin{algorithm}[ht!]
	\caption{$\Sherman(\ma, b, c, \epsilon, T)$}
	\begin{algorithmic}[1]\label{alg:sherman}
		\STATE \textbf{Input:} $\ma \in \R^{m \times n}_{\ge 0}$, $c \in \R^m$, $b \in \R^n$, $0 \le \eps \le \norm{\ma}_\infty$, $T \in \N$
		\STATE $c \gets \max((\cmax- 2\norm{\ma}_\infty)\1, c) - (\cmax - 2\norm{\ma}_\infty)\1$ entrywise, where $\cmax \defeq \max_{i \in [m]} c_i$.
		\COMMENT{\emph{This ensures $\norm{c}_\infty \le 2\norm{\ma}_\infty$ as required by Lemma 5 of \cite{CohenST21} for general box-simplex games, so $c$ to lies in $[0, 2\norm{\ma}_\infty]$. It is implementable by storing $\norm{c}_\infty$ for semi-streaming applications. For all our combinatorial applications this step is unnecessary.}}
		\STATE $r(x, y) \defeq x^\top \ma (y^2) + 10\norm{\ma}_\infty \sum_{i \in [m]} x_i \log x_i$, $g(x, y) \defeq (\ma y + c, -\ma^\top x + b)$
		\STATE $z_0 \gets (\tfrac{1}{m}\1_m, \mzero_n)$ 
		\FOR{$0 \le t < T$}
		\STATE $w_t \defeq (x'_t, y'_t) \gets \tfrac{\eps}{2}$-approx.\ minimizer to $\inprod{\tfrac 1 3 g(z_t) - \nabla r(z_t)}{w} + r(w)$ in $\Delta^m \times [-1, 1]^n$ calling $\AltMin(\gamma\x, \gamma\y, \ma, K, x_t, y_t)$ with $K = O(\log\tfrac{\max(\norm{\ma}_\infty, \norm{c}_\infty, \norm{b}_1)}{\eps})$,
		$\gamma\x = \tfrac 1 3 (\ma y_t + c) - 10\norm{\ma}_\infty\log x_t - \ma (y_t^2)$,	$\gamma\y = \frac 1 3 \Par{b - \ma^\top x_t} - 2\diag{y_t}\ma^\top x_t.$
		\STATE $z_{t + 1} \defeq (x_{t + 1}, y_{t + 1}) \gets \tfrac{\eps}{2}$-approx.\ minimizer to $\inprod{\tfrac 1 3 g(w_t) - \nabla r(z_t)}{z} + r(z)$ in $\Delta^m \times [-1, 1]^n$ calling $\AltMin(\gamma\x, \gamma\y, \ma, K, x_t, y_t)$ with $K = O(\log\tfrac{\max(\norm{\ma}_\infty, \norm{c}_\infty, \norm{b}_1)}{\eps})$,
		$\gamma\x = \tfrac 1 3 (\ma y'_t + c) - 10\norm{\ma}_\infty\log x_t - \ma (y_t^2)$, $\gamma\y = \frac 1 3 \Par{b - \ma^\top x'_t} - 2\diag{y_t}\ma^\top x_t.$
		\ENDFOR
	\end{algorithmic}
\end{algorithm}
}

\arxiv{
\begin{algorithm}
	
	\DontPrintSemicolon
	\KwInput{$\ma \in \R^{m \times n}_{\ge 0}$, $c \in \R^m$, $b \in \R^n$, $0 \le \eps \le \norm{\ma}_\infty$, $T \in \N$}
	Remove all coordinates $i \in [m]$ where $c_i > \min_{i^* \in [m]} c_{i^*} + 2\norm{\ma}_\infty$ from $c$ and corresponding rows of $\ma$ from the problem (ignoring the entries from the stream), then set $c \gets c - (\min_{i^* \in [m]} c_{i^*})\1$ entrywise
	\label{alg:sherman-pre-1}\;
	$b\gets  \min\left(\max\left(b,-\|\ma\|_\infty\right),\|\ma\|_\infty\right)$\label{alg:sherman-pre-2}\;
	$r(x, y) \defeq x^\top |\ma| (y^2) + 10\norm{\ma}_\infty \sum_{i \in [m]} x_i \log x_i$, $g(x, y) \defeq (\ma y + c, -\ma^\top x + b)$\;
	$z_0 \gets (\tfrac{1}{m}\1_m, \mzero_n)$\;
	\For{$0 \le t < T$}
	{
		
		$\gamma\x \gets \tfrac 1 3 (\ma y_t + c) - 10\norm{\ma}_\infty\log x_t - |\ma| (y_t^2)$\;
		$\gamma\y \gets \frac 1 3 \Par{b - \ma^\top x_t} - 2\diag{y_t}|\ma|^\top x_t$\; 
		$w_t \defeq (x'_t, y'_t) \gets \AltMin(\gamma\x, \gamma\y, \ma, K, x_t, y_t)$ with $K = O(\log\tfrac{\max(\norm{\ma}_\infty, \norm{b}_1)}{\eps})$ giving $\tfrac{\eps}{2}$-approx.\ minimizer to $\inprod{\tfrac 1 3 g(z_t) - \nabla r(z_t)}{w} + r(w)$ in $\Delta^m \times [-1, 1]^n$ \label{alg:sherman-sub-1}\;
			$\gamma\x \gets \tfrac 1 3 (\ma y'_t + c) - 10\norm{\ma}_\infty\log x_t - |\ma| (y_t^2)$\;
			$\gamma\y \gets \frac 1 3 \Par{b - \ma^\top x'_t} - 2\diag{y_t}|\ma|^\top x_t$\; 
		$z_{t + 1} \defeq (x_{t + 1}, y_{t + 1}) \gets \AltMin(\gamma\x, \gamma\y, \ma, K, x_t, y_t)$ with $K = O(\log\tfrac{\max(\norm{\ma}_\infty, \norm{b}_1)}{\eps})$ giving $\tfrac{\eps}{2}$-approx.\ minimizer to $\inprod{\tfrac 1 3 g(w_t) - \nabla r(z_t)}{z} + r(z)$ in $\Delta^m \times [-1, 1]^n$ \label{alg:sherman-sub-2}\;
	}
	\caption{$\Sherman(\ma, b, c, \epsilon, T)$}\label{alg:sherman}
\end{algorithm}
}

\notarxiv{
\begin{algorithm}[ht!]
	\caption{$\AltMin(\gamma\x, \gamma\y, \ma, \delta, K, x_{\text{init}}, y_{\text{init}})$}
	\begin{algorithmic}[1]\label{alg:altmin}
		\STATE \textbf{Input:} $\ma \in \R^{m \times n}_{\ge 0}$, $\gamma\x \in \R^m$, $\gamma\y \in \R^n$, $K \in \N$
		\STATE \textbf{Output:} Approx.\ minimizer to $\inprod{(\gamma\x, \gamma\y)}{z} + r(z)$ for $r(z)$ in Line 2, Algorithm~\ref{alg:sherman}
		\STATE $x^{(0)} \gets x_{\text{init}}$, $y^{(0)} \gets y_{\text{init}}$
		\FOR{$0 \le k < K$}
		\STATE $x^{(k + 1)} \gets \argmin_{x \in \Delta^m}\Brace{\inprod{\gamma\x}{x} + r(x, y^{(k)})}$
		\STATE $y^{(k + 1)} \gets \argmin_{y \in [-1, 1]^n}\Brace{\inprod{\gamma\y}{y} + r(x^{(k + 1)}, y)}$
		\ENDFOR
		\RETURN $(x^{(K)}, y^{(K)})$
	\end{algorithmic}
\end{algorithm}}
\arxiv{
\begin{algorithm}[ht!]
		\KwInput{$\ma \in \R^{m \times n}_{\ge 0}$, $\gamma\x \in \R^m$, $\gamma\y \in \R^n$, $K \in \N$}
		\KwOutput{ Approx.\ minimizer to $\inprod{(\gamma\x, \gamma\y)}{z} + r(z)$ for $r(z)$ in Line~\ref{alg:sherman-sub-1} and~\ref{alg:sherman-sub-2}, Algorithm~\ref{alg:sherman}}
		 $x^{(0)} \gets x_{\text{init}}$, $y^{(0)} \gets y_{\text{init}}$\;
		\For{$0 \le k < K$}
		{
	 $x^{(k + 1)} \gets \argmin_{x \in \Delta^m}\Brace{\inprod{\gamma\x}{x} + r(x, y^{(k)})}$\;
		$y^{(k + 1)} \gets \argmin_{y \in [-1, 1]^n}\Brace{\inprod{\gamma\y}{y} + r(x^{(k + 1)}, y)}$\;
		}
		\Return $(x^{(K)}, y^{(K)})$
		\caption{$\AltMin(\gamma\x, \gamma\y, \ma, \delta, K, x_{\text{init}}, y_{\text{init}})$}\label{alg:altmin}
\end{algorithm}
}

The guarantees of Algorithms~\ref{alg:sherman} and~\ref{alg:altmin} were respectively demonstrated by \cite{CohenST21} and \cite{JambulapatiST19}, and summarized in Proposition~\ref{prop:shermancorrect}. The main technical modification in it is a tighter characterization of the alternating minimization subroutine, based on bounding the initial error independently of the number of iterations (whereas \cite{JambulapatiST19} assumed a worst-case bound on initial error). 

\begin{restatable}{proposition}{restateshermancorrect}\label{prop:shermancorrect}
By taking $T = O(\frac{\norm{\ma}_\infty \log m}{\eps})$ for a sufficiently large constant, Algorithm~\ref{alg:sherman} results in iterates $\{(x'_t, y'_t)\}_{0 \le t < T}$ so that $\bar{x} \defeq \tfrac{1}{T}\sum_{0 \le t < T} x'_t$ is an $\eps$-approximate minimizer to \eqref{eq:boxsimplexprimal}.
\end{restatable}

We defer a proof of Proposition~\ref{prop:shermancorrect} to Appendix~\ref{app:deferredsherman}, as it largely stems from ideas used in prior work. We next turn to our main technical contribution, the development of semi-streaming algorithms for implementing Algorithms~\ref{alg:sherman} and~\ref{alg:altmin} (Lemma~\ref{lem:implsimplex}), the primary innovation of this section. Concretely, our implementation has the following key property: in every iteration all box variables are maintained explicitly, and all simplex variables are maintained implicitly as points proportional to $\exp(\ma v + |\ma| u + \lam c)$, for some vectors $v,u \in \R^n$ and scalar $\lam$ computed and stored in $O(n)$ space. We first make a crucial technical observation regarding computing gradients with a single pass.

\begin{lemma}\label{lem:implicitx}
	For $v, u \in \R^n$ and $\lam \in \R$ there is a semi-streaming algorithm computing $\norm{\exp(\ma v + |\ma|u + \lam c)}_1$, $\ma^\top \exp(\ma v + |\ma|u + \lam c)$, $|\ma|^\top \exp(\ma v + |\ma|u + \lam c)$ in one pass, $O(n)$ %
	space, and $\nnz(\ma)$ work.
\end{lemma}
\begin{proof}
	First, to compute $\norm{\exp(\ma v + |\ma|u +\lam c)}_1$, for each $i \in [m]$ the algorithm sequentially computes each $\exp(\inprod{\ma_{i:}}{v} + \inprod{|\ma_{i:}|}{u} + \lam c_i)$ through the pass and sums these values. Second, note
\begin{align*}
 	\ma^\top \exp(\ma v + |\ma|u + \lam c) & = \sum_{i \in [m]} \ma_{i:} \exp(\inprod{\ma_{i:}}{v} + \inprod{|\ma_{i:}|}{u} + \lam c_i),\\
 	|\ma|^\top \exp(\ma v + |\ma|u + \lam c) & = \sum_{i \in [m]} |\ma_{i:}| \exp(\inprod{\ma_{i:}}{v} + \inprod{|\ma_{i:}|}{u} + \lam c_i).
\end{align*}
Throughout the pass, the algorithm adds $\exp(\inprod{\ma_{i:}}{v} + \inprod{|\ma_{i:}|}{u} + \lam c_i)$ to the coordinates of the resulting vector in $\ma$ and $|\ma|$ it contributes to (specified by the row given in the stream).
\end{proof}

Using ideas from Lemma~\ref{lem:implicitx}, we demonstrate that the entirety of Algorithm~\ref{alg:sherman} and its subroutine Algorithm~\ref{alg:altmin} can be implemented in low space. \notarxiv{We defer the proof to Appendix~\ref{app:helpersherman}.}

\begin{restatable}{lemma}{restateimplicitx}\label{lem:implsimplex}
Consider an instance of Algorithm~\ref{alg:sherman} using Algorithm~\ref{alg:altmin} to implement Lines 5 and 6. Throughout, the following invariants on the iterates hold:
\[x_t \propto \exp(\ma v_t +  |\ma|u_t + \lam_t c),\; x'_t \propto \exp(\ma v'_t + |\ma|u'_t+  \lam'_t c) \text{ for some } v_t,  v'_t, u_t,u'_t \in \R^n,\; \lam_t, \lam'_t \in \R.\]
Given their values in the previous iteration, in every iteration of Lines $6$-$11$ of Algorithm~\ref{alg:sherman} we can compute such $(v_t, v'_t, u_t, u'_t, \lam_t, \lam'_t, y_t, y'_t)$ in $O(n)$ space. Each iteration takes $O(K)$ passes and $O(\nnz(\ma) \cdot K)$ total work, where $K$ is the input parameter to Algorithm~\ref{alg:altmin}.
\end{restatable}

\arxiv{
\begin{proof}
	We proceed by induction on $t$; in the first iteration, we have $v_0 = u_0 = \mzero_n$, $\lam_0 = 0$.
	
	\textbf{Preserving the $w_t$ invariant.} Suppose inductively that $z_t = (x_t, y_t)$ where $x_t \propto \exp(\ma v_t + |\ma| u_t + \lam_t c)$ for explicitly stored values $v_t$, $u_t$, $\lam_t$, $y_t$; we will drop the index $t$ for simplicity and refer to these as $\bv$, $\bu$, $\blam$, $\by$. Consider the procedure Algorithm~\ref{alg:altmin} initalized with these values, and note that
	\begin{align*}\gamma\x &= \frac{1}{3}\Par{\ma \by + c} - 10\norm{\ma}_\infty\Par{\ma \bv + |\ma|\bu +  \blam c} - |\ma| \Par{\by^2},\\
	\gamma\y &=\frac{1}{3}\Par{b - \ma^\top \frac{\exp(\ma \bv + |\ma|\bu +  \blam c)}{\norm{\exp(\ma \bv + |\ma|\bu +  \blam c)}_1}} - 2\diag{\by} |\ma|^\top \frac{\exp(\ma \bv + |\ma|\bu +  \blam c)}{\norm{\exp(\ma \bv + |\ma|\bu +  \blam c)}_1}.\end{align*}
Since $r(x, y)  = \inprod{|\ma| (y^2)}{x} + 10\norm{\ma}_\infty\sum_{i \in [m]} x_i\log x_i$, we can compute that for each $0 \le k < K$,
	\begin{equation}\label{eq:altminstep}
	\begin{aligned}
	& x^{(k + 1)} = \argmin_{x \in \Delta^m}\Brace{\inprod{|\ma| \Par{(y^{(k)})^2} + \gamma\x}{x} + 10\norm{\ma}_\infty \sum_{i \in [m]} x_i \log x_i}\\
	\propto & \exp\Par{-\tfrac{1}{10\norm{\ma}_\infty}\Par{\ma \Par{\tfrac 1 3 \by - 10\norm{\ma}_\infty \bv} + |\ma|\Par{ - (\by)^2-10\norm{\ma}_\infty\bu+ (y^{(k)})^2}+ \Par{\tfrac 1 3 - 10\norm{\ma}_\infty\blam} c}},\\
	& y^{(k + 1)} = \argmin_{y \in [-1, 1]^n}\Brace{\inprod{\gamma\y}{y} + \inprod{|\ma|^\top x^{(k + 1)}}{y^2}} \\
	= & \text{median}\Par{-1, 1, -\frac{\gamma\y}{2|\ma|^\top x^{(k + 1)}}} \text{ entrywise.}
	\end{aligned}
	\end{equation}

	In the last line, the median operation truncates the vector $-\tfrac{\gamma\y}{2\ma^\top x^{(k + 1)}}$ coordinatewise on the box $[-1, 1]^n$. Now, suppose at the start of iteration $k$ of Algorithm~\ref{alg:altmin}, we have the invariant 
	\begin{equation}\label{eq:altmininvar}
	x^{(k)} \propto \exp\Par{\ma v^{(k)} + |\ma |u^{(k)} + \lam^{(k)} c},
	\end{equation}
	and we have explicitly stored the tuple $(v^{(k)}, u^{(k)}, \lam^{(k)}, y^{(k)})$; in the first iteration, we can clearly choose $(v^{(0)}, u^{(0)}, \lam^{(0)}, y^{(0)}) = (\bv,\bu, \blam, \by)$. By the derivation \eqref{eq:altminstep}, we can update
	\begin{align*}
		v^{(k + 1)} & \gets -\frac{1}{10\norm{\ma}_\infty} \Par{\frac 1 3 \by - 10\norm{\ma}_\infty \bv},\\
		u^{(k+1)} & \gets -\frac{1}{10\norm{\ma}_\infty} \Par{- (\by)^2+ (y^{(k)})^2- 10\norm{\ma}_\infty \bu},\\
		\lam^{(k + 1)} & \gets -\frac{1}{10\norm{\ma}_\infty}\Par{\frac 1 3 - 10\norm{\ma}_\infty\blam},
	\end{align*}
	preserving representation \eqref{eq:altmininvar} in the next iteration. Moreover, note that $\gamma\y$ can be explicitly stored (it can be computed in one pass using Lemma~\ref{lem:implicitx}), and by applying Lemma~\ref{lem:implicitx} and the form \eqref{eq:altmininvar}, we can compute the vector $|\ma|^\top x^{(k + 1)}$ explicitly in one pass over the data and $O(m)$ work. This allows us to explicitly store $y^{(k + 1)}$ as well. The final point $w_t$ is one of the iterates $(x^{(K)}, y^{(K)})$, proving the desired invariant. Finally, we remark in order to perform these computations we only need to store the tuple $(v^{(k)}, u^{(k)}, \lam^{(k)}, y^{(k)})$ from the prior iteration, in $O(n)$ memory.
	
	\textbf{Preserving the $z_{t + 1}$ invariant.} The argument for preserving the invariant on $z_{t + 1}$ is exactly analogous to the above argument regarding $w_t$; the only modification is that the input vector to Algorithm~\ref{alg:altmin} is 
	\begin{align*}
	\gamma\x &= \frac{1}{3}\Par{\ma y'_t + c} - 10\norm{\ma}_\infty\Par{\ma \bv +|\ma|\bu +  \blam c} - |\ma| \Par{\by^2},\\
	\gamma\y &=\frac{1}{3}\Par{b-\ma^\top x'_t} - 2\diag{\by} |\ma|^\top \frac{\exp(\ma \bv + |\ma|\bu +  \blam c)}{\norm{\exp(\ma \bv + |\ma|\bu +  \blam c)}_1}.
	\end{align*}
	However, the previous argument shows that the vector $y'_t$ can be explicitly computed, and the vector $x'_t$ satisfies the invariant in the lemma statement. The same inductive argument shows that we can represent every intermediate iterate in the computation of $z_{t + 1}$ in the desired form.
	
	\textbf{Numerical stability.} We make a brief comment regarding numerical stability in the semi-streaming model, which may occur due to exponentiating vectors with a large range $\omega(\log m)$ (in defining simplex variables). It was shown in \cite{JambulapatiST19} that Algorithm~\ref{alg:sherman} is stable to increasing the value of any coordinate of a simplex variable which is $m^{10}$ multiplicatively smaller than the largest to reach this threshold, and renormalizing (i.e.\ this implicit ``padding'' operation only affects the resulting minimizer by an inverse-polynomial amount, which we can assume $\eps$ is larger than since otherwise the number of passes is super-polynomial). In computations of Lemma~\ref{lem:implicitx}, we can first store the largest coordinate of $\ma v + |\ma|u + \lam c$ in one pass. Then, for every coordinate more than $10\log m$ smaller than the largest coordinate, we will instead treat it as if its value was $10\log m$ smaller than the maximum in all computations, requiring one extra pass over the data.
\end{proof}
}

\arxiv{
We explicitly state a complete low-space implementation of Algorithms~\ref{alg:sherman} and~\ref{alg:altmin} under our semi-streaming implementation here for completeness as Algorithm~\ref{alg:lowspacesherman}. 
}

\begin{algorithm}[ht!]
	\DontPrintSemicolon
	\KwInput{ $\ma \in \R^{m \times n}_{\ge 0}$, $c \in \R^m$, $b \in \R^n$, $0 \le \eps \le \norm{\ma}_\infty$}
	\KwOutput{ $\{v'_t\}_{0 \le t < T} \subset \R^n$, $\{u'_t\}_{0 \le t < T} \subset \R^n$, $\{\lam'_t\}_{0 \le t < T} \subset \R$, $\{y'_t\}_{0 \le t < T} \subset \R^n$ so that for
		\[\by \defeq \frac{1}{T}\sum_{0 \le t < T} y'_t,\; \bx \defeq \frac{1}{T}\sum_{0 \le t < T} \frac{\exp(\ma v'_t + |\ma|u'_t +  \lam'_t c)}{\norm{\exp(\ma v'_t + |\ma|u'_t + \lam'_t c)}_1},\]
		the pair $(\bx, \by)$ is an $\eps$-approximate saddle point to \eqref{eq:boxsimplex}}
	$T \gets O(\tfrac{\norm{\ma}_\infty \log m}{\eps})$, $K \gets O(\log \tfrac{n\norm{\ma}_\infty}{\eps})$\;
	$t \gets 0$, $\lam_0 \gets 0$, $v_0 \gets \mzero_n$, $u_0 \gets \mzero_n$, $y_0 \gets \mzero_n$\;
	\While{$t < T$}{
		$(v^{(0)}, u^{(0)}, \lam^{(0)}, y^{(0)}) \gets (v_t, u_t, \lam_t, y_t)$\;
		$\gamma\y \gets \tfrac 1 3 (b - \ma^\top \tfrac{\exp(\ma v_t + |\ma| u_t + \lam_t c)}{\norm{\exp(\ma v_t +  |\ma| u_t + \lam_t c)}_1}) - 2\diag{y_t} |\ma|^\top \tfrac{\exp(\ma v_t + |\ma| u_t + \lam_t c)}{\norm{\exp(\ma v_t + |\ma| u_t + \lam_t c)}_1}$, computed using Lemma~\ref{lem:implicitx}\;
		\For{$0 \le k < K$}{
			$v^{(k + 1)} \gets -\tfrac{1}{10\norm{\ma}_\infty} (\tfrac 1 3 y_t  - 10\norm{\ma}_\infty v^{(0)})$\;
			$u^{(k+1)} \gets -\tfrac{1}{10\norm{\ma}_\infty} (- (y^{(0)})^2 + (y^{(k)})^2-10\norm{\ma}_\infty u^{(0)})$\;
			$\lam^{(k + 1)} \gets -\tfrac{1}{10\norm{\ma}_\infty} (\tfrac 1 3 - 10\norm{\ma}_\infty \lam^{(0)})$\;
			$d^{(k + 1)} \gets 2|\ma|^\top \tfrac{\exp(\ma v^{(k + 1)}+ |\ma| u^{(k + 1)} + \lam^{(k + 1)} c)}{\norm{\exp(\ma v^{(k + 1)}+ |\ma| u^{(k + 1)} +\lam^{(k + 1)} c)}_1}$, computed using Lemma~\ref{lem:implicitx}\;
			$y^{(k + 1)} \gets \textup{med}(-1, 1, -\tfrac{\gamma\y}{d^{(k + 1)}})$ entrywise\; }
		$(v'_t,u'_t, \lam'_t, y'_t) \gets (v^{(K)}, u^{(K)}, \lam^{(K)}, y^{(K)})$\;
		$(v^{(0)}, u^{(0)}, \lam^{(0)}, y^{(0)}) \gets (v_t, u_t, \lam_t, y_t)$\;
		$\gamma\y \gets \tfrac 1 3 (b - \ma^\top \tfrac{\exp(\ma v'_t + |\ma| u'_t + \lam'_t c)}{\norm{\exp(\ma v'_t + |\ma| u'_t + \lam'_t c)}_1}) - 2\diag{y_t} |\ma|^\top \tfrac{\exp(\ma v_t + |\ma| u_t + \lam_t c)}{\norm{\exp(\ma v_t + |\ma| u_t + \lam_t c)}_1}$, computed using Lemma~\ref{lem:implicitx}\;
		\For{$0 \le k < K$}{
			$v^{(k + 1)} \gets -\tfrac{1}{10\norm{\ma}_\infty} (\tfrac 1 3 y'_t - 10\norm{\ma}_\infty v^{(0)})$\;
			$u^{(k+1)} \gets -\tfrac{1}{10\norm{\ma}_\infty} (- (y^{(0)})^2 + (y^{(k)})^2-10\norm{\ma}_\infty u^{(0)})$\;
			$\lam^{(k + 1)} \gets -\tfrac{1}{10\norm{\ma}_\infty} (\tfrac 1 3 - 10\norm{\ma}_\infty \lam^{(0)})$\;
			$d^{(k + 1)} \gets 2|\ma|^\top \tfrac{\exp(\ma v^{(k + 1)} + |\ma| u^{(k + 1)}+ \lam^{(k + 1)} c)}{\norm{\exp(\ma v^{(k + 1)} + |\ma| u^{(k + 1)}+ \lam^{(k + 1)} c)}_1}$, computed using Lemma~\ref{lem:implicitx}\;
			$y^{(k + 1)} \gets \textup{med}(-1, 1, -\tfrac{\gamma\y}{d^{(k + 1)}})$ entrywise \;
		}		
		$(v_{t + 1}, u_{t+1}, \lam_{t + 1}, y_{t + 1}) \gets (v^{(K)}, u^{(K)}, \lam^{(K)}, y^{(K)})$\;
		$t \gets t + 1$\;
	}
	\caption{$\LSSherman(\ma, b, c, \epsilon)$}\label{alg:lowspacesherman}
\end{algorithm}

\begin{corollary}\label{corr:lss}
Algorithm~\ref{alg:lowspacesherman} is an implementation of Algorithms~\ref{alg:sherman} and~\ref{alg:altmin} with parameters given by Proposition~\ref{prop:shermancorrect} in $O(n)$ space and $O(T \cdot \log \frac{\max(\norm{\ma}_\infty, \norm{b}_1)}{\eps})$ passes.
\end{corollary}
\begin{proof}
It is immediate from Lemma~\ref{lem:implsimplex} that each loop of Lines 6-10 and Lines 14-18 in Algorithm~\ref{alg:lowspacesherman} is an implementation of Algorithm~\ref{alg:altmin} for computing each iterate $w_t$ and $z_{t + 1}$ in iteration $t$ of Algorithm~\ref{alg:sherman}, in $O(n)$ space. The pass complexity follows since passes are only used $O(1)$ times in each run of Lines 6-11 and 15-20 of Algorithm~\ref{alg:lowspacesherman}.
\end{proof}

Finally, we conclude with our proof of Theorem~\ref{thm:lss}.

\begin{proof}[Proof of Theorem~\ref{thm:lss}]
We use the parameter settings in Proposition~\ref{prop:shermancorrect}, which implies that it suffices to compute the value of \eqref{eq:boxsimplexprimal} with $\bx \defeq \frac{1}{T}\sum_{0 \le t < T} x'_t$. By linearity, we have
\[c^\top \bx = \frac{1}{T}\sum_{0 \le t < T} \inprod{c}{x'_t} = \frac{1}{T} \sum_{0 \le t < T} \sum_{i \in [m]} c_i \exp\Par{\inprod{\ai}{v'_t} + \inprod{|\ai|}{u'_t} + \lam'_t c_i}.\]
Since the summands in the above display (for each $0 \le t < T$) separate componentwise in the stream, we can compute each $\sum_{i \in [m]} c_i \exp(\inprod{\ai}{v'_t} + \inprod{|\ai|}{u'_t} + \lam'_t c_i)$ in one pass per iteration and store their average. Moreover, we can similarly compute the vector $\ma^\top \bx$ using Lemma~\ref{lem:implicitx} in one pass per iteration and store it explicitly in $O(n)$ space, at which point we can compute $\norm{\ma^\top \bx - b}_1$. Combining these two parts, we have the desired approximation to the value.

Finally, to demonstrate the guarantee on streaming $\bx$, the approximate minimizer has the form
\[\bx = \frac{1}{T}\sum_{0 \le t < T} \frac{\exp\Par{\ma v'_t + |\ma| u'_t + \lam'_t c}}{\norm{\exp\Par{\ma v'_t + |\ma| u'_t + \lam'_t c}}_1}.\]
We can thus take one additional pass per iteration to compute the value of
\[\Brack{\frac{1}{T} \frac{\exp\Par{\ma v'_t + |\ma| u'_t + \lam'_t c}}{\norm{\exp\Par{\ma v'_t + |\ma| u'_t + \lam'_t c}}_1}}_i \text{ for all } i \in [m],\]
and produce a stream of these values without affecting the overall pass complexity.
\end{proof}

We note that $\max(\norm{\ma}_\infty, \norm{b}_1)$ is never larger than $n\norm{\ma}_\infty$ by Line 2 of Algorithm~\ref{alg:sherman}. Finally, on a query $S \subseteq E$ with $|S| = q$, using $O(n + q)$ space (and no more passes), for $(\bx, \by)$ an $\eps$-approximate saddle point to \eqref{eq:boxsimplex} computed by our algorithm, we can output the set of values $(\{\bx_i\}_{i \in S}, y)$. This is accomplished by calling Lemma~\ref{lem:implicitx} and explicitly storing values of $x'_t$ on coordinates in $S$ each iteration, which is implementable in $O(q)$ additional space. While we do not use this fact for finding an approximate MCM solution or for our applications, this gives the memory requirement for querying entries of an approximate solution in the most general setting for box-simplex games.

\section{Approximate maximum cardinality matching}
\label{sec:flowround}

In this section, we give specific treatment to the problem of approximate MCM and prove our main result Theorem~\ref{thm:match-rounding}. We prove Theorem~\ref{thm:match-rounding} by assembling a variety of tools, centered around casting MCM as an instance of~\eqref{eq:boxsimplex} and using several helper procedures to complete the result.  \notarxiv{A variety of proofs in this section are deferred to Appendix~\ref{app:helpermcm}.}

\begin{enumerate}
	\item In Section~\ref{ssec:formulation}, we give the specific box-simplex problem formulation which serves as the optimization workhorse of this section, and prove that an appropriate approximate solution to this problem results in an approximate MCM.
	\item In Section~\ref{ssec:tools}, we give pre- and post-processing tools for manipulating the box-simplex problem and its output. Specifically, we give a vertex-size reduction enabling a tighter runtime analysis, and a cycle cancelling procedure for sparsifying the support of the approximate MCM.
	\item Finally, we put these pieces together to prove Theorem~\ref{thm:match-rounding} in Section~\ref{ssec:fullproof}.
\end{enumerate}

\subsection{Reducing MCM to a box-simplex problem}\label{ssec:formulation}

Throughout this section, we consider the problem of finding a $\eps$-approximate MCM of a bipartite graph $G = (V, E)$ with $|V|=n$, $|E|=m$, with unsigned incidence matrix $\mb \in \{0, 1\}^{E \times V}$, and maximum matching size $M^*$. Any maximum (possibly fractional) matching solves the problem 
\[\max_{M^* > 0} M^* \text{ such that } \exists x \in \Delta^E \text{ with } \mb^\top \Par{M^* x} \le \1_V \text{ entrywise.}\]
Here, the variable $M^* x$ takes on the role of the matching; observe that since $x \in \Delta^E$, the $\ell_1$ norm of $M^* x$ is precisely $M^*$. Before giving the box-simplex objective which we will approximate with the value algorithm of Section~\ref{sec:value}, we recall that we can obtain a 2-approximation to $M^*$ in one pass.

\begin{restatable}{lemma}{restategreedy}\label{lem:greedy}
The greedy algorithm can be implemented in one semi-streaming pass over a graph using $O(n)$ space and $O(m)$ work to return a matching of size $M$ with $M \le M^* \le 2M$.
\end{restatable}
The proof is standard and deferred to Appendix~\ref{sec:vertex-size-redx}. As the first step of our approximate MCM algorithm, we obtain $M$ so that $M \le M^* \le 2M$ using the above greedy algorithm. We then use this estimate to construct a problem of the form \eqref{eq:boxsimplexprimal}, whose approximate solution also yields an approximate MCM, which we describe now. 

From the unweighted graph $G = (V, E)$, we construct a modified graph $\tG = (\tV, \tE)$ with two extra vertices and one extra zero-weight edge between them, and refer to its weighted adjacency matrix as $\tmb$, with an all-zero row for the extra edge. We define $\dmcm \in \{0, 1\}^{\tV}$ to be the vector which is $1$ in every coordinate of $V \in \tV$, and $0$ in the two extra vertices. Finally, we require one additional piece of notation: for a graph $(\tV, \tE)$, some \emph{demands} $d \in \R^{\tV}_{\ge 0}$, and some flow $\tx \in \R^{\tE}_{\ge 0}$, we define the \emph{overflow} by
\[\text{Overflow}_d(\tx) \defeq \sum_{j \in \tV} \Brack{\Par{\tmb^\top \tx - d}_+}_j. \]
Combinatorially, this can be interpreted as the total amount the flow $\tx$ violates $d$. We show that to obtain an approximate matching in $G$, it suffices to find an $\eps M$-approximate minimizer to the following problem (where $\dmcm$ as defined earlier is the demands for the MCM problem):
\begin{equation}\label{eq:matchoverflow}
\begin{gathered}
\min_{x \in \Delta^{\tE}} -2M\norm{x_E}_1 + \text{Overflow}_{\dmcm}(2Mx).
\end{gathered}
\end{equation}
In particular, note that we can represent any flow on the original edges $E$ with $\ell_1$ norm \emph{at most} $2M$ (in particular, the maximum cardinality matching) as the restriction of a flow in $\R_{\ge 0}^{\tE}$ with $\ell_1$ norm exactly $2M$. Hence, we can interpret \eqref{eq:matchoverflow} combinatorially as attempting to put as much flow on the edges of the original graph $E$ as possible, while penalizing the overflow. 

We begin by demonstrating that to find approximate minimizers to \eqref{eq:matchoverflow}, it suffices to find an approximate minimizer to the following $\ell_1$ regression problem:
\begin{equation}\label{eq:matchl1reg}
\begin{gathered}
\min_{x \in \Delta^{\tE}} \norm{\ma^\top x - b}_1, \text{ where }  \ma \defeq M\tmb,\; b \defeq \half \dmcm.
\end{gathered}
\end{equation}

\begin{lemma}\label{lem:equivalentsaddlepoint}
Any $\Delta$-additively approximate minimizer to problem \eqref{eq:matchoverflow} is a $\Delta$-additively approximate minimizer to problem \eqref{eq:matchl1reg}, for all $\Delta > 0$, and vice versa.
\end{lemma}
\begin{proof}
The property of being a $\Delta$-approximate minimizer is preserved if the entire problem is shifted by a constant scalar, so it suffices to show \eqref{eq:matchoverflow} and \eqref{eq:matchl1reg} differ by a constant scalar. Observe that for any scalar $v$, $\max(v, 0) = \half(v + |v|)$. Thus, we can write
\begin{align*}
\Overflow_{\dmcm}\Par{2Mx} &= \sum_{j \in V} \max\Par{2M\tmb^\top x - 1, 0} \\
&= \half \norm{2M\tmb^\top x - \dmcm}_1 + \half \sum_{j \in V} \Par{2M\tmb^\top x - 1} \\
&= \norm{\ma^\top x - b}_1 + \inprod{M\mb \dmcm}{x} - \half |V|.
\end{align*}
Comparing with \eqref{eq:matchoverflow} yields the conclusion, since $M\mb \dmcm$ is $2M$ times the indicator on $E$.
\end{proof}
Finally, to conclude this section, we give a simple procedure for rounding an almost-feasible flow to a feasible matching. This procedure comes with a guarantee which converts an overflow bound on the original flow to a bound on the difference in quality of the resulting rounded matching.

\notarxiv{
\begin{algorithm}
	\caption{$\mathsf{RemoveOverflow}(x, G, d)$}
	\begin{algorithmic}[1]\label{alg:removeOverflow}
		\STATE \textbf{Input:} Bipartite graph $G = (V, E)$ with incidence matrix $\mb$, $x \in \R^E_{\ge 0}$, demands $d \in \R^V_{\geq 0}$
		\STATE \textbf{Output:} $\tx \in \R^E_{\geq 0}$ with $\tx \leq x$ entrywise, $\mb^\top \tx \leq d$, and $\norm{\tx}_1 \geq \norm{x}_1 - \Overflow_d(x)$
		\STATE $d^x \gets \mb^\top x$
		\STATE $f \gets (d^x - d)_{+}$
		\STATE $\tx_e \gets x_e\left(1 - \max\Par{\tfrac{f_a}{d_a^x} , \tfrac{f_b}{d_b^x}}\right)$ for all $e = (a, b) \in E$ with $x_e > 0$
		\RETURN $\tx$
	\end{algorithmic}
\end{algorithm}	}

\arxiv{
\begin{algorithm}
\DontPrintSemicolon
		\KwInput{Bipartite graph $G = (V, E)$ with incidence matrix $\mb$, $x \in \R^E_{\ge 0}$, demands $d \in \R^V_{\geq 0}$}
		\KwOutput{$\tx \in \R^E_{\geq 0}$ with $\tx \leq x$ entrywise, $\mb^\top \tx \leq d$, and $\norm{\tx}_1 \geq \norm{x}_1 - \norm{\Par{\mb^\top x - d}_+}_1$}
		$d^x \gets \mb^\top x$\;
		$f \gets (d^x - d)_{+}$\;
		$\tx_e \gets x_e\left(1 - \max\Par{\tfrac{f_a}{d_a^x} , \tfrac{f_b}{d_b^x}}\right)$ for all $e = (a, b) \in E$ with $x_e > 0$\;
		\Return{$\tx$}
	\caption{$\mathsf{RemoveOverflow}(x, G, d)$}\label{alg:removeOverflow}
\end{algorithm}	
}

\begin{lemma}[Overflow removal]\label{lem:overflowremoval}
$\mathsf{RemoveOverflow}$ is correct, i.e.\ on input $x \in \R^E_{\ge 0}$, it outputs $\tx \in \R^E_{\geq 0}$ with $\tx \leq x$, $\mb^\top \tx \leq d$, and $\norm{\tx}_1 \geq \norm{x}_1 - \Overflow_d(x)$.
\end{lemma}
\begin{proof}
	First, $x_e > 0$ for $e = (a,b)$ implies $d_a^x > 0$ and $d_b^x > 0$ and therefore $\tx$ is well-defined. As $d, f \geq 0$ and $e \leq d_x$ entrywise we also see that $\tx \geq 0$ and $\tx \leq x$, giving the first guarantee.
	 
	Next, note that for all $a \in V$ if $[\mb^\top x]_a \leq d$ then $[\mb^\top \tx]_a \leq d_a$ as $\tx \leq x$. On the other hand, if 
	$[\mb^\top \tx]_a > d_a$ then $f_a = d^x_a - d_a$ and
	\begin{align*}
	[\mb^\top \tx]_a 
	=
	\sum_{(a,b) \in E} \tx
	\leq \sum_{(a,b) \in E} x_e \left(1 - \frac{f_a}{d_a^x}\right) 
	= d^x_a \left( 1 - \frac{d^x_a - d_a}{d_a} \right)
	= d_a.
	\end{align*}
	Consequently, in either case $[\mb^\top \tx]_a \leq d_a$ so $\mb^\top \tx \leq d$. 
	The last claim of the lemma follows from
	\begin{align*}
	\norm{\tx}_1
	&= \sum_{e \in E} x_e - \sum_{e = (a,b) \in E \mid x_e \neq 0} x_e \max \left\{\frac{f_a}{d_a^x} , \frac{f_b}{d_b^x} \right\}
	\geq \norm{x}_1 - \sum_{e = (a,b) \in E \mid x_e \neq 0} x_e \left( \frac{f_a}{d_a^x} + \frac{f_b}{d_b^x} \right) \\
	&= \norm{x}_1 - \sum_{a \in V} \sum_{e = (a,b) \in E \mid x_e \neq 0} x_e \frac{e_a}{d_a^x}
	= \norm{x}_1 - \sum_{a \in V} e_a
	= \norm{x}_1 - \norm{\Par{\mb^\top x - d}_+}_1.
	\end{align*}
\end{proof}
As a consequence of Lemma~\ref{lem:overflowremoval} with $d = \dmcm$, we have an (algorithmic) proof that any $\tx$ satisfying
\begin{equation}\label{eq:mcm_guarantees}
\norm{\tx}_1 - \Overflow_d(\tx) \ge (1 - \eps) M^*
\end{equation}
 can be rounded to a feasible matching on the same support with $\ell_1$ norm at least $(1 - \eps) M^*$.

\subsection{Additional tools}\label{ssec:tools}

Before proving Theorem~\ref{thm:match-rounding}, we give a few helper tools building upon prior works in the literature, to prove Theorem~\ref{thm:match-rounding}, as Propositions~\ref{prop:vertexreduce} and~\ref{prop:cc}. 
First, when $n \gg M^*$, we note that we can reduce the number of vertices to $O(M^*\log(\eps^{-1}))$ while only slightly decreasing the MCM size; a similar result is given in \cite{AhnG13}. \arxiv{We state the reduction here, and defer its proof to Appendix~\ref{sec:vertex-size-redx} for completeness.}

\begin{restatable}[Vertex size reduction]{proposition}
{restatevertexreduce}\label{prop:vertexreduce}There is a procedure $\RedSize$ (Algorithm~\ref{alg:red-MCM-Vsize}) which takes as input unweighted bipartite $G = (V, E)$ with MCM size $M^*$ and with $O(\log\eps^{-1})$ passes, $O(M^*\log\eps^{-1})$ space, and $O(m\log\eps^{-1})$ work outputs a subset $V' \subseteq V$ of size $O(M^*\log\eps^{-1})$ such that the induced subgraph $G[V']$ has a MCM of size at least $(1 - \eps)M^*$.
\end{restatable}

We also require a procedure for sparsifying supports based on cycle cancelling. Our cycle cancelling procedure takes a (possibly infeasible) flow given as an insertion-only stream, and produces a flow of at least the same value and with no additional overflow, with support of size at most $n$. The procedure generalizes to weighted bipartite matching problems as well; we remark that similar procedures have appeared in prior work (see e.g.\ \cite{KangP15}). We provide a detailed implementation of a data structure to prove Proposition~\ref{prop:cc} in Appendix~\ref{sec:cc}.

\begin{restatable}[Cycle cancelling]{proposition}{restatecc}\label{prop:cc}
Consider a (possibly weighted) matching problem on a bipartite graph $G = (V, E, w)$ (for MCM, we let $w = \1$). There is an algorithm that has the following property: given a stream of length $L$, consisting of edge-flow tuples $(e, f_e)$ where $e \in E$ and $f_e \in \R_{\ge 0}$, define $x \in \R_{\ge 0}^E$ to be the sum of all $f_e \1_e$ in the stream where $\1_e$ is the $1$-sparse indicator of edge $e$. Then the algorithm runs in $O(n)$ space and $O(L\log n)$ time, and outputs a flow $\tilde{x}$ supported on $O(n)$ edges forming a forest, so that $\inprod{w}{\tx} \ge \inprod{w}{x}$, and $\mb^\top x = \mb^\top \tx$.
\end{restatable}

\arxiv{

\subsection{Approximate MCM in fewer passes and optimal space}\label{ssec:fullproof}

We finally give a proof of Theorem~\ref{thm:match-rounding} by applying the tools we have built.

\restatemcmdown*
\begin{proof}
We first define the steps of the algorithm, where we adjust the error parameter $\eps$ by an appropriate constant in each of the following subroutines.
\begin{enumerate}
	\item Use Lemma~\ref{lem:greedy} to obtain an estimate $M \le M^* \le 2M$.
	\item If $M \log \frac 1 \eps \le n$, use Proposition~\ref{prop:vertexreduce} to reduce to $O(M \log \frac 1 \eps)$ vertices; overload $V$, $E$ to be the vertex and edge sets on the resulting induced subgraph.
	\item Solve \eqref{eq:matchl1reg} to $\eps M$ additive accuracy using Algorithm~\ref{alg:sherman}, with approximate minimizer $x$. Set $\tx = 2M x_E$.
	\item Apply Proposition~\ref{prop:cc} to $\tx$ in streaming fashion to obtain a flow $\hat{x}$.
	\item Greedily find an exact MCM on the support of $\hat{x}$.
\end{enumerate}	

We next demonstrate correctness. Proposition~\ref{prop:shermancorrect} shows that Algorithm~\ref{alg:sherman} is correct for implementing Step 3. By Proposition~\ref{prop:vertexreduce}, the result of Step 3 of the above algorithm satisfies \eqref{eq:mcm_guarantees} (by appropriately adjusting constants), since $\eps M^* \ge \eps M$ by Step 1, and the optimal value of \eqref{eq:matchoverflow} is at most $M^*$ by choosing $2Mx$ to be the restriction of the MCM to $E$, with all overflow placed on the extra edge. Proposition~\ref{prop:cc} then returns a sparse $O(\eps)$-approximate MCM supported on a forest, also satisfying \eqref{eq:mcm_guarantees}. Lemma~\ref{lem:overflowremoval} then implies that this support contains an $O(\eps)$-approximate MCM. It is well-known that to compute a maximum cardinality matching on a tree, the greedy algorithm of repeatedly taking a leaf edge and removing both endpoints from the graph suffices \cite{Goddard04}. Applying this to each tree in the forest yields correctness.

Finally, we discuss pass, space, and work complexities. Since Proposition~\ref{prop:vertexreduce} gives a vertex size bound of at most $M \log \frac 1 \eps$, using Algorithm~\ref{alg:sherman} and~\ref{alg:altmin} in Step 3 of the algorithm implies by Corollary~\ref{corr:lss} that the pass complexity is $\log n \cdot \log(\eps^{-1}) \cdot \eps^{-1}$, since the additive accuracy level is $O(\eps M^*)$ and $\norm{\ma}_\infty = \norm{b}_1 = O(M^*)$. Next, by applying Proposition~\ref{prop:cc} to the input stream $x$ which is the average of the iterates of Algorithm~\ref{alg:sherman}, we reduce the support of $x$ to be on $O(n)$ edges supported on a forest, while maintaining that \eqref{eq:mcm_guarantees} is satisfied. We can give $x$ in a stream of size $O(m\log n \cdot \eps^{-1})$ which induces the average iterate, per Theorem~\ref{thm:lss}. This takes $O(m\log^2 n \cdot \eps^{-1})$ work by Proposition~\ref{prop:vertexreduce}, and does not dominate the pass complexity. Finally, we explicitly store the support of the forest, so it is clear that Step 5 is implementable in $O(n)$ space and work.
\end{proof}

\section{Further matching applications}\label{ssec:further}
\arxiv{
Here we give a few applications using the semi-streaming matching algorithm proposed in the paper. Section~\ref{ssec:exactMCM} shows how to apply the semi-streaming matching algorithm to give exact MCM. Sections~\ref{ssec:wbg},~\ref{ssec:randsam-ot} and~\ref{ssec:mwm} extend our semi-streaming framework and algorithms to the model where the cost vector $c$ is non-uniform. We first state a general result on solving box-simplex games induced by weighted matching problems in Section~\ref{ssec:wbg}. Then in Section~\ref{ssec:randsam-ot} and~\ref{ssec:mwm}, we show how to adapt the model to specific applications. Specifically, we develop space-effcient solvers for optimal transportation and MWM respectively, by reductions to our more general solver.}

\subsection{Exact maximum cardinality matching}\label{ssec:exactMCM}

Here we  show how to leverage Theorem~\ref{thm:match-rounding} to obtain the following result on computing MCMs.

\restateexact*

We remark that the number of passes in Theorem~\ref{thm:mcm-exact} can be reduced to $\widetilde{O}((M^*)^{\frac 3 4 + o(1)})$ for any MCM problem with MCM size $M^*>0$, using the vertex size reduction of Proposition~\ref{prop:vertexreduce} in Section~\ref{ssec:tools}. To prove Theorem~\ref{thm:mcm-exact}, we use the following result from~\cite{LiuJS19} regarding reachability in directed graphs. This result was originally proved for parallel (PRAM) algorithms, but its extension to semi-streaming algorithms is straightforward and has been observed previously, e.g.\ in~\cite{LiuSZ20}. 

\begin{proposition}[\cite{LiuJS19}]\label{prop:JLS19}	There is a semi-streaming algorithm that given  \emph{directed} $G=(V,E)$ with $|V|=n$ and vertices $s$, $t$, with high probability finds a path from $s$ to $t$ (assuming one exists) in $O(n^{\frac12+o(1)})$ passes. 
\end{proposition}

We remark that as stated, \cite{LiuJS19}  only solves the decision problem of whether $s$ can reach $t$, while in Proposition~\ref{prop:JLS19} we require the path. 

\arxiv{
However, here we provide a brief sketch of how it is straightforward to modify \cite{LiuJS19} to compute a $s$-$t$ path and thereby prove Proposition~\ref{prop:JLS19}.

 The algorithm of~\cite{LiuJS19} computes a \emph{hopset} $H$ of $\tO(n)$ edges, i.e. a graph $H$ with $\tO(n)$ edges with the property that any vertex $u$ can reach another vertex $v$ in $G \cup H$ if and only $u$ can reach $v$ in $G$ also. Further, this hopset has the property that with high probability there exists a path of length $O(n^{\frac12+o(1)})$ from $s$ to $t$ in $G \cup H$. The hopset $H$ is computed in $O(\log n)$ iterations as $H_1 \cup H_2 \cup \dots H_{O(\log n)}$, where each $H_i$ is computed from $O(n^{\frac12+o(1)})$-depth breath-first searches in vertex-induced subgraphs of $G_i = G \cup H_1 \cup H_2 \cup \dots \cup H_{i-1}$. This computation requires $\tO(m)$ work and $O(n^{\frac12+o(1)})$ depth in the PRAM model (as was the motivation in~\cite{LiuJS19}). Similarly, it can also be implemented in $O(n^{\frac12+o(1)})$ passes in the semi-streaming model by performing the operations of each level of parallel depth within $O(1)$ passes. Furthermore, a directed edge $(u,v)$ is added to $H_i$ if and only if $u$ is the root of one of the BFS arborescences used to form $H_i$, and can reach $v$ in the same BFS computation. 
 
 To recover a path from $s$ to $t$, we run the algorithm of~\cite{LiuJS19} and store the arborescences used to form $H$. We then perform one last $O(n^{\frac12+o(1)})$-depth breadth-first search from $s$ in $G \cup H$: if a path $P$ from $s$ to $t$ in $G \cup H$ is found, we mark the edges of $P$ and store them in the memory. 
 We can then replace the edges of $P$ that belong to $H$ to the edges of $G$ stored as parts of arborescences used to form $H$. For $i$ from $O(\log{n})$ down to $1$, we replace all edges in this $s$-$t$ path belonging to $H_i$ with the corresponding arborescence paths belonging to $G_i$. In the end we obtain a (not necessarily simple) path from $s$ to $t$ in $G$, which we may make simple by removing duplicates (notice that technically for this step, we only need to run a BFS algorithm offline, i.e., with no more passes on the stream, from $s$ to $t$ among the stored edges of $G$ that form $H \cup P$).  
 }

\begin{proof}[Proof of Theorem~\ref{thm:mcm-exact}] 
	We set $\eps := n^{-\frac{3}{4}}$ and run our algorithm in Theorem~\ref{thm:match-rounding} on $G$ to obtain a matching $F$ with $|F| \geq (1 - \epsilon) \cdot M^* \geq M^* - n^{\frac 1 4}$ (as $M^* \leq n$). 
	This requires $O(n^{\frac 3 4 + o(1)})$ passes. 
	
	We then augment this approximate MCM $F$ to an exact MCM by repeatedly finding augmenting paths for it. We form the  (directed) residual graph for $F$ as follows: we take vertices $\{V, s, t \}$ and add the edges $(s,v)$ for all $v \in L$, $(v,t)$ for all $v \in R$, 
	and for any $e = (i,j) \in E(G)$ for $i \in L, j \in R$, we add $(i,j)$ if $e \notin F$ and $(j,i)$ if $e \in F$. As long as $F$ is suboptimal, we can increase its size by $1$ by finding a path from $s$ to $t$ in this residual graph and augmenting the flow with this path's edges: such a path can be found in $O(n^{\frac12 + o(1)})$ passes by applying 
	Proposition~\ref{prop:JLS19}. Since $|F| \geq M^* - n^{\frac 1 4}$ initially, it can be augmented at most $n^{\frac14}$ times, and the result follows.
\end{proof}

\begin{remark}\label{rem:space-pass-efficiency}
	The exact MCM in Theorem~\ref{thm:mcm-exact} is a direct byproduct of the simultaneous space and pass-efficiency of our approximate algorithm and cannot be achieved directly by the prior state-of-the-art. While the space of algorithm of~\cite{AssadiLT21} is $O(n)$ like ours, their algorithm requires $O(\eps^{-2})$ passes which combined with the augmentation approach and $O(n^{\frac12+o(1)})$-pass algorithm of~\cite{LiuJS19} for each augmentation, leads to $\tO(n)$ passes for exact MCM. On the other hand, while the pass complexity of the algorithm of~\cite{AhnG18} is sufficient on the surface for this augmentation approach, its space-complexity exceeds the $\tO(n)$ 
	requirement of semi-streaming algorithms when $\eps \ll (\poly\log{(n)})^{-1}$. More precisely, by setting $\eps$ small enough in the algorithm of~\cite{AhnG18} to obtain an $o(n)$ pass algorithm using the above approach, the space requirement would become $\Omega(n^2)$, which matches the trivial bound obtained by storing the input  in just one 
	pass and solving the problem optimally offline.)
\end{remark}

\subsection{Weighted bipartite matching under an $\ell_1$ constraint}\label{ssec:wbg}

In this section $G = (V, E)$ is a bipartite graph where $V = L \cup R$ with unweighted incidence matrix $\mb \in \{0, 1\}^{E \times V}$. We consider a weighted matching problem parameterized by a (possibly non-uniform) demand vector $d\in[0,1]^{V}$, and weights $w \in \R_{\ge 0}^E$, formally defined as follows:
\begin{equation}
M^* \defeq \max_{x \ge 0} w^\top x \text{ subject to } \mb^\top x\le d.\label{eq:weighted-gen}
\end{equation}

We also assume we know the value $S$ of the $\ell_1$ norm of an optimal matching which is feasible for the demands $d$, yielding the maximum matching weight $M^*$; in all our relevant applications, we will have $S \ge 1$. The assumption that we exactly know $S$ may seem restrictive; in all our applications (cf.\ Section~\ref{ssec:randsam-ot} and~\ref{ssec:mwm}), it will suffice to know an upper bound and pad the graph with a dummy edge appropriately. We will refer to the optimal matching as $Sx^*$, for some $x^* \in \Delta^E$. 

Next, we state the $\ell_1$ regression problem we solve to compute an approximate solution to \eqref{eq:weighted-gen}. Define the extended graph with vertex and edge sets $\widetilde{V}$, $\widetilde{E}$ and incidence matrix $\tmb$, as in the reduction of Section~\ref{ssec:formulation} (as a reminder, it simply adds an extra isolated edge with no constraints on endpoints). The overflow formulation we will solve to obtain an approximate minimizer to \eqref{eq:weighted-gen} is:
\begin{equation}\label{eq:matchminimaxweighted-overflow}
\begin{gathered}
\min_{x \in \Delta^{\tE}} -\inprod{w_E}{Sx} + \norm{w}_{\infty} \text{Overflow}_{d}(Sx).
\end{gathered}
\end{equation}
By the arguments of Lemma~\ref{lem:equivalentsaddlepoint}, \eqref{eq:matchminimaxweighted-overflow} is equivalent to the following $\ell_1$ regression problem up to a constant scalar shift:
 \begin{equation}\label{eq:matchminimax-weighted}\begin{gathered} \min_{x \in \Delta^{\tE}} -c^\top x + \norm{\ma^\top x - b}_1, \\
 \text{where } \ma \defeq \half S\norm{w}_\infty\tmb,\; b \defeq \half \norm{w}_{\infty} d,\text{ and } c \defeq S\Par{\norm{w}_\infty \1_E - w_E}.
\end{gathered}\end{equation} 

We next show we can use Algorithm~\ref{alg:removeOverflow} to obtain a feasible approximate MWM on a weighted graph.

\begin{lemma}[Weighted overflow removal]\label{lem:overflowremoval-weighted}
Suppose that for a weighted graph $(V, E, w)$ with MWM value $M^*$, $\hx \in \R_{\ge 0}^{\tE}$ is a flow satisfying 
\[\inprod{w_E}{\hx} - \norm{w}_\infty \Overflow_d(\hx) \ge M^* - \eps. \]
Applying $\mathsf{RemoveOverflow}$ to $\hx$ outputs $\tx \in \R^E_{\geq 0}$ with $\tx \leq x$, $\mb^\top \tx \leq d$, and $\inprod{w_E}{\tx} \ge M^* - \eps$.
\end{lemma}
\begin{proof}
By the $\ell_1$-$\ell_\infty$ H\"older's inequality and the guarantees of Algorithm~\ref{alg:removeOverflow},
\[\norm{w_E}{\tx} \ge \norm{w_E}{\hx} - \norm{w}_\infty \norm{\hx - \tx}_1 \ge \norm{w_E}{\hx} - \norm{w}_\infty \Overflow_d(\hx) \ge M^* - \eps.\]
\end{proof}

Finally, by combining Proposition~\ref{prop:shermancorrect} applied to \eqref{eq:matchminimax-weighted}, the cycle cancelling procedure of Proposition~\ref{prop:cc}, and the rounding procedure in Algorithm~\ref{alg:removeOverflow}, we obtain the following guarantee for solving \eqref{eq:weighted-gen}, summarized in Corollary~\ref{cor:gen}.

\begin{restatable}{corollary}{corgen}\label{cor:gen}
There is an algorithm that computes an $n$-sparse $\eps$-additively approximate solution $\hx$ to the problem \eqref{eq:weighted-gen} satisfying $w^\top\hx -w^\top (Sx^*) \ge \eps$, $\mb^\top \hx \le d$ using $O\Par{\gamma \log n \log (\gamma+\|w\|_\infty\|d\|_1 \eps^{-1})}$ passes, $O\Par{n}$ memory, and $O\Par{m \gamma \log n \log(n\gamma+n\|w\|_\infty\|d\|_1\eps^{-1})}$ work, where $\gamma = \frac{S\|w\|_\infty}{\eps}$.
\end{restatable}

\begin{proof}

The proof is analogous to that of Theorem~\ref{thm:match-rounding}. 	
We first define the steps of the algorithm, where we adjust the error parameter $\eps$ by an appropriate constant in each of the following subroutines.
\begin{enumerate}
	\item Solve \eqref{eq:matchminimax-weighted} to $\eps$ additive approximation, and let the approximate solution be $x$. Set $\tx = S x_E$.
	\item Apply Proposition~\ref{prop:cc} to $\tx$ in streaming fashion to obtain a flow $\hat{x}$.
	\item Greedily find an exact MWM on the support of $\hat{x}$.
\end{enumerate}	
The correctness and runtime follow in the same way as in Theorem~\ref{thm:match-rounding}, where we use Lemma~\ref{lem:overflowremoval-weighted} in place of Lemma~\ref{lem:overflowremoval}, which implies the support of $\hx$ contains a $\eps$-additively approximate MWM.
\end{proof}

We briefly remark on the utility of Corollary~\ref{cor:gen}. The generality of being able to handle arbitrary costs has the downside of an additive error guarantee rather than multiplicative. We will show how to apply this general result in different settings where either the optimal solution is supported on simplex (e.g.\ $S=1$ for optimal transport), or we can modify the graph appropriately to have a saturated optimal matching (e.g.\ $S=n$ for maximum weight matching).

\subsection{Optimal transportation}
\label{ssec:randsam-ot}

In this section, we give a semi-streaming implementation for solving the discrete optimal transportation problem. This problem is parameterized by costs\footnote{Costs are without loss of generality nonnegative, as adding a uniform multiple of $\norm{c}_\infty\1$ affects the cost of all transportation plans by a fixed amount and the quantity $\cmax$ by at most a constant factor.} $c \in \R^{E}_{\ge 0}$, and two sets of demands $\ell \in \Delta^L$, $r \in \Delta^R$ where $L \cup R = V$ is the bipartition of the vertices, we wish to find a transportation plan $x \in \Delta^{E}$ between the demands with (approximately) minimal cost, as specified by $c$; we defer a further discussion of this formulation to \cite{AltschulerWR17}. \cite{JambulapatiST19} showed that to obtain a transport plan approximating the optimum to $\eps$-additive accuracy, it suffices to solve the following problem to $\eps$ duality gap, where $d$ is the vertical concatenation of $\ell$ and $r$, $\cmax \defeq \norm{c}_\infty$, and $\mb$ is the adjacency matrix of the unweighted complete bipartite graph:
\begin{equation}\label{eq:otminiax}
\min_{x \in \Delta^{E}} -c^\top x + 2C_{\max} \text{Overflow}_d(x).
\end{equation}

The formulation above is an instance of the weighted formulation~\eqref{eq:matchminimaxweighted-overflow} with $S=1$, $w=c$. Note here we also use the standard adjacency matrix $\mb$ instead of the extended one $\tilde{\mb}$ in denoting overflows $\text{Overflow}(\cdot)$ for simplicity as opposed to the rest of the paper. We thus apply Corollary~\ref{cor:gen} to solve the problem and conclude by giving a complete result for semi-streaming optimal transportation.%

\restateotthm*
\begin{proof}
The proof follows by first applying Corollary~\ref{cor:gen} to~\eqref{eq:otminiax}, with error parameter $\eps \norm{c}_\infty$, to obtain an $O(n)$ sparse solution $x$ in the desired work and space budget. For demands $d \in \R^V$ set to be the vertical concatenation of the given $\ell$ and $r$, this is a transportation plan that satisfies $\mb^\top x \le d$, achieves an additive $\eps$ approximation in the optimal cost, and can be stored in $O(n)$ space. To make $\mb^\top x = d$, it suffices to add a rank-one correction term as in Algorithm 2 of \cite{AltschulerWR17}, whose coordinates can be computed in a stream in one pass and $O(n)$ space (by storing the rank-one components). This can only help the objective, and the resulting exact transport plan can be made sparse via Proposition~\ref{prop:cc} within the specified work budget.
\end{proof}

\subsection{Maximum weight matching}\label{ssec:mwm}

We give a result for computing an approximate maximum weight matching for a bipartite graph in the semi-streaming model. Using the construction in Section~\ref{ssec:wbg} with a demand vector $d = \1_{V}$ (i.e.\ the vector which is all-ones on $V$ and zeroes on $\tV\setminus V$), it is clear the $\ell_1$ norm of a maximum weight matching is $n$, simply by putting all additional flow on the extra edge. Applying Corollary~\ref{cor:gen} then directly gives a complete result for semi-streaming maximum weight matching.

\begin{restatable}{theorem}{restatemwmthm}\label{thm:mwm}
	There is a deterministic semi-streaming algorithm which given any weighted bipartite graph $G = (V, E, w)$ with $|V| = n$, $|E| = m$, and defining $\gamma \defeq \tfrac{n\norm{w}_\infty}{\eps}$, finds an $\eps$-additive maximum weight matching using $O\Par{\gamma\log n\log(n\gamma)}$ passes, $O\Par{n}$ space, and $O(m\gamma\log n\log (n\gamma))$ total work.
\end{restatable}

Finally, in the case when we have side information upper bounding the $\ell_1$ norm of any MWM by $S < n$, note that it suffices to solve \eqref{eq:matchminimaxweighted-overflow} with this scale $S$. This implies analogous wins in  Theorem~\ref{thm:mwm}, so that the parameter $\gamma$ scales linearly in $S$ rather than $n$.

\section{Transshipment}\label{sec:tran} 
To demonstrate the versatility of our approach, we also adapt our method to solve the transshipment problem on graphs in the semi-streaming model. In this problem we are given access to an undirected graph with non-negative edge weights $G = (V,E,w)$ and a vector $d \in \R^{V}$. Here $d$ is a demand vector, i.e.\ it specifies the desired flow imbalance on every vertex, and $w_e$ specifies the cost per  routing each unit of flow (in magnitude) on edge $e$. 

The goal of the transshipment problem is to compute a flow which routes the demands $d$ of minimum cost, given as the sum of weighted flow magnitudes. Formally, for a flow $f \in \R^E$ the imbalance of the flow at every vertex is given by $\mb^\top f$ where $\mb$ in this section refers to the signed incidence matrix of $G$.\footnote{We overload the notation $\mb$ for just this section, the only application where it is signed. The signed unweighted incidence matrix is defined in the same way as the unsigned one, except we choose an arbitrary orientation for every edge $e = (u, v)$ so the corresponding row in $\mb$ has a $1$ in column $u$ and a $-1$ in column $v$ or vice versa.}  Consequently, the transshipment problem is $\min_{f \in \R^E : \mb^\top f = d} \sum_{e \in E} w_e |f_e|$. Rescaling $f \gets \mw^{-1} f$ for $\mw \defeq \diag{w}$ shows that equivalently we can define the problem as follows.

\newcommand{\fopt}{f^{*}}
\newcommand{\zopt}{z_{*}}
\newcommand{\mv}{\mathbf{V}}
\newcommand{\fs}{f_{\mathsf{stream}}}
\newcommand{\optd}{\mathsf{opt}(d)}
\begin{definition}
For a graph $G =(V,E,w)$ with nonnegative edge weights, let $\mb \in \{-1,0,1\}^{m \times n}$ be its oriented (signed) unweighted incidence matrix. For any demand vector $d \in \R^n$, we write the transshipment problem, for a diagonal positive\footnote{For simplicity, we assume no zero-weight edges. Otherwise, we can form a spanning forest in one pass to remove zero-weight edges (since these shortcut the transshipment problem), and decompose into connected components.} matrix $\mw \in \R^{m \times w}$ as
\begin{equation}
\label{eqn:l1flow}
\min_{f:\mb^\top \mw^{-1} f =d} \norm{f}_1.
\end{equation}
Additionally, let $\optd$ refer to an arbitrary minimizer of problem~\eqref{eqn:l1flow} for a given demand $d$. 
\end{definition}

Applying our algorithmic framework, namely using the semi-streaming box-simplex game solver in Section~\ref{sec:value} and our tools for rounding and maintaining sparsity, we prove Theorem~\ref{thm:l1main}:

\restatetransshipthm*

We prove this result in several parts. We begin by giving a semi-streaming construction of an $\ell_1$-stretch approximator matrix $\mr  \in \R^{K \times n}$, defined as follows.

\begin{definition}\label{def:stretchapprox}
We say $\mr \in \R^{K \times n}$ is an $(\alpha, \beta)$-stretch approximator if it satisfies the following.
\end{definition}
\begin{itemize}
	\item For any $d \in \R^n$, $\norm{\optd}_1 \leq \norm{\mr d}_1 \leq \alpha \norm{\optd}_1$.
	\item $\sum_{v \in V} \text{deg}(v) \norm{\mr_{:v}}_0 \leq m \beta$, where $\text{deg}(v)$ is the degree of $v$ in $G$ and $\norm{\mr_{:v}}_0$ is the number of nonzero entries in the $v^{\textup{th}}$ column of $\mr$. 
	\item $\mr$ has at most $K \beta$ nonzero entries.
\end{itemize}

We construct $\mr \in \R^{K \times n}$, an $(\alpha, \beta)$-stretch approximator with $K = O(n \log n)$ and $\alpha, \beta = \log^{O(1)} n$, in Section~\ref{ssec:transshipment-approximator}. In Section~\ref{ssec:transshipment-redux}, we then use our stretch approximator to reduce a decision variant of transshipment to an appropriate box-simplex game. In Section~\ref{ssec:sparsetransshipment}, we demonstrate how to use our cycle-cancelling toolkit to round an approximate transshipment plan $f$ to be sparse. Finally, we put these pieces together in Section~\ref{ssec:transshipment-final}, proving Theorem~\ref{thm:l1main}.

\subsection{Constructing stretch approximators}\label{ssec:transshipment-approximator}

In this section, we give a semi-streaming construction of a stretch approximator. Our construction is a straightforward adaptation of a previous parallel construction by \cite{Li20}, proceeding in two steps. First we compute $H$, an $O(\log n)$ spanner subgraph of $G$ using $O(\log n)$ semi-streaming passes. We apply the construction of \cite{Li20} to $H$ offline, which no longer requires additional access to the edges of $G$. We then prove the approximation and sparsity bounds of the obtained matrix. 

\begin{definition}[Spanner]\label{def:spanner}
We say subgraph $H \subseteq G$ is an $\sigma$-spanner of graph $G = (V, E)$ if for all $u, v \in V$, where $d_G(u, v)$ denotes the shortest path distance through $G$,
\[d_G(u, v) \le d_H(u, v) \le \sigma d_G(u, v).\]
\end{definition}

It is well-known that Definition~\ref{def:spanner} implies that the optimal value of the objective \eqref{eqn:l1flow} is preserved up to a multiplicative $\sigma$ factor when restricting to $H$. This follows as every path in a decomposition of the solution to \eqref{eqn:l1flow} has its cost preserved up to a $\sigma$ factor upon routing through $H$.

\begin{lemma}[Corollary~5.2 of \cite{BeckerKKL17}]
	Let $G=(V,E,w)$ be a weighted graph given in an insertion-only stream. An $O(\log n)$-spanner of $G$ with $O(n \log n)$ edges can be computed using $O(\log n)$ passes, $O(n \log n)$ space, and $O(m \log n)$ total work.
	\label{lemma:spann}
\end{lemma}

\begin{lemma}[Theorem~4.2 of \cite{Li20}]
	\label{lemma:li20}
	Given a graph $G = (V,E,w)$, we can compute a matrix $\mr \in \R^{O(n \log n) \times n}$ satisfying the following properties with high probability in $n$:
	\begin{itemize}
		\item $\norm{\optd}_1 \leq \norm{\mr d}_1 \leq O(\log^{4.5} n) \norm{\optd}_1$ for all $d \in \mr^n$. 
		\item For any $v \in V$, the $v^{\textup{th}}$ column of $\mr$ has $O(\log^5 n (\log \log n)^{O(1)})$ nonzero entries in expectation over the construction of $\mr$.
	\end{itemize}
	The algorithm runs in $O(n \log^{10} n (\log \log n)^{O(1)})$ work and space. 
\end{lemma}
\begin{proof}
	The first and third condition on $\mr$ follow directly from Theorem~4.2 of \cite{Li20}. The second condition follows from Lemma~4.15 and the proof of Lemma~4.16 of \cite{Li20}.
\end{proof}

\begin{lemma}\label{lem:rconstruction}
	Let $G = (V,E,w)$ be a weighted undirected graph given in the semi-streaming graph model. Algorithm~\ref{alg:mr} computes $\mr \in \R^{K \times n}$, an $(\alpha, \beta)$-stretch approximator with $K = O(n \log n)$, $\alpha = O(\log^{5.5} n)$, and $\beta = O(\log^5 n (\log \log n)^{O(1)})$, in $O(\log n)$ passes, $O(n \log^{10} n (\log \log n)^{O(1)})$ space, and $O(m  \log^{10} n (\log \log n)^{O(1)})$ work with high probability in $n$.
\end{lemma}
\begin{proof}
	We first show correctness of the algorithm. First, we note that the cost of transshipment over $H$ is at most $O(\log n)$ times greater than the cost over $G$ since $d_G(u,v) \leq d_H(u,v) \leq O(\log n) d_G(u,v)$ for all $u,v \in V$. We thus have $\norm{\optd}_1 \leq \norm{\mr d}_1 \leq O(\log^{5.5} n) \norm{\optd}_1$ by the first condition of Lemma~\ref{lemma:li20} applied to the matrix $\mr$ we return. Next, the second condition of Lemma~\ref{lemma:li20} implies that each column of $\mr$ computed on line~\ref{line:mr} of Algorithm~\ref{alg:mr}\ has $\gamma$ nonzero entries in expectation. Thus $\mathbb{E} \left[ \sum_{v \in V} \norm{\mr_{:v}}_0 \right]\leq n \gamma$ and 
	\[
	\mathbb{E} \left[ \sum_{(u,v) \in E} \norm{\mr_{:u}}_0 + \norm{\mr_{:v}}_0\right] = 2 \mathbb{E} \left[ \sum_{v \in V} \mathsf{deg}(v) \norm{\mr_{:v}}_0 \right] \leq 2m \gamma.
	\]
	By Markov's inequality and a union bound, we thus have $S \leq 6 m \gamma$ and $\nnz(\mr) \leq 3 n \gamma$ with probability at least $\frac 1 3$ over the randomness of Lemma~\ref{lemma:li20}. Consequently, with high probability in $n$ after $O(\log n)$ repetitions the algorithm returns a matrix $\mr$ with all the desired properties.
	
	We now bound the pass, space, and work complexity. By Lemma~\ref{lemma:spann}, the graph $H$ can be computed using $O(\log n)$ passes over $G$ using $O(n \log n)$ space and $O( m \log n)$ total work. Next, the matrix $\mr$ on Line~\ref{line:mr} of Algorithm~\ref{alg:mr} can be computed offline in $O(m \log^{10} n (\log \log n)^{O(1)})$ work and $O(n \log^{10} n (\log \log n)^{O(1)})$ space. The condition on Line~\ref{line:compute-S} can be checked in a single pass over $G$ with $O(m \gamma)$ extra work (as we can terminate if the check fails).
\end{proof}

\begin{algorithm}
	\DontPrintSemicolon
	\KwInput{Graph $G = (V, E,w)$}
	\KwOutput{$\ell_1$-stretch approximator $\mr$}
	$\gamma = O(\log^5 n (\log \log n)^{O(1)})$\;
	$H = O(\log n)$-spanner of $G$ computed by Corollary~5.2 of \cite{BeckerKKL17}\;
	\For{$t \in [O(\log n)]$}{
		$\mr =$ matrix computed by Theorem~4.2 of \cite{Li20} applied to $H$\; \label{line:mr}
		Compute $S = \sum_{v \in V} \mathsf{deg}(v) \norm{\mr \mathbf{e}_v}_0$ in a stream over $G$\;\label{line:compute-S}
		\If{$S \leq 6 m \gamma$ and $\nnz(\mr) \leq 3 n \gamma$}{
			\Return{$\mr$}
		}
	}
	\caption{$\mathsf{StretchApprox}(G)$}
	\label{alg:mr}
\end{algorithm}

\subsection{Reduction to box-simplex game}\label{ssec:transshipment-redux}
In this subsection, we describe our reduction of transshipment to a box-simplex game. We begin by defining a flow-constrained variant of the transshipment problem.

\begin{definition}
Let $G=(V,E,w)$ be a graph, and let $d \in \R^n$. Let $\mr$ be any matrix satisfying $\norm{\mathsf{opt}(d)}_1 \leq \norm{\mr d}_1$. For any $t \ge 0$, we define the \emph{flow-constrained} transshipment problem as 
\begin{equation}
\label{eqn:l1cons}    
\min_{f:\norm{f}_1 \leq t} \norm{\mr \mb^\top \mw^{-1} f - \mr d}_1. 
\end{equation}
\end{definition}

We next relate solutions to the flow-constrained transshipment problem to the original problem~\eqref{eqn:l1flow}.

\begin{lemma}\label{lem:valueconstrained}
Let $G$ be a graph, $d \in \R^n$, and $\mr$ satisfy $\norm{\optd}_1 \leq \norm{ \mr d}_1$.
\begin{itemize}
    \item If $t \geq \norm{\optd}_1$, the optimal value of~\eqref{eqn:l1cons} is $0$.
    \item If $t < \norm{\optd}_1$, the optimal value of~\eqref{eqn:l1cons} is at least $\norm{\optd}_1 - t$.
\end{itemize}
\end{lemma}

\begin{proof}
For notational convenience, let $\fopt = \optd$. By definition we have $\mb^\top \mw^{-1} \fopt = d$. If $t \geq \norm{\fopt}_1$, we note $\fopt$ is feasible for~\eqref{eqn:l1cons} and therefore achieves an objective value of $0$. If instead $t < \norm{\fopt}_1$, consider any $f$ with $\norm{f}_1 \leq t$ and define $f' = \mathsf{opt}(d - \mb^\top \mw^{-1} f)$. We have
\[
\norm{\mr \mb^\top \mw^{-1} f - \mr d}_1 = \norm{\mr (d- \mb^\top \mw^{-1} f)}_1 \geq \norm{f'}_1,
\]
where for the last inequality we use the given condition on $\mr$ and the optimality of $f'$. On the other hand we have $\mb^\top \mw^{-1} f' = d - \mb^\top \mw^{-1} f$ and so $\mb^\top \mw^{-1} (f+ f') = d$. By optimality of $\fopt$ for \eqref{eqn:l1flow},
\[
\norm{\fopt}_1 \leq \norm{f+f'}_1 \leq \norm{f}_1 + \norm{f'}_1.
\]
Combining these inequalities, we have the desired
\[
\norm{\mr \mb^\top \mw^{-1} f - \mr d}_1 \geq \norm{\fopt}_1 - \norm{f}_1 \ge \norm{\fopt}_1 - t.
\]
\end{proof}

We remark that any $(\alpha, \beta)$-stretch approximator $\mr$ satisfies the assumption of Lemma~\ref{lem:valueconstrained}. Finally, we demonstrate how to express problems of the form~\eqref{eqn:l1cons} as box-simplex games.

\begin{lemma}\label{lem:l1boxsimplex}
Let $G$ be a graph, $d \in \R^n$, and $t \ge 0$. Let $\mr \in \R^{K \times n}$ be an $(\alpha, \beta)$-stretch approximator. Then~\eqref{eqn:l1cons} is equivalent to
\[
\min_{f' \in \Delta^{2E}} \norm{t \ma^\top f' - b}_1 = \min_{f' \in \Delta^{2E}} \max_{y \in [-1, 1]^K}  t y^\top \ma^\top f' - b^\top y
\]
for $b = \mr d$ and 
\[
\ma = \begin{pmatrix} \mw^{-1} \mb \mr^\top \\ - \mw^{-1} \mb \mr^\top \end{pmatrix}.
\]
Additionally, $\norm{\ma}_{\infty} \leq \alpha$ and $\nnz(\ma) = O(m \beta)$. If $R$ is stored explicitly, we may simulate access to a stream of the rows of $\ma$ and $|\ma|$ using $O(m \beta)$ total work and a single pass over $G$.
\end{lemma}
\begin{proof}
By duality of the $\ell_1$ norm and the $\ell_\infty$ norm, \eqref{eqn:l1cons} is equivalent to
\begin{align*}
\min_{f:\|f\|_1\le t} \norm{\mr \mb^\top \mw^{-1} f - \mr d}_1 & \equiv  \min_{f:\norm{f}_1 \leq t}  \max_{y:\substack{\norm{y}_{\infty} \leq 1}} y^\top \left( \mr \mb^\top \mw^{-1}f - \mr d \right).
\end{align*}

Next, we can write any $f \in \R^E$ as $f = f_+ - f_-$, where $f_{+} = \max \{ f, 0 \}$ and $f_{-} = - \min \{ f, 0 \}$ are entrywise nonnegative. Further, if $\norm{f}_1 \leq t$ there exists nonnegative $f_+$ and $f_-$ satisfying this condition with $\1^\top f_+ + \1^\top f_- = t$: one can simply select an arbitrary edge in $G$ and increase the corresponding entries in $f_+$ and $f_-$ by the same amount, which increases $\norm{f}_1$ without affecting $f_+ - f_-$. Thus, for any $f$ with $\norm{f}_1 \leq t$ there exists $\hat{f}= \frac 1 t [f_+; f_- ] \in \Delta^{2E}$ such that 
\[
f= t \begin{pmatrix} 
\id \\ 
-\id \end{pmatrix}^\top \hat{f} .
\]
Employing this variable substitution, our problem is equivalent to 
\[
\min_{\hat{f} \in \Delta^{2 E}} \max_{\norm{y}_\infty \leq 1} t y^\top \ma^\top \hat{f} - y^\top b
\]
as desired. The bound on $\norm{\ma}_\infty$ follows from the observation  
\[\norm{\mr \mb^\top \mw^{-1} u}_1 \le \alpha \norm{\mathsf{opt}\Par{\mb^\top \mw^{-1} u}}_1 \le \alpha \norm{u}_1\] 
for any $u \in \R^E$, by the first condition of Definition~\ref{def:stretchapprox} and definitions of $\mb, \mw$. This yields
\[
\norm{\ma}_\infty = \norm{\ma^\top}_1 = \max_{\norm{v}_1 = 1} \norm{\ma^\top v}_1 \le  \max_{\norm{v_1}_1 + \norm{v_2}_1 = 1} \norm{\mr \mb^\top \mw^{-1} v_1}_1 +  \norm{\mr \mb^\top \mw^{-1} v_2}_1 \leq \alpha
\]
as claimed. We conclude by bounding $\nnz(\ma)$ and showing that we can simulate streaming access to the columns of $\ma$ and $|\ma|$ in $O(m \beta)$ total work. Let $\delta_v$ denote the number of nonzero entries in the $v^{\text{th}}$ column of $\mr$. The column of $\mr \mb^\top \mw^{-1}$ corresponding to $e=(u,v)\in E$ is of the form
\[
w_e^{-1} \left(\mr_{:u}-\mr_{:v} \right),
\]
and thus has $O(\delta_u + \delta_v)$ nonzero entries. By the second assumed condition on $\mr$, summing over all columns gives $O(m \beta)$ nonzero entries together, thus the bound on $\nnz(\ma)$ follows. Finally, we may simulate access to columns of $\ma^\top$ and $|\ma^\top|$ via a stream over $G$ and forming the corresponding columns of $\ma^\top$ and $|\ma^\top|$ directly using $O(m\beta)$ work, from the above characterization of $\mr \mb^\top \mw^{-1}$.
\end{proof}

\subsection{Recovering a sparse flow}\label{ssec:sparsetransshipment}

We next give a simple procedure for taking a flow $f$ and rounding it to a sparse flow $f'$, such that the weighted cost is no larger and the marginal imbalances are preserved. At a high level, our algorithm performs the following steps.

\begin{enumerate}
	\item We form the ``double cover graph'' of $G = (V, E, w)$, a bipartite graph consisting of vertices $V_{\text{in}} \cup V_{\text{out}}$, which are two copies of $V$. For every $(u, v) \in E$, the double cover graph contains edges $(u_{\text{in}}, v_{\text{out}})$, denoting positive flow, and $(u_{\text{out}}, v_{\text{in}})$, denoting negative flow.
	\item We sparsify both the positive flow $f_+$ and the negative flow $f_-$ using Proposition~\ref{prop:cc}.
	\item We identify the sparsified flows into the original graph by using the signs appropriately, preserving marginal demands and not hurting the weighted cost.
\end{enumerate}

We make the transformation $f \gets \mw f$ in the following algorithm for consistency with Proposition~\ref{prop:cc}, so the marginal imbalances in accordance with transshipment are $\mb^\top f$ and the $\ell_1$ weight is $\norm{\mw f}_1$.

\begin{algorithm}
\DontPrintSemicolon
		\KwInput{Incremental stream of $f_+, f_- \in \R^E_{\geq 0}$ on disjoint supports, graph $G = (V,E,w)$}
		\KwOutput{Flow $f' \in \R^E$ with 
			\[\norm{f'}_0 = O(n),\; \mb^\top f' = \mb^\top (f_+ - f_-),\; \norm{\mw f'}_1 \leq \norm{\mw f_+}_1 + \norm{\mw f_-}_1.\]}
		Form $G'_+ = (V' = V_{\text{in}} \cup V_{\text{out}} , E'_+ = \emptyset)$ and $G'_- = (V' = V_{\text{in}} \cup V_{\text{out}} , E'_- = \emptyset)$ \;
		\For{$x_{(u,v)} \1_{(u,v)}$ in stream over $f_+$}{
		Add edge $(u_{\text{in}}, v_{\text{out}})$ with weight $x_{(u,v)}$ to $G'_+$\;
		Apply Proposition~\ref{prop:cc} to $G'_+$\;
		}
		\For{$x_{(u,v)} \mathbf{1}_{(u,v)}$ in stream over $f_-$}{
		Add edge $(u_{\text{out}}, v_{\text{in}})$ with weight $x_{(u,v)}$ to $G'_-$\;
		Apply Proposition~\ref{prop:cc} to $G'_-$\;
		}
	Let $f'_+$ be the output of Proposition~\ref{prop:cc} on $G'_+$ and $f'_-$ the output of Proposition~\ref{prop:cc} on $G'_-$\;
	\Return{$f'_+ - f'_-$}
	\label{alg:l1round}
	\caption{$\mathsf{RoundStream}(f_+, f_-)$}
\end{algorithm}

\begin{lemma}
\label{lemma:l1round}
Algorithm~\ref{alg:l1round} satisfies its output statement using a single pass over the streams $f_+$ and $f_-$, $O(n)$ space, and $O( \|f_+ + f_-\|_0) \log n)$ work.
\end{lemma}
\begin{proof}
The pass, space, and work complexities follow from Proposition~\ref{prop:cc}. The sparsity of $\|f'\|_0$ and the guarantee on weighted $\ell_1$ norm follows from the two applications of Proposition~\ref{prop:cc}, since $f_+$ and $f_-$ have disjoint supports. Finally, by construction of $G'_+$ and $G'_-$, the signed marginals of $f_+$ are preserved by $f'_+$, and the signed marginals of $f_-$ are preserved by $f'_-$. The conclusion follows since $f_+ - f_-$ then places the same amount of net flow on any vertex as $f'_+ - f'_-$.
\end{proof}

\subsection{Semi-streaming transshipment}\label{ssec:transshipment-final}

We finally give our full algorithm to solve the transshipment problem, and a proof of Theorem~\ref{thm:l1main}.

\begin{algorithm}
\DontPrintSemicolon
		\KwInput{Graph $G = (V, E,w)$, demand vector $d \in \R^n$, error tolerance $\eps \in (0, 1)$}
		\KwOutput{Flow $f$ with $\norm{f}_0 = O(n), \mb^\top \mw^{-1} f = d, $ and $\norm{f}_1 \leq (1+\eps) \norm{\optd}_1$.}
		$H = O(\log n)$-spanner of $G$ computed by Corollary~5.2 of \cite{BeckerKKL17}\;
        $\mr = \mathsf{StretchApprox}(G)$\;
        $t_{\max} \gets $ $2$-approximation to $\min_{\mb_H^\top \mw_H^{-1} f = d} \norm{f}_1$, scaled to be an overestimate to $\optd$\;
        
        \Comment*[f]{We let $\mb_H$ and $\mw_H$ be the appropriate restrictions of $\mb$ and $\mw$ to the subgraph $H$.}
        
        $t_{\min} \gets \frac{t_{\max}}{O(\log n)}$\;
        \While{$t_{\max} \geq (1+\frac \eps {\log n}) t_{\min}$}{
        $t \gets \frac{1}{2} (t_{\min} + t_{\max})$\;
        $\ma \gets t \begin{pmatrix} \mr \mb^\top \mw^{-1} & - \mr \mb^\top \mw^{-1}\end{pmatrix}^\top, b \gets \mr d$\;
        $Z, \fs \gets$ approximate value and streamed solution of \eqref{eq:boxsimplexprimal} with $\ma$, $b$ as defined above, and $c = \mzero$, to accuracy $\frac{\eps t}{O(\log n)}$\;
        \eIf{$Z \leq \frac{\eps t}{O(\log n)}$}{
        $t_{\max} \gets t$\; 
        }{
        $t_{\min} \gets t$\;
        }
        $f' \gets \mathsf{RoundStream}([\mw \fs]_+, [\mw \fs]_-)$ where $[\mw \fs]_+$, $[\mw \fs]_-$ are the restrictions of $\fs$ to the top and bottom rows of $\ma$ respectively after cancelling any nonnegative flow on corresponding edges in both parts\;
        $d' \gets d - \mb^\top f'$\;
        $f_{\mathsf{res}} \gets 2$-approximate solution to $\min_{\mb_H^\top \mw_H^{-1} f = d'} \norm{f}_1$\;
        \Return{$\mw^{-1}\left( \mathsf{RoundStream}([f' + \mw f_{\mathsf{res}}]_+, [f' + \mw f_{\mathsf{res}}]_-)\right)$}
        }
	\caption{$\mathsf{ApproxTransshipment}(G,d,\eps)$}
	\label{alg:l1}
\end{algorithm}

\begin{proof}[Proof of Theorem~\ref{thm:l1main}]
We begin by proving correctness of Algorithm~\ref{alg:l1}, and then discuss implementation costs to show they fit within the specified budgets.

First, $t_{\max}$ is an overestimate of $\optd$ by the definition of a spanner in Definition~\ref{def:spanner}, and similarly $t_{\min}$ is an underestimate. Next, by solving the problem in Line 8 to the stated accuracy, Lemma~\ref{lem:l1boxsimplex} shows we have an $\frac{\eps t}{O(\log n)}$-additive approximation to the value of the problem \eqref{eqn:l1cons}. Hence, Lemma~\ref{lem:valueconstrained} shows that when the binary search terminates, $t$ is an $O(\frac \eps {\log n})$-multiplicative approximation to $\norm{\optd}_1$. The additive error incurred due to the approximate routing of demands $d'$ on Line 15 can only affect the solution by $O(\eps)$ multiplicatively because of the quality of the spanner. Moreover, the applications of $\mathsf{RoundStream}$ do not affect meeting the demands and can only improve the $\ell_1$ value, by Lemma~\ref{lemma:l1round}. Hence, the output of the algorithm meets the demands and attains an $\eps$-multiplicative approximation to the value of the transshipment problem \eqref{eqn:l1flow}.

We next discuss pass, space, and work complexities. The costs of Lines 1 and 2 are given by Lemma~\ref{lemma:spann} and Lemma~\ref{lem:rconstruction} respectively. Lines 3 and 15 can be performed without streaming access to $G$ once $H$ has been stored, e.g.\ using \cite{Li20}, and do not dominate any of the costs.

There are at most $\log \frac{\log n} \eps$ iterations of the binary search, given the multiplicative range on $t$. The cost of each run in Line 8 is given by Theorem~\ref{thm:lss}, where $\norm{\ma}_\infty = t\alpha$ for the $\alpha$ in Lemma~\ref{lem:rconstruction}, and where the additive accuracy is given in Line 8. Finally, the costs of calls to $\mathsf{RoundStream}$ are given by Lemma~\ref{lemma:l1round} and do not dominate, completing the proof.

To conclude, we describe our algorithm for $(1+\eps)$-approximate shortest paths. We choose the demands $d = \1_u - \1_v$ for transshipment, and consider the flow $f$ computed by Algorithm~\ref{alg:l1}. We simply return the shortest $s$-$t$ path contained in the support of $f$: as $f$ satisfied $\mb^\top f = d$, it corresponds to a linear combination of $s$-$t$ paths. Returning the shortest path found in its support must therefore have length at most that of $f$, which is itself $\leq (1+\eps) \optd$. 
\end{proof} 
\arxiv{	\subsection*{Acknowledgments}

	We thank Alireza Farhadi, Cliff Liu, and Robert Tarjan for helpful conversations. 
	Researchers are supported in part by a Microsoft Research Faculty Fellowship, NSF CAREER Award CCF-1844855, NSF CAREER Award CCF-2047061, NSF Grant CCF-1955039, a PayPal research gift, a Sloan Research Fellowship, a gift from Google Research, and a Stanford Graduate Fellowship.}

	\arxiv{\bibliographystyle{alpha}	}
\newcommand{\etalchar}[1]{$^{#1}$}

	\newpage
\appendix
\addtocontents{toc}{\protect\setcounter{tocdepth}{1}}
\notarxiv{

\section{Further applications}
\label{sec:app}

Here we give omitted proofs for applications considered in the paper. For exact MCM, we discuss Proposition~\ref{prop:JLS19} in detail in Section~\ref{ssec:hopset}. We then extend our semi-streaming framework and algorithms to the model where the cost vector $c$ is non-uniform. We first state a general result on solving box-simplex games induced by weighted matching problems in Section~\ref{ssec:wbg}. Then in Section~\ref{ssec:randsam-ot} and~\ref{ssec:mwm}, we show how to adapt the model to specific applications. Specifically, we develop space-effcient solvers for optimal transportation and maximum weight matching respectively, by reductions to our more general solver.

\subsection{Discussion of Proposition~\ref{prop:JLS19}}\label{ssec:hopset}

The algorithm of~\cite{LiuJS19} computes a \emph{hopset} $H$ of $\tO(n)$ edges, which has the property that $u$ can reach $v$ in $G \cup H$ if and only $u$ can reach $v$ in $G$. Further, with high probability there exists a path of length $O(n^{\frac12+o(1)})$ from $s$ to $t$ in $G \cup H$. $H$ is computed in $O(\log n)$ iterations as $H_1 \cup H_2 \cup \dots H_{O(\log n)}$, where $H_i$ is computed from $O(n^{\frac12+o(1)})$-depth breath-first searches in vertex-induced subgraphs of $G_i = G \cup H_1 \cup H_2 \cup \dots \cup H_{i-1}$. This computation requires $\tO(m)$ work and $O(n^{\frac12+o(1)})$ depth in the PRAM model, and can be implemented in $O(n^{\frac12+o(1)})$ passes in the semi-streaming model by performing the operations of each level of parallel depth within a single pass. Further, a directed edge $(u,v)$ is added to $H_i$ if and only if $u$ can reach $v$ in one of the BFS arborescences used to form $H_i$, and the total size of these arborescences is $\tO(n)$. To recover a path from $s$ to $t$, we run the algorithm of~\cite{LiuJS19} and store the arborescences used to form $H$. We then perform a $O(n^{\frac12+o(1)})$-depth breadth-first search from $s$ in $G \cup H$: if a path from $s$ to $t$ in $G \cup H$ is found, we mark the edges in this path. For any $i$ from $O(\log n)$ to $1$, we then replace all edges in this $s-t$ path belonging to $H_i$ with the corresponding arborescence paths belonging to $G_i$. In the end we obtain a (not necessarily simple) path from $s$ to $t$ in $G$, which we may make simple by removing duplicates.

\subsection{Weighted bipartite matching under an $\ell_1$ constraint}\label{ssec:wbg}

In this section $G = (V, E)$ is a bipartite graph where $V = L \cup R$, $|L| = \tfrac n 2$ with unweighted edge incidence matrix $\mb \in \{0, 1\}^{E \times V}$. We consider a weighted matching problem parameterized by a (possibly non-uniform) demand vector $d\in[0,1]^{V}$, and weights $w \in \R_{\ge 0}^E$. We also assume we know the value $S$ of the $\ell_1$ norm of an optimal matching which is feasible for the demands $d$, yielding the maximum matching weight $M^*$; in all our applications, we will have $S \ge 1$.\footnote{The assumption that we exactly know $S$ may seem restrictive. In some applications (cf.\ Section~\ref{ssec:mwm}), it will suffice to know an upper bound and pad the graph with a dummy edge appropriately.} We will refer to the optimal matching as $Sx^*$, for some $x^* \in \Delta^E$. 

Under this problem parameterization, we give a MWM meta-result in Corollary~\ref{cor:gen}. We first give the specific box-simplex problem formulation we use. Let $c = \norm{w}_\infty\1_{V} - w$ be an all-positive vector with $\cmax \defeq \norm{c}_\infty \le \norm{w}_\infty$. Consider the problem
\begin{equation}\label{eq:mwmproblem-gen-app} \min_{x \in \Delta^{E}} S(c^\top x) + \sum_{v \in V} \max\Par{[\ma^\top x - b]_v, 0} = \min_{x \in \Delta^{E}} \max_{y \in [0, 1]^{V}} S(c^\top x) + y^\top\Par{\ma^\top x - b},\end{equation} 
where $\ma \defeq S\cmax \mb$, and $b \defeq \cmax d$.  Analogously to Lemma~\ref{lem:feasiblecorrect}, we show that similarly to MCM, we can round an approximate solution to the problem \eqref{eq:mwmproblem-gen-app} to be feasible without much loss in approximation guarantees, through the following variant of Algorithm~\ref{alg:feasible}. 

\begin{algorithm}
	\caption{$\Feasible(x, G)$}
	\begin{algorithmic}[1]\label{alg:feasible-gen}
		\STATE \textbf{Input:} Bipartite graph $G = (V, E)$ with vertex partition $V = L \cup R$ and edge-incidence matrix $\mb$, $x \in \Delta^E$ satisfying 
		\begin{equation}\label{eq:weightedroundreqs}
		\begin{aligned}
		\inprod{c}{Sx} \ge \inprod{c}{Sx^*} - \eps,\\
		\cmax \sum_{v\in L\cup R} \max\Par{[\mb^\top (Sx)-d]_v,0} \le \eps
		\end{aligned}
		\end{equation}
		\STATE \textbf{Output:} Fractional matching $\tx$ satisfying $-c^\top\tx +c^\top (Sx^*)\le 3\eps$, $\mb^\top \tx \le d$
		\STATE Let $\mx \in \R^{L \times R}$ be a zero-padded reshaped $Sx$
		\STATE $\lsum \gets \mx\1_{\frac{n}{2}}$
		\STATE $\mx' \gets \diag{\min\Par{\tfrac{d_l}{\lsum}, \1_{\frac{n}{2}}}} \mx$ entrywise
		\STATE $\rsum \gets (\mx')^\top \1_{\frac{n}{2}} $
		\STATE $\tmx \gets \mx' \diag{\min\Par{\tfrac{d_r}{\rsum}, \1_{\frac{n}{2}}}}$
		\RETURN $\tx =$ vectorized $\tmx$
	\end{algorithmic}
\end{algorithm}

\begin{equation}\label{eq:weighteddualgap}0 \le S(c^\top(x - x^*)) +\cmax \sum_{v\in L\cup R} \max\Par{[\mb^\top (Sx)-d]_v,0} \le \eps\end{equation}

\begin{lemma}\label{lem:feasiblecorrect-2}
Algorithm~\ref{alg:feasible-gen} is correct, i.e.\ it produces $\tx$ satisfying $-c^\top\tx +c^\top (Sx^*) \le 3\eps$, $\mb^\top \tx \le d$. Moreover, supposing that $\norm{x}_0 = O(n)$, using $O(1)$ passes, $O(n)$ auxiliary memory, and $O(m)$ work, we can explicitly compute and store $\tx$. 
\end{lemma}
\begin{proof}
For bounding the rounding loss, clearly the output $\tx$ is entrywise less than $Sx$, so it suffices to bound $\inprod{c}{Sx - \tx} \le 2\eps$. We note that $\|c\|_{\infty}\le \cmax$ and thus letting $x'$ be the vectorized $\mx'$, we have (following essentially identically from analogous calculations in Lemma~\ref{lem:feasiblecorrect})
\begin{align*}
-c^\top \tx &= -c^\top (Sx) +\cmax \sum_{v\in L\cup R} \max\Par{[\mb^\top (Sx)-d]_v,0} \le -c^\top(Sx) + \eps.
\end{align*}
The other parts of the proof (bounding the loss from $x'$ to $\tx$ by $\eps$ and the semi-streaming implementation) follow exactly the same as their counterparts in Lemma~\ref{lem:feasiblecorrect} so we omit them here. 
\end{proof}

Applying Proposition~\ref{prop:shermancorrect}, with the cycle cancelling  as in Appendix~\ref{sec:cc} and the rounding procedure in Algorithm~\ref{alg:feasible-gen} on the implicit solution whose coordinates can be computed in streaming fashion, we can obtain the following guarantees as in Corollary~\ref{cor:gen}.

\corgen*

\begin{proof}[Proof of Corollary~\ref{cor:gen}]
The proof is analogous to that of Theorem~\ref{thm:match-rounding}. 	
First, we apply Algorithms~\ref{alg:sherman} and~\ref{alg:altmin} to obtain a streaming representation of an average iterate which satisfies
\begin{equation}\label{eq:weighteddualgap}0 \le S(c^\top(x - x^*)) +\cmax \sum_{v\in L\cup R} \max\Par{[\mb^\top (Sx)-d]_v,0} \le \eps.\end{equation}
It is immediate from nonnegativity of the second term in \eqref{eq:weighteddualgap} that the first inequality in \eqref{eq:weightedroundreqs} holds for the average iterate. The second inequality follows from the arguments of Lemma~\ref{lem:cansearch} because removing one unit of violated demands can only affect the first summand in \eqref{eq:weighteddualgap} by $S\cmax$ units.

We then apply Proposition~\ref{prop:cc} on this average iterate, whose coordinates can be computed in streaming fashion by the representation in Corollary~\ref{corr:barxrep} and using Lemma~\ref{lem:implicitx}. By the guarantees of Proposition~\ref{prop:cc}, the requirements \eqref{eq:weightedroundreqs} are preserved, and the resulting fractional matching is supported on $O(n)$ edges. Finally, we run the rounding procedure in Algorithm~\ref{alg:feasible-gen} to obtain a fractional matching $\tx$ satisfying $c^\top\tx\ge c^\top (Sx^*)-\eps$ and $\mb^\top\tx\le d$, via Lemma~\ref{lem:feasiblecorrect-2} as the requirements are met.
\end{proof}

We briefly remark on the utility of Corollary~\ref{cor:gen}. The generality of being able to handle arbitrary costs has the downside of an additive error guarantee rather than multiplicative. We will show how to apply this general result in different settings where either the optimal solution is supported on simplex (e.g.\ $S=1$ for optimal transport), or we can modify the graph appropriately to have a saturated optimal matching (e.g.\ $S=n$ for maximum weight matching).

\subsection{Optimal transportation}
\label{ssec:randsam-ot}

In this section, we give a semi-streaming implementation for solving the discrete optimal transportation problem. In this problem parameterized by costs\footnote{Costs are without loss of generality nonnegative, as adding a uniform multiple of $\norm{c}_\infty\1$ affects the cost of all transportation plans by a fixed amount and the quantity $\cmax$ by at most a constant factor.} $c \in \R^{n^2}_{\ge 0}$, and two sets of demands $\ell, r \in \Delta^n$, we wish to find a transportation plan $x \in \Delta^{n^2}$ between the demands with (approximately) minimal cost, as specified by $c$; we defer a further discussion of this formulation to \cite{AltschulerWR17}. \cite{JambulapatiST19} showed that to obtain a transport plan approximating the optimum to $\eps$-additive accuracy, it suffices to solve the following problem to $\eps$ duality gap, where $d$ is the vertical concatenation of $\ell$ and $r$, $\cmax \defeq \norm{c}_\infty$, and $\mb$ is the adjacency matrix of the unweighted complete bipartite graph:
\begin{equation}\label{eq:otminiax}
\min_{x\in\Delta^{n^2}}\max_{y\in [0,1]^n}  -c^\top x + C_{\max}\cdot y^\top (\mb^\top x-d).
\end{equation}

The formulation above is an instance of the weighted formulation~\eqref{eq:mwmproblem-gen-app} with $S=1$. We conclude by giving a complete result for semi-streaming optimal transportation.%

\restateotthm*
\begin{proof}
The proof follows by first applying Corollary~\ref{cor:gen} to ~\eqref{eq:otminiax} to obtain an $O(n)$ sparse solution $\hx$ in the desired work and space budget. Then, we round this fractional transport plan to an exact plan with a rank-1 correction step, i.e. by letting
\[[x]_{ij} = [\hx]_{ij} + e_r e_c^\top,\]
where $e_r = \tfrac{d_L - \hat{\mx}\1_{n/2}}{\|d_L - \hat{\mx}\1_{n/2}\|_1}$ and $e_c = d_R - \hat{\mx}^\top\1_{n/2}$ for a reshaped $\hx$ in the form $\hat{\mx}\in\R^{L\times R}$. These corrections can be explicitly computed using $O(n)$ additional memory, and can only improve the quality of the solution since $c\ge0$. Adjusting constants, we have the approximation guarantee. Finally, by applying Proposition~\ref{prop:cc} to the input consisting of a reshaped $\hx$ plus $e_r e_c^\top$, we reduce the support size to $O(n)$ while preserving that the demands are satisfied, as desired.
\end{proof}

\subsection{Maximum weight matching}\label{ssec:mwm}

We give a result for computing an approximate maximum weight matching for a bipartite graph in the semi-streaming model. We first reduce the MWM problem to a box-simplex game. Given a bipartite graph $G = (V, E)$ where $V = L \cup R$, $|L| = n$, with unweighted edge incidence matrix $\mb \in \{0, 1\}^{E \times V}$ and weights $w \in \R_{\ge 0}^E$, to simplify the problem slightly so the maximum matching has $\ell_1$ norm $n$ without loss of generality, we consider a modified weighted graph $\tG = (\tV, \tE, \tw)$ constructed from $G$ as follows: we add one additional vertex $\ell$ to $L$, and a vertex $r$ to $R$. We further add additional edges from $r$ to all vertices in $L$, and from $\ell$ to all vertices in $R$, all with weight $0$. We call $\tmb \in \{0, 1\}^{\tE \times \tV}$ the unweighted edge incidence matrix of $\tG$; note that $|\tV|$ and $|\tE|$ are $|V|$ and $|E|$ up to constant factors. 

\begin{lemma}\label{lem:tgprops}
The modified graph $\tG$ satisfies the following properties.
\begin{enumerate}
	\item The maximum matching weight for $\tG$ is also $M^*$, the maximum matching weight for $G$.
	\item Define function $\opt(v) \defeq \max_{\tmb^\top x \le \1} \inprod{\tw}{x}$ over $x \in \R_{\ge 0}^{\tE}$, $\norm{x}_1 = v$, so $\opt(v)$ is the largest weight a matching on $\tG$ of size $v$ can have. Then, $\opt(v) \le \opt(v')$ for any $0 \le v \le v' \le n$.
\end{enumerate}
\end{lemma}
\begin{proof}
For the first property, note any maximum weight matching on $\tG$ without loss of generality puts no flow on any of the additional edges (as they are unweighted), and $G$ is a subgraph of $\tG$.

To see the second property, we need to exhibit a matching on $\tG$ with $\ell_1$ norm $v'$ and value at least $\opt(v)$. To do so, we can simply take the matching attaining $\opt(v)$ with $\ell_1$ norm $v$, and arbitrarily route $v' - v$ units of flow from unmatched vertices in $L$ to $r$, and similarly from $\ell$ to unmatched vertices in $R$, as they both have capacity $n$. 
\end{proof}

Lemma~\ref{lem:tgprops} implies that without loss of generality, we can set the $\ell_1$ norm of the simplex variable to be exactly $n$.  We next give the specific box-simplex problem formulation we use. Let $c = \norm{w}_\infty\1_{\tV} - w$ be an all-positive vector with $\cmax \defeq \norm{c}_\infty \le \norm{w}_\infty$. Consider the problem
\begin{equation}\label{eq:mwmproblem} \min_{x \in \Delta^{\tE}} n(c^\top x) + \sum_{v \in \tV} \max\Par{[\ma^\top x - b]_v, 0} = \min_{x \in \Delta^{\tE}} \max_{y \in [0, 1]^{\tV}} n(c^\top x) + y^\top\Par{\ma^\top x - b},\end{equation}
where $\ma^\top = n\cmax \tmb^\top$, and $b=\cmax\1_{\tV}$. Note this is an instance of the weighted formulation~\eqref{eq:mwmproblem-gen-app} with $S=n$. Applying Corollary~\ref{cor:gen} with the cycle cancelling procedure in Proposition~\ref{prop:cc} directly gives a complete result for semi-streaming maximum weight matching.

\begin{restatable}{theorem}{restatemwmthm}\label{thm:mwm}
	For a maximum weight matching problem parameterized by weights $w$ and letting $\gamma = \tfrac{n\norm{w}_\infty}{\eps}$, we can obtain an $\eps$-additive maximum weight matching using $O\Par{\gamma\log n\log\gamma}$ passes, $O\Par{n\log n}$ auxiliary memory, and $O(m\gamma\log\gamma)$ work.
\end{restatable}

Finally, in the case when we have side information showing $\opt(S) = \opt(n)$ (in the notation of Lemma~\ref{lem:tgprops}), for some $S \le n$, note that it suffices to solve an analogous problem to \eqref{eq:mwmproblem} with the cost vector $c$ and the adjacency matrix $\tmb$ scaled by  $S$, rather than $n$. This implies analogous wins in  Theorem~\ref{thm:mwm}, so that the parameter $\gamma$ scales linearly in $S$ rather than $n$.
\section{Deferred proofs from Section~\ref{sec:value}}\label{app:helpersherman}

We provide proofs of Proposition~\ref{prop:shermancorrect} and Lemma~\ref{lem:implsimplex} here.

\restateshermancorrect*
\begin{proof}
	The claim on the number of iterations required by Algorithm~\ref{alg:sherman} follows immediately from Theorem 2 of \cite{CohenST21}. It remains to show the bound on $K$ (the number of alternating minimization steps needed) is correct. For simplicity, we will prove this bound suffices for the computation of $w_t$; the bound for $z_{t + 1}$ follows analogously. Lemma 7 of \cite{JambulapatiST19} demonstrates that in solving the minimization subproblem of Algorithm~\ref{alg:altmin}, every iteration of alternating minimizaton decreases the error (additive difference in value to optimality) by a constant factor, so we only require a bound on the initial error of each subproblem in Line 5 of Algorithm~\ref{alg:sherman}. Note that the subproblem is of the form: minimize over $(x, y) \in \Delta^m \times [-1, 1]^n$,
	\begin{equation}\label{eq:minwt}\begin{aligned}\inprod{\frac 1 3(\ma y_t + c) + \ma(y_t^2)}{x} + \inprod{\frac 1 3(-\ma^\top x_t + b) - 2\diag{y_t}\ma^\top x_t}{y} \\
	+ x^\top \ma (y^2) + 10\norm{\ma}_\infty \sum_{i \in [m]} x_i \log \frac{x_i}{[x_t]_i}. \end{aligned}\end{equation}
	It is immediate that the range of the first three summands in the objective \eqref{eq:minwt} over the range $\Delta^m \times [-1, 1]^n$ is $O(\max(\norm{\ma}_\infty, \norm{c}_\infty, \norm{b}_1))$. Moreover, Theorem 2 of \cite{CohenST21} shows that the $x$ block of the minimizer of the objective~\ref{eq:minwt} is entrywise within a multiplicative factor $2$ from $x_t$; in other words, for all $i \in [m]$, $[x'_t]_i \in [\thalf [x_t]_i, 2[x_t]_i]$. Hence, the absolute value of the fourth summand in \eqref{eq:minwt} at the minimizer is bounded by 
	\[10\norm{\ma}_\infty \sum_{i \in [m]} x_i \log 2 = O\Par{\norm{\ma}_\infty}.\] 
	and at the initial point $x^{(0)} = x_t$ it has value $0$. Combining these bounds, and using that the desired accuracy is $\tfrac \eps 2$, yields the required bound on $K$. Regarding the value approximation, we claim that any $\eps$-approximate saddle point $(\bx, \by)$ to a convex-concave function $f$ on the domain $\xset, \yset$ approximates the value of the saddle point within an additive $\eps$. To see this, let $(x^\star, y^\star)$ be the exact saddle point of $f$ so that $f(x^\star, y^\star) = \opt$. By definition of $(\bx, \by)$,
	\[\Par{\max_{y \in \yset} f(\bx, y) - \opt} + \Par{\opt - \min_{x \in \xset} f(x, \by)} \le \eps.\]
	Moreover, both quantities on the left hand side are nonnegative by definition of $\opt$, i.e.
	\[\opt = f(x^\star, y^\star) = \min_{x \in \xset} \Par{\max_{y \in \yset} f(x, y)} = \max_{y \in \yset} \Par{\min_{x \in \xset} f(x, y)}.\] 
	The conclusion follows by applying the bounds $\max_{y \in \yset} f(\bx, y) \ge f(\bx, \by) \ge \min_{x \in \xset} f(x, \by)$. 
	
	Finally, we discuss changing the domain. Algorithm~\ref{alg:sherman} is also Algorithm 1 of \cite{CohenST21} as analyzed in its Theorem 2. The guarantees and requirements of each step of Algorithm~\ref{alg:sherman} are thus unchanged; the only difference is that the bound on the domain size decreased, which can only help the runtime bounds (cf.\ Proposition 1, \cite{CohenST21}). The analysis of Algorithm~\ref{alg:altmin} is identical, and its complexity follows since Line 6 decomposes into separable minimization problems.
\end{proof}

\restateimplicitx*
\begin{proof}
	We proceed by induction; in the first iteration, we have $v_0 = \mzero_n$, $\lam_0 = 0$.
	
	\textbf{Preserving the $w_t$ invariant.} Suppose inductively that $z_t = (x_t, y_t)$ where $x_t \propto \exp(\ma v_t + \lam_t c)$ for explicitly stored values $v_t$, $\lam_t$, $y_t$; we will drop the index $t$ for simplicity and refer to these as $\bv$, $\blam$, $\by$. Consider the procedure Algorithm~\ref{alg:altmin} initalized with these values, and note that
	\begin{align*}\gamma\x &= \frac{1}{3}\Par{\ma \by + c} - 10\norm{\ma}_\infty\Par{\ma \bv + \blam c} - \ma \Par{\by^2},\\
	\gamma\y &=\frac{1}{3}\Par{b - \ma^\top \frac{\exp(\ma \bv + \blam c)}{\norm{\exp(\ma \bv + \blam c)}_1}} - 2\diag{\by} \ma^\top \frac{\exp(\ma \bv + \blam c)}{\norm{\exp(\ma \bv + \blam c)}_1}.\end{align*}
Using $r(x, y) = \inprod{\ma (y^2)}{x} + 10\norm{\ma}_\infty\sum_{i \in [m]} x_i\log x_i$, we explicitly compute that for each $0 \le k < K$,
	\begin{equation}\label{eq:altminstep}
	\begin{aligned}
	x^{(k + 1)} &= \argmin_{x \in \Delta^m}\Brace{\inprod{\ma \Par{(y^{(k)})^2} + \gamma\x}{x} + 10\norm{\ma}_\infty \sum_{i \in [m]} x_i \log x_i}\\
	&\propto \exp\Par{-\frac{1}{10\norm{\ma}_\infty}\Par{\ma \Par{\frac 1 3 \by - 10\norm{\ma}_\infty \bv - (\by)^2+ (y^{(k)})^2}+ \Par{\frac 1 3 + 10\norm{\ma}_\infty\blam} c}},\\
	y^{(k + 1)} &= \argmin_{y \in [-1, 1]^n}\Brace{\inprod{\gamma\y}{y} + \inprod{\ma^\top x^{(k + 1)}}{y^2}} \\
	&= \text{med}\Par{-1, 1, -\frac{\gamma\y}{2\ma^\top x^{(k + 1)}}} \text{ entrywise.}
	\end{aligned}
	\end{equation}
	In the last line, the median operation truncates the vector $-\tfrac{\gamma\y}{2\ma^\top x^{(k + 1)}}$ coordinatewise on the box $[-1, 1]^n$. Now, suppose at the start of iteration $k$ of Algorithm~\ref{alg:altmin}, we have the invariant 
	\begin{equation}\label{eq:altmininvar}
	x^{(k)} \propto \exp\Par{\ma v^{(k)} + \lam^{(k)} c},
	\end{equation}
	and we have explicitly stored the tuple $(v^{(k)}, \lam^{(k)}, y^{(k)})$; in the first iteration, we can clearly choose $(v^{(0)}, \lam^{(0)}, y^{(0)}) = (\bv, \blam, \by)$. By the derivation \eqref{eq:altminstep}, we can update
	\begin{align*}
		v^{(k + 1)} & \gets -\frac{1}{10\norm{\ma}_\infty} \Par{\frac 1 3 \by - 10\norm{\ma}_\infty \bv - (\by)^2+ (y^{(k)})^2},\\
		\lam^{(k + 1)} & \gets -\frac{1}{10\norm{\ma}_\infty}\Par{\frac 1 3 + 10\norm{\ma}_\infty\blam},
	\end{align*}
	preserving \eqref{eq:altmininvar}. Moreover, note that $\gamma\y$ can be explicitly stored (it can be computed in one pass using Lemma~\ref{lem:implicitx}), and by applying Lemma~\ref{lem:implicitx} and the form \eqref{eq:altmininvar}, we can compute the vector $\ma^\top x^{(k + 1)}$ explicitly in one pass over the data and $O(m)$ work. This allows us to explicitly store $y^{(k + 1)}$ as well. The final point $w_t$ is one of the iterates $(x^{(K)}, y^{(K)})$, proving the desired invariant. Finally, we remark in order to perform these computations we only need to store the tuple $(v^{(k)}, \lam^{(k)}, y^{(k)})$ from the prior iteration, in $O(n)$ auxiliary memory.
	
	\textbf{Preserving the $z_{t + 1}$ invariant.} The argument for preserving the invariant on $z_{t + 1}$ is exactly analogous to the above argument regarding $w_t$; the only modification is that the input vector to Algorithm~\ref{alg:altmin} is 
	\begin{align*}
	\gamma\x &= \frac{1}{3}\Par{\ma y'_t + c} - 10\norm{\ma}_\infty\Par{\ma \bv + \blam c} - \ma \Par{\by^2},\\
	\gamma\y &=\frac{1}{3}\Par{\ma^\top x'_t} - 2\diag{\by} \ma^\top x'_t.
	\end{align*}
	However, by the previous argument shows that the vector $y'_t$ can be explicitly computed, and the vector $x'_t$ satisfies the invariant in the lemma statement. The same inductive argument shows that we can represent every intermediate iterate in the computation of $z_{t + 1}$ in the desired form.
	
	\textbf{Numerical stability.} We make a brief comment regarding numerical stability in the semi-streaming model, which may occur due to exponentiating vectors with a large range $\omega(\log m)$ (in defining simplex variables). It was shown in \cite{JambulapatiST19} that Algorithm~\ref{alg:sherman} is stable to increasing the value of any coordinate of a simplex variable which is e.g.\ $m^{10}$ multiplicatively smaller than the largest to reach this threshold, and renormalizing. In computations of Lemma~\ref{lem:implicitx}, we can first store the largest coordinate of $\ma v + \lam c$ in one pass. Then, for every coordinate more than $10\log m$ smaller than the largest coordinate, we will instead treat it as if its value was $10\log m$ smaller than the maximum in all computations, requiring one extra pass over the data.
\end{proof}

We explicitly state a complete low-space implementation of Algorithm~\ref{alg:sherman} here for completeness as Algorithm~\ref{alg:lowspacesherman}. The correctness of the implementation follows from the proof of Lemma~\ref{lem:implsimplex}.

\begin{algorithm}[ht!]
	\caption{$\LSSherman(\ma, b, c, \epsilon)$}
	\begin{algorithmic}[1]\label{alg:lowspacesherman}
		\STATE \textbf{Input:} $\ma \in \R^{m \times n}_{\ge 0}$, $c \in \R^m$, $b \in \R^n$, $0 \le \eps \le \norm{\ma}_\infty$
		\STATE \textbf{Output:} $\{v'_t\}_{0 \le t < T} \subset \R^n$, $\{\lam'_t\}_{0 \le t < T} \subset \R$, $\{y'_t\}_{0 \le t < T} \subset \R^n$ so that for
		\[\by \defeq \frac{1}{T}\sum_{0 \le t < T} y'_t,\; \bx \defeq \frac{1}{T}\sum_{0 \le t < T} \frac{\exp(\ma v'_t + \lam'_t c)}{\norm{\exp(\ma v'_t + \lam'_t c)}_1},\]
		the pair $(\bx, \by)$ is an $\eps$-approximate saddle point to \eqref{eq:boxsimplex}
		\STATE $T \gets O(\tfrac{\norm{\ma}_\infty \log m}{\eps})$, $K \gets O(\log \tfrac{\max(\norm{\ma}_\infty, \norm{c}_\infty, \norm{b}_1)}{\eps})$
		\STATE $t \gets 0$, $\lam_0 \gets 0$, $v_0 \gets \mzero_n$, $y_0 \gets \mzero_n$
		\WHILE{$t < T$}
		\STATE $(v^{(0)}, \lam^{(0)}, y^{(0)}) \gets (v_t, \lam_t, y_t)$
		\STATE $\gamma\y \gets \tfrac 1 3 (b - \ma^\top \tfrac{\exp(\ma v_t + \lam_t c)}{\norm{\exp(\ma v_t + \lam_t c)}_1}) - 2\diag{y_t} \ma^\top \tfrac{\exp(\ma v_t + \lam_t c)}{\norm{\exp(\ma v_t + \lam_t c)}_1}$, computed using Lemma~\ref{lem:implicitx}
		\FOR{$0 \le k < K$}
		\STATE $v^{(k + 1)} \gets -\tfrac{1}{10\norm{\ma}_\infty} (\tfrac 1 3 y_t  - 10\norm{\ma}_\infty v^{(0)} - (y^{(0)})^2 + (y^{(k)})^2)$
		\STATE $\lam^{(k + 1)} \gets -\tfrac{1}{10\norm{\ma}_\infty} (\tfrac 1 3 + 10\norm{\ma}_\infty \lam^{(0)})$
		\STATE $d^{(k + 1)} \gets 2\ma^\top \tfrac{\exp(\ma v^{(k + 1)}+ \lam^{(k + 1)} c)}{\norm{\exp(\ma v^{(k + 1)}+ \lam^{(k + 1)} c)}_1}$, computed using Lemma~\ref{lem:implicitx}
		\STATE $y^{(k + 1)} \gets \textup{med}(-1, 1, -\tfrac{\gamma\y}{d^{(k + 1)}})$ entrywise 
		\ENDFOR
		\STATE $(v'_t, \lam'_t, y'_t) \gets (v^{(K)}, \lam^{(K)}, y^{(K)})$
		\STATE $(v^{(0)}, \lam^{(0)}, y^{(0)}) \gets (v_t, \lam_t, y_t)$
		\STATE $\gamma\y \gets \tfrac 1 3 (b - \ma^\top \tfrac{\exp(\ma v'_t + \lam'_t c)}{\norm{\exp(\ma v'_t + \lam'_t c)}_1}) - 2\diag{y_t} \ma^\top \tfrac{\exp(\ma v_t + \lam_t c)}{\norm{\exp(\ma v_t + \lam_t c)}_1}$, computed using Lemma~\ref{lem:implicitx}
		\FOR{$0 \le k < K$}
		\STATE $v^{(k + 1)} \gets -\tfrac{1}{10\norm{\ma}_\infty} (\tfrac 1 3 y'_t - 10\norm{\ma}_\infty v^{(0)} - (y^{(0)})^2 + (y^{(k)})^2)$
		\STATE $\lam^{(k + 1)} \gets -\tfrac{1}{10\norm{\ma}_\infty} (\tfrac 1 3 + 10\norm{\ma}_\infty \lam^{(0)})$
		\STATE $d^{(k + 1)} \gets 2\ma^\top \tfrac{\exp(\ma v^{(k + 1)}+ \lam^{(k + 1)} c)}{\norm{\exp(\ma v^{(k + 1)}+ \lam^{(k + 1)} c)}_1}$, computed using Lemma~\ref{lem:implicitx}
		\STATE $y^{(k + 1)} \gets \textup{med}(-1, 1, -\tfrac{\gamma\y}{d^{(k + 1)}})$ entrywise 
		\ENDFOR		
		\STATE $(v_{t + 1}, \lam_{t + 1}, y_{t + 1}) \gets (v^{(K)}, \lam^{(K)}, y^{(K)})$
		\STATE $t \gets t + 1$
		\ENDWHILE
	\end{algorithmic}
\end{algorithm}
\section{Deferred proofs for maximum cardinality matching}\label{app:helpermcm}

\subsection{MCM preliminaries}

\restategreedy*
\begin{proof}
	It is folklore that the greedy algorithm (adding edges one at a time to a matching, if and only if neither endpoint is already matched) is a $2$-approximation to the MCM, proving the approximation guarantee. Regarding the implementation, we can keep marking an array of vertex indices and check if a current edge violates a matched vertex, in $O(1)$ time per edge.
\end{proof}

\restatecansearch*
\begin{proof}
	Let $x^*$ be the maximum matching with all overflow placed on the extra edge in $\tG$, i.e.\ $x^*_E$ is the maximum matching on the original graph, and $x^*$ puts a value of $2M - M^*$ on the extra edge. It is clear by definition that $\tfrac{1}{2M} x^*$ is feasible for the objective \eqref{eq:matchminimax}, has $\inprod{\1_E}{x^*} = M^*$, and has no overflow (i.e.\ $\ma^\top x - \1_V$ is entrywise nonpositive). Hence, for any $y \in [0, 1]^{\tV}$,
	\begin{equation}\label{eq:choosexstar}2y^\top\Par{\mb^\top x^* - \1_V} - \inprod{\1_E}{ x^*} \le -M^*. \end{equation}
	The definition of duality gap, as well as $\ma^\top x - \1_V$ being entirely zero in coordinates of $\tV\setminus V$, imply
	\begin{equation}\label{eq:dualitygap}\begin{aligned}
	\max_{y' \in [0, 1]^{V}} 2(y')^\top \Par{\mb^\top \tx - \1} - \inprod{\1_E}{\tx} &\le \min_{x^* \in 2M\Delta^n} 2y^\top\Par{\mb^\top x^* - \1_V} - \inprod{\1_E}{x^*} + \eps M \\
	&\le -M^* + \eps M.
	\end{aligned}\end{equation}
	The second inequality used \eqref{eq:choosexstar}. Now, by $M \le M^*$ and $\max_{y' \in [0, 1]^{V}} (y')^\top \Par{\mb^\top \tx - \1} \ge 0$ (since we can always choose the all-zeroes vector), the second conclusion in \eqref{eq:mcm_guarantees} is immediate. To see the first, call $F$ the overflow of $\tx$, i.e.\
	\[F \defeq \max_{y' \in [0, 1]^{V}} (y')^\top \Par{\mb^\top \tx - \1} = \sum_{j \in V} \max\Par{\Brack{\mb^\top \tx}_j - 1, 0}.\]
	By definition, we must have $\inprod{\1_E}{\tx} \le M^* + F$, since removing $F$ units makes the matching feasible, and hence it cannot have $\ell_1$ norm larger than $M^*$. Now, we conclude from \eqref{eq:dualitygap} the desired
	\[F - M^* \le 2F - \inprod{\1_E}{\tx} \le -M^* + \eps M \implies F \le \eps M \le \eps M^*.\]
\end{proof}

\subsection{Vertex size reduction}

In this section, we prove Proposition~\ref{prop:vertexreduce}. We first state our algorithm, a greedy iterative procedure which repeatedly adds a maximal matching from a maintained vertex set to the remaining vertices.

\begin{algorithm}
	\caption{$\RedSize(G, \eps)$}
	\begin{algorithmic}[1]\label{alg:red-MCM-Vsize}
		\STATE \textbf{Input:} Bipartite graph $G = (V, E)$ with MCM size $M^*$
		\STATE \textbf{Output:} Vertex subset $V' \subseteq V$, $|V'| = O(M^*\log(\eps^{-1}))$, $G[V']$ has MCM size $\ge (1 - \eps) M^*$
		\STATE $V_0 \gets$ vertices incident to some maximal matching in $G$
		\STATE $N \gets \lceil 3\log(\eps^{-1}) \rceil$
		\FOR{$i \in [N]$}
		\STATE $S \gets$ vertices in $V \setminus V_{i - 1}$ incident to some maximal matching between $V_{i - 1}$ and $V \setminus V_{i - 1}$
		\STATE $V_i \gets V_{i - 1} \cup S$
		\ENDFOR
		\RETURN $V' \gets V_N$
	\end{algorithmic}
\end{algorithm}

We require one key technical claim for analyzing Algorithm~\ref{alg:red-MCM-Vsize}. 

\begin{lemma}\label{lem:progressreduce}
Fix an iteration $0 \le i \le N - 1$ of Algorithm~\ref{alg:red-MCM-Vsize}, and suppose the MCM size of $G[V_i]$ is $M_i$. Then, the MCM size of $G[V_{i + 1}]$ increases by at least $\tfrac 1 3 (M^* - M_i)$.
\end{lemma}
\begin{proof}
The symmetric difference between the maximum matchings of $G[V_i]$ and $G$ can be decomposed into even cycles and paths (to see this, every vertex in the symmetric difference has degree at most $2$). Hence, ignoring all even cycles, there is a vertex-disjoint augmenting path decomposition of size $k \defeq M^* - M_i$; call these paths $P_1, P_2, \ldots P_k$. For each path $P_j$, $j \in [k]$, all of its edges except the first and last are contained in $G[V_i]$; call the first and last edges $(x_j, x'_j)$ and $(y_j, y'_j)$, where $x_j, y_j \in G[V_i]$. We consider a few cases, ``processing'' paths sequentially.
\begin{enumerate}
	\item If both $(x_j, x'_j)$ and $(y_j, y'_j)$ are included in the maximal matching between $V_i$ and $V \setminus V_i$, this increases the maximum matching size in $V_{i + 1}$ by one since we can use the augmenting path.
	\item If both $x_j$ and $y_j$ are matched to other vertices $x''_j \neq x'_j$ and $y''_j \neq y'_j$, we can still use the resulting augmenting path in $V_{i + 1}$. The other vertices $x''_j$, $y''_j$ can only remove two other later paths from consideration in the process, where we consider paths sequentially. Note that if $x''_j$ or $y''_j$ is not an endpoint from a later path, this can only help the argument here.
	\item If one of $x_j$ and $y_j$ is matched to its corresponding endpoint in $P_j$ but the other is not, we can use the augmenting path and remove one other path from consideration.
\end{enumerate}
Because we took a maximal matching between $V_i$ and $S \setminus V_i$, at least one of the above cases must happen, so we increase the MCM size by at least $\tfrac 1 3 k$.
\end{proof}
We conclude with our analysis of Algorithm~\ref{alg:red-MCM-Vsize}.

\restatevertexreduce*
\begin{proof}
The correctness follows immediately by applying Lemma~\ref{lem:progressreduce}, and the fact that our original maximal matching (contained in $G[V_0]$) is of size at least $\thalf M^*$. We now prove correctness. For the first pass, the complexity follows directly from Lemma~\ref{lem:greedy}. For all additional iterations of the algorithm, note that we can only add $O(M^*)$ additional vertices to the current set since it is a valid matching in the original graph, and the iteration can be implemented in a single pass analogously to the proof of Lemma~\ref{lem:greedy}. The space overhead is only the maintenance of the current set $V_i$.
\end{proof}

\subsection{Cycle cancelling in low space}\label{sec:cc}

In this section, we provide implementation of a data structure which proves the following claim.

\restatecc*

Specifically, we reduce its proof to demonstrating the existence of a \emph{bipartite cycle-cancelling oracle}.

\begin{definition}[Bipartite cycle-cancelling oracle]
We call $\oracle$ a \emph{bipartite cycle-cancelling oracle} (BCCO) if it is associated with a (weighted) bipartite graph $G = (V, E, w)$, and given any vector $x \in \R_{\ge 0}^E$ supported on at most $2n$ edges, $\oracle$ outputs a vector $\tx \in \R_{\ge 0}^E$ such that $\mb^\top x = \mb^\top \tx$ and $\inprod{w}{\tx} \ge \inprod{w}{x}$, so that $\tx$ is supported on at most $n$ edges. \end{definition}

\begin{lemma}\label{lem:oracleexists}
A BCCO $\oracle$ is implementable in $O(n)$ space and $O(n \log n)$ work.
\end{lemma}

\begin{proof}[Proof of Proposition~\ref{prop:cc}]
It suffices to divide the stream into chunks of $n$ edges, and repeatedly call $\oracle$. Specifically, we input the first $2n$ edges of the stream into $\oracle$ and produce a set of $n$ edges, then input these $n$ edges with the next $n$ edges of the stream, and so on. Inducting on the properties of $\oracle$, we obtain the value and feasibility guarantees on the final output. Regarding the space and work guarantees, it suffices to use $L \ge n$ and apply Lemma~\ref{lem:oracleexists}.
\end{proof}

The remainder of this section is devoted to proving Lemma~\ref{lem:oracleexists}. Specifically, we demonstrate how to use the link/cut tree data structure of \cite{SleatorT83, Tarjan83}, with a few small modifications, to provide the required oracle. In Section~\ref{sssec:lctree} we state some preliminaries on link/cut trees which are known from prior work (namely a set of supported operations), and we put these components together in Section~\ref{sssec:oracleproof} to give an implementation which proves Lemma~\ref{lem:oracleexists}.

\subsubsection{Link/cut tree description}\label{sssec:lctree}

The description in this section primarily follows the implementation of link/cut trees given in Chapter 5 of \cite{Tarjan83}, with a few terminology changes following the later lecture notes of \cite{Dem12}. It is primarily a summary of prior work, where we state the supported operations that we will use.

\textbf{Data structure representation.} The link/cut tree implementation of \cite{SleatorT83, Tarjan83} we will use is maintained as a forest, where each (undirected, rooted) tree in the forest is decomposed into ``preferred paths'' and each path is stored as a splay tree, along with a parent pointer to the neighbor of the preferred path endpoint closer to the root. Every edge in the forest has an associated value (which can change). We refer to the tree representations of preferred paths as ``auxiliary trees'' and the tree representation of all connected vertices as a ``main tree.'' The link/cut tree supports queries or modifications to main trees, including the following operations.
\begin{enumerate}
	\item $\Link(v, w, C)$: add an edge of value $C$ between $v$ and $w$ in the main tree (if they are in different trees, join them so that $w$ is the parent of $v$).
	\item $\Cut(v)$: remove the edge between $v$ and its parent in the main tree.
	\item $\ChangeRoot(r)$: change the root of the main tree to the vertex $r$.
	\item $\LCA(v, w)$: return the least common ancestor of $v$ and $w$ in the main tree.
	\item $\Min(v)$: return the smallest edge value between $v$ and its root in the main tree.
	\item $\Add(v, C)$: add $C$ to all edge values between $v$ and its root in the main tree.
	\item $\Sum(v)$: return the sum of all edge values between $v$ and its root in the main tree.
\end{enumerate}

For a more complete and formal description, we refer the reader to \cite{Tarjan83, Dem12}.

\textbf{Space complexity of link/cut trees.} Based on the representation described above, it is clear that as long as the amount of additional information we need to store in the auxiliary trees to implement path aggregation operations is a constant per vertex, the overall space complexity is $O(n)$ where $n$ is an upper bound on the number of vertices of all main trees. The only remaining space overhead is storing all auxiliary trees, and the parent pointers of preferred paths in main trees. We include this discussion because it is not explicitly stated in the source material.

\textbf{Work complexity of link/cut trees.} We state the guarantees of the link/cut tree in the following claim, which follows by the analysis of the works \cite{SleatorT83, Tarjan83}.

\begin{proposition}[Main result of \cite{SleatorT83, Tarjan83}]\label{prop:reducetopath}
The amoritized cost of $L$ calls to any of the operations $\Link$, $\Cut$, $\ChangeRoot$, $\LCA$, $\Min$, $\Add$, or $\Sum$ is bounded by $O(L\log n)$.
\end{proposition}

\subsubsection{Proof of Lemma~\ref{lem:oracleexists}}\label{sssec:oracleproof}

We first outline our approach, and then describe how to implement it using link/cut trees.

\textbf{Approach outline.} The main goal is to show how to preserve an acyclic matching $\tx$ so that $\mb^\top \tx = \mb^\top x$ and $\inprod{w}{\tx} \ge \inprod{w}{x}$ after processing every edge, in a total of $O(n\log n)$ work using the link/cut tree. Inductively it is clear that the output will satisfy all the requirements of $\oracle$.

We now describe how to process a single input edge $e = (u, v)$ where $u \in L$ and $v \in R$ are on opposite sides of the bipartition. First, if $u$ and $v$ are both not in the data structure, create a new main tree consisting of just this edge (with say $u$ as the root). If only one is in the data structure then create a new edge with value $x_e$ appropriately using $\Link$. If $u$ and $v$ are in different main trees, we can call $\Link(u, v, x_e)$. In these cases, no cycles are created and no edge values are changed.

The last case is when $u$ and $v$ are in the same main tree. The path between $u$ and $v$ in their main tree and the new edge $(u, v)$ forms a cycle. In order to preserve $\mb^\top x$, it suffices to alternate adding and subtracting some amount $C$ from edges along the cycle. In order to make sure $\inprod w x$ is monotone, it suffices to compute the alternating sum of weights along the cycle to pick a parity, i.e.\ whether we add from every odd edge in the cycle and subtract from every even edge, or vice versa (since all cycles are even length, one of these will increase the weight). Once we pick a parity, we compute the minimum value amongst all edges of that parity in the cycle and alternate adding and subtracting this value; this will result in the tree becoming acyclic and preserve all invariants.

We finally remark that it suffices to divide the support of the input $x$ into connected components, and remove cycles from each connected component. By designating a ``root vertex'' in each connected component and processing the edges in that component in BFS order from the root, it is clear that the tree corresponding to that component will never become disconnected (edges are only removed when a cycle is formed, and removing any edge on the cycle will not disconnect the tree). Thus, it suffices to discuss how to implement cycle cancelling for a single connected component.

\textbf{Implementation.} We discuss the implementation of cycle cancelling for a single connected component. It will use three link/cut trees, called $W$, $T_+$, and $T_-$; the topology of these three trees will always be the same (i.e.\ they are the same tree up to the edge values). Roughly speaking, $W$ will contain edge weights with alternating signs (e.g.\ each edge's weight will be signed $(-1)^{\text{depth of edge}}$), and will be used to determine the direction of cycle cancelling to preserve $\inprod{w}{\tx} \ge \inprod{w}{x}$. Further, $T_+$ and $T_-$ will each correctly contain half of the edge values of the current fractional matching (depending on their sign in $W$), and the other edge values will be set to a large quantity. They are used to compute the minimum (signed) edge value in cycles and to modify the fractional matching.

Consider processing a new edge $e = (u, v)$ with weight $w_e$ and value $x_e$. If this edge is in any of the non-cycle-creating cases described in the outline, i.e.\ its endpoints do not belong in the same tree with a path between them, then we add a copy to each of $W$, $T_+$, and $T_-$. Depending on the depth of the edge in $W$ (which we can check by looking at the sign of its parent edge), we either give it value $w_e$ or $-w_e$. If it is positive in $W$, then we give it value $x_e$ in $T_+$ and value $n^2 \norm{x}_\infty$ in $T_-$; otherwise, we swap these values. We choose the value $n^2 \norm{x}_\infty$ for one copy of the edge, so that it never becomes the minimum value edge returned by $\Min$ throughout the algorithm.

Finally, we handle the case where we need to remove a cycle (when $u$ and $v$ are in the same tree, and the edge is not in the tree). We start by determining which half of edges in the cycle we want to remove value from, and which we should add to. To do this, we query $r = \LCA(u, v)$ and $\ChangeRoot(r)$ on the tree in $W$. Then, we compute the alternating sums of weights in each half of edges in the cycle. This allows us to determine a direction to preserve $\inprod w \tx \ge \inprod w x$. We also call $\LCA(u, v)$ and $\ChangeRoot(r)$ for both copies of the main tree in $T$.

We next determine the amount we wish to alternatingly add and subtract along the $u$-to-$r$ and $v$-to-$r$ paths by calling $\Min$ on the appropriate tree, $T_-$ or $T_+$ (corresponding to the direction obtained from querying $W$). We then call $\Add$ on both copies of the main tree in $T$ for these paths, with the value obtained by the $\Min$ queries or its negation appropriately. This zeros one edge, which we $\Cut$ from all three copies of the tree. Finally, we revert to the original root in all three trees to maintain correctness of signs. Overall, the number of link/cut tree operations per edge in the stream is a constant, so Proposition~\ref{prop:reducetopath} bounds the total work by $O(n\log n)$.

\subsection{Rounding to feasible integral solutions}\label{ssec:taming}

We state our algorithm for proving Lemma~\ref{lem:feasiblecorrect}, and give its semi-streaming analysis in this section.

\begin{algorithm}
	\caption{$\Feasible(x, G)$}
	\begin{algorithmic}[1]\label{alg:feasible}
		\STATE \textbf{Input:} Bipartite graph $G = (V, E)$ with vertex partition $V = L \cup R$ and incidence matrix $\mb$, fractional $x \in \R^E_{\ge 0}$ satisfying $\norm{x}_1 = M$, $\sum_{j \in V} \max\Par{[\mb^\top x]_j - 1, 0} \le \eps M$
		\STATE \textbf{Output:} Fractional matching $\tx$ satisfying $\norm{\tx - x}_1 \le 2\eps M$, $\mb^\top \tx \le \1_n$
		\STATE Let $\mx \in \R^{L \times R}$ be a zero-padded reshaped $x$
		\STATE $\lsum \gets \mx\1_{\frac{n}{2}}$
		\STATE $\mx' \gets \diag{\min\Par{\tfrac{1}{\lsum}, \1_{\frac{n}{2}}}} \mx$ entrywise
		\STATE $\rsum \gets (\mx')^\top \1_{\frac{n}{2}} $
		\STATE $\tmx \gets \mx' \diag{\min\Par{\tfrac{1}{\rsum}, \1_{\frac{n}{2}}}}$
		\RETURN $\tx =$ vectorized $\tmx$
	\end{algorithmic}
\end{algorithm}

To briefly explain Algorithm~\ref{alg:feasible}, in Line 3, for each edge $e = (u, v) \in E$ where $u \in L$, $v \in R$, we index $\mx_{uv} = x_e$, and for every nonexistent edge we set the corresponding entry to 0; the vectorization in Line 8 undoes this process. We first scale the matching downward so that there is no overflow on vertices in $L$, and then scale the resulting matching so there is no overflow on vertices in $R$. We now prove a helper result, which shows that the resulting $\ell_1$ loss $\norm{\tx - x}_1$ is bounded, and that we can implement Algorithm~\ref{alg:feasible} in few semi-streaming passes and low space.

\begin{restatable}[Feasible rounding]{proposition}{restatefeasiblecorrect}\label{lem:feasiblecorrect}
For bipartite graph $G = (V, E)$ with incidence matrix $\mb$, Algorithm~\ref{alg:feasible} takes as input a fractional approximate matching $x \in \R^E_{\ge 0}$ satisfying $\norm{x}_1 = M$ and $\sum_{j \in V} \max([\mb^\top x]_j, 0) \le \eps M$, and produces $\tx$ satisfying $\norm{\tx - x}_1 \le 2\eps M$, $\mb^\top \tx \le \1_V$ with the same support. Moreover, supposing that $\norm{x}_0 = O(n)$, using $O(1)$ passes, $O(n)$ auxiliary memory, and $O(m)$ work, we can explicitly compute and store $\tx$.
\end{restatable}

\begin{proof}
	We will let the scaling vectors be defined by $\lscale \defeq \min\Par{\tfrac{1}{\lsum}, \1_{\frac{n}{2}}}$, $\rscale \defeq \min\Par{\tfrac{1}{\rsum}, \1_{\frac{n}{2}}}$; the representation $\tx_e = [\lscale]_u [\rscale]_v x_e$ for all $e = (u, v) \in L \times R$ follows immediately. Next, note that clearly $\tmx \le \mx' \le \mx$ entrywise, by positivity and that the entrywise rescalings are bounded above by $1$. The conclusion $\mb^\top \tx \le \1_n$ then follows, since left-scaling by $\lscale$ forces the marginals on $L$ to be bounded, and this boundedness is preserved upon right-scaling by $\rscale$ (which also bounds marginals on $R$). It is straightforward to see the output $\tx$ has the same support as $x$ given the operations. It remains to bound $\norm{\tx - x}_1$, and give a semi-streaming implementation.
	
	\textbf{Bounding the rounding loss.} We use shorthand $\norm{\mx}_1$ to mean the entrywise sum, and similarly define $\norm{\mx'}_1$, $\|\tmx\|_1$. We first bound $\norm{\mx}_1 - \norm{\mx'}_1$, which we note is the total overflow on $L$:
	\begin{align*}
	\norm{\mx}_1 - \norm{\mx'}_1 &= \sum_{u \in L} \max\Par{[\mb^\top x]_u - 1, 0} \\
	&\le \sum_{j \in V} \max\Par{[\mb^\top x]_j - 1, 0} \le \eps M.
	\end{align*}
	The last line followed by assumption. Similarly, to bound $\norm{\mx'}_1 - \|\tmx\|_1$, first define $x'$ to be the vectorized form of $\mx'$, and note that since $x' \le x$ entrywise, we have the overflow bound $ \sum_{j \in V} \max\Par{[\mb^\top x']_j - 1, 0} \le \eps M$. Thus, a similar calculation yields that
	\begin{align*}
	\norm{\mx'}_1 - \norm{\tmx}_1 &= \sum_{v \in R} \max\Par{[\mb^\top x']_v - 1, 0} \\
	&\le \sum_{j \in V} \max\Par{[\mb^\top x']_j - 1, 0} \le \eps M.
	\end{align*}
	Combining these calculations and using entrywise monotonicity yields $\norm{\tx - x}_1 \le 2\eps M$.
	
	\textbf{Semi-streaming implementation.} It suffices to show how to compute $\lscale$ and $\rscale$ in one pass and $O(m)$ work each; clearly they can be stored in $O(n)$ space, and applied to $x$ to explicitly form $\tx$. Since the support size of $x$ is $O(n)$, we can store it explicitly and add appropriately to the marginals on $L$ to compute $\lscale$ as the edges are encountered in the stream. Similarly, once $\lscale$ is explicitly stored, we can compute and store $x'$ and use it to compute $\rscale$ in one pass.
\end{proof}

Now we can formally prove the postprocessing procedure to round the approximate fractional solution to a feasible integral matching, with guarantees as stated here.

\restateintegralcorrect*

\begin{proof}[Proof of Proposition~\ref{lem:integralcorrect}]

By guarantees of Proposition~\ref{prop:cc} and~\ref{lem:feasiblecorrect}, we can obtain a feasible solution $\tx'$ satisfying $\norm{\tx' - x}_1 \le 2\eps M$, $\mb^\top \tx' \le \1_V$, and that $\tx'$ has support corresponding to a forest on $G = (V,E)$ using $O(1)$ passes, $O(n)$ auxiliary memory, and $O(m)$ work. 

Now considering the forest supported on vertex set $V$, we know there is at least a good approximate fractional MCM solution $\tx'$ for the original problem. Now we can directly compute an MCM on each component as a tree by running a dynamic programming procedure for subtrees. This runs in a total time of $O(n)$. The MCM $\tx$ uses $O(n)$ memory to store, and by construction satisfies the property that $1^\top \tx' \le 1^\top \tx$, and that $\mb^\top \tx \le \1_V$, which proves the statement.
\end{proof}

\section{Sampling for rounding linear programming solutions}\label{sec:samplelp}

In this section, we give a general procedure for rounding a fractional solution which is returned by our algorithm for box-simplex games in Section~\ref{sec:value}. This procedure applies to general box-simplex problems, beyond those with combinatorial structure, so we include this section for completeness. In particular, directly applying the sparsity bounds of this section to our matching problems directly imply $\tO(n \cdot \text{poly}(\eps^{-1}))$ space bounds for the various matching-related applications in the paper, up to width parameters, which roughly match our strongest results up to the $\eps^{-1}$ dependence. As an example, we give an application of this technique to MCM at the end of this section.

 Given streaming access to an approximate fractional solution $x$ on the simplex (via an implicit representation), a natural way of constructing a low-space approximate solution is to randomly sample each entry of $x_i$ and reweight to preserve expected objective value; this is the rounding strategy we analyze. We use the following Algorithm~\ref{alg:randsample}, parameterized by some prescribed $\{M_i\}_{i \in [m]}$.

\begin{algorithm}
	\caption{$\RandomSample(x, K, \{M_i\}_{i \in [m]})$}
	\begin{algorithmic}[1]\label{alg:randsample}
		\STATE \textbf{Input:} Coordinates of $x \in \R_{\ge 0}^m$ in streaming fashion, sample count $K$, parameters $\{M_i\}_{i \in [m]}$
		\FOR{$i\in[m]$}
		\STATE Draw $K$ independent random variables $\{X_i^k\}_{k \in [K]}$, where 
		\[X_i^k = \begin{cases}M_i & \text{with probability } \tfrac{x_i}{M_i}\\ 0 &\text{otherwise}\end{cases}.\]
		\STATE $\hat{x}_i \gets \tfrac{1}{K}\sum_{k\in[K]}X_i^k$.
		\ENDFOR
		\RETURN{$\hat{x}$}.
	\end{algorithmic}
\end{algorithm}

\subsection{Concentration bounds}

For proofs in this section, as well as later, we crucially rely on well-known concentration properties of bounded random variables. We state here a few facts used repeatedly throughout.

\begin{proposition}[Chernoff bound]\label{prop:cher-down}
	For $K$ independent scaled Bernoulli random variables $\{X_k\}_{k\in[K]}$ satisfying $X_k = N_k$ with probability $p_k$, $0<N_k \le 1$ for all $k \in [K]$, and all $0<\delta<1$,
	\[
	\Pr\left(\left|\sum_{k\in[K]}X_k-\sum_{k\in[K]}\E X_k\right| \ge \delta \sum_{k\in[K]}\E X_k\right)\le 2\exp\left(-\frac{\delta^2\sum_{k\in[K]}\E X_k}{3}\right)
	\]
\end{proposition}

We give a simple generalization of Proposition~\ref{prop:cher-down} to the case where the scaled Bernoulli variables are allowed to take on negative values.

\begin{corollary}[Generalized Chernoff bound]\label{cor:cher-down}
	For $K$ independent scaled Bernoulli random variables $\{X_k\}_{k\in[K]}$	satisfying $X_k= N_k$ with probability $p_k$, $0< |N_k| \le 1$ for all $k \in [K]$, and all $0<\delta<1$,
	\[
	\Pr\left(\left|\sum_{k\in[K]}X_k-\sum_{k\in[K]}\E X_k\right| \ge \delta \sum_{k\in[K]}\E |X_k|\right)\le 4\exp\left(-\frac{\delta^2\sum_{k\in[K]}\E |X_k|}{3}\right)
	\]
\end{corollary}
\begin{proof}
	Divide the set $[K] = \mathcal{K}^+\cup\mathcal{K}^-$, where we define $\mathcal{K}^+\defeq\{k\in[K]| N_k\ge0\}$ and $\mathcal{K}^+\defeq\{k\in[K]| N_k<0\}$. Applying Proposition~\ref{prop:cher-down} to $\sum_{k\in\mathcal{K}^+}X_k$ 	 and $\sum_{k\in\mathcal{K}^-}-X_k$ yields the result.
\end{proof}

Furthermore, the following one-sided Chernoff bound holds when $0<N_k\le 1$ and $\delta\ge 1$.

\begin{proposition}[One-sided Chernoff bound]\label{prop:cher-down-oneside}
	For $K$ independent scaled Bernoulli random variables $\{X_k\}_{k\in[K]}$ satisfying $X_k = N_k$ with probability $p_k$, $0<N_k \le 1$ for all $k \in [K]$, and all $\delta>0$,
	\[
	\Pr\left(\sum_{k\in[K]}X_k-\sum_{k\in[K]}\E X_k \ge \delta \sum_{k\in[K]}\E X_k\right)\le\exp\left(-\frac{\delta^2\sum_{k\in[K]}\E X_k}{2+\delta}\right)
	\]
\end{proposition}

\begin{proposition}[Bernstein's inequality]\label{prop:bernstein}
	For $K$ independent random variables $\{X_k\}_{k \in [K]}$	satisfying $|X_k|\le C$ with probability one, let $V = \sum_{k\in[K]} \Var[X_k]$. Then for all $t \ge 0$,
	\[
	\Pr\left(\left|\sum_{k\in[K]}X_k-\sum_{k\in[K]}\E X_k\right| \ge t\right) \le 2\exp\left(-\frac{t^2}{2V+2Ct/3}\right).
	\]
\end{proposition}

\subsection{Random sampling guarantees}

We first give a general guarantee on the approximation error incurred by random sampling via Algorithm~\ref{alg:randsample}. While the guarantees are a bit cumbersome to state, they become significantly simpler in applications. For instance, in all our applications, all binary random variables are scaled with $M_i$ such that  $\max_{i \in [n]} M_i \le 1$ (see Lemma~\ref{lem:gen-rounding} for definition), and the bounds on $\ma^\top \hx - \ma^\top x$ become standard multiplicative error approximations when the matrix $\ma$ is all-positive.

\begin{lemma}\label{lem:gen-rounding}
	Consider an instance of problem \eqref{eq:boxsimplex} parameterized by $\ma$, $b$, $c$. For some $x \in \Delta^m$ whose coordinates can be computed in streaming fashion, define for all $j \in [n]$,
	\[B_j = \Brack{\left|\ma\right|^\top x}_j.\] 
	Let $\hat{x}$ be the output of Algorithm~\ref{alg:randsample} on input $x$ with 
	\[M_{i} = \min_{j\in [n]}\frac{B_j}{|\ma_{ij}|}\text{ and } K= \frac{12\log(mn)}{\eps^2}.\]
	With probability at least $1 - (mn)^{-1}$, $\hat{x}$ satisfies the following properties for $B \defeq \max_{i\in[m]} M_i$:
	\begin{align*}
	\left|\Brack{\ma^\top \hat{x}-\ma^\top x}_j\right|& \le \eps B_j \text{ for all } j \in [n], \quad |\norm{\hat{x}}_1 - \norm{x}_1| \le \eps \max(1,B), \\
	\left|c^\top\hat{x}-c^\top x\right| & \le \eps\norm{c}_\infty \max(1,B),\\
	\|\hat{x}\|_0 & =  O\left(\left(1 + \sum_{i \in [m]} x_i \max_{j\in[n]}\frac{|\ma_{ij}|}{B_j}\right)\cdot\frac{\log(mn)}{\eps^2}\right).
	\end{align*}
\end{lemma}

\begin{proof}
	First note that it is immediate by definition of the $B_j$ to see that all $\tfrac{x_i}{M_i}=\max_{j\in[n]}\tfrac{|\ma_{ij}|x_i}{B_j}\le 1$ are valid sampling probabilities for all $i\in[m]$. Also, recall that entrywise
	\[\hat{x}_i = \frac{1}{K}\sum_{k \in [K]} X_i^k.\]	
	We show the first property; fix some $j \in [n]$ and consider $[\ma^\top \hx]_j - [\ma^\top x]_j $. Applying Corollary~\ref{cor:cher-down} to the random variables $\{\tfrac{1}{B_j}\ma_{ij}X^k_{i}\}_{i \in [m], k \in [K]}$ with $\delta = \eps$, we see by definition of $K$ that
	\begin{align*}
	\sum_{i \in \mathcal[m], k \in [K]} \frac{1}{B_j}\ma_{ij}X^k_{i} = \frac{K}{B_j}  \Brack{\ma^\top \hx}_j, & \\
	\sum_{i \in \mathcal[m], k \in [K]} \E\Brack{\frac{1}{B_j} \ma_{ij}X_{i}^k} = \frac{K}{B_j} [\ma^\top x]_j, & \;\\ \sum_{i \in \mathcal[m], k \in [K]} \E\left|\frac{1}{B_j} \ma_{ij}X_{i}^k\right| = \frac{K}{B_j} [|\ma|^\top x]_j,& \\
	\implies \Pr\left(\frac{K}{B_j}\left|[\ma^\top \hx]_j-[\ma^\top x]_j\right| \ge \delta \frac{K}{B_j}[|\ma|^\top x]_j\right)
	& \le 4\exp\left(-\frac{\delta^2 K[|\ma|^\top x]_j}{3B_j}\right)\le \frac{4}{(nm)^4},
	\end{align*}
	where for the last inequality we use definitions of $\delta$, $K$, and that $B_j= [|\ma|^\top x]_j$, for all $j\in[n]$.

	This conclusion for a coordinate $j \in [n]$ is equivalent to $\left|[\ma^\top \hx]_j - [\ma^\top x]_j \right| \le \eps B_j$. Union bounding over all $j \in [n]$, we thus have with probability at least $1-\tfrac{4}{n^3m^4}$ 
	\[\left|\left[\ma^\top \hx - \ma^\top x\right]_j\right|\le \eps B_j,\quad\forall j\in[n].\]
	Now we show the second and third properties. 
	Given $B\defeq\max_{i\in[m]}M_i$, we first apply Proposition~\ref{prop:bernstein} on the sum $\sum_{i\in[m]}K c_{i}\hat{x}_{i} = \sum_{i \in [m], k \in [K]} c_i X_i^{k}$, using the bound
	\begin{align*}\sum_{i \in [m], k \in [K]} \Var\Brack{c_{i} X_{i}^k} 
	&\le K\norm{c}_{\infty}^2 \sum_{i \in [m]} \Var\Brack{X_i^k} \\
	&\le K\norm{c}_{\infty}^2\sum_{i \in [m]} x_i M_i \le K\norm{c}_\infty^2 B.\end{align*}
	Thus, we can choose parameters $C = \norm{c}_{\infty}B$, $V = K\norm{c}_{\infty}^2B$ to obtain the following bounds for $K\ge \tfrac{12\log(mn)}{\eps^2}$ depending on whether $B>1$ or $B\le1$. For $B>1$, we have
	\begin{align*}
	\Pr\left(K\left|\sum_{i \in [m]} c_{i}\hx_{i}- \sum_{i \in [m]}c_{i} x_{i} \right|\ge \eps B\norm{c}_{\infty} K\right)&\le2\exp\left(-\frac{\eps^2B^2\norm{c}_{\infty}^2 K^2}{2K\norm{c}_{\infty}^2B+ 2\norm{c}_{\infty}^2B^2K\eps/3}\right) \\&\le \frac{1}{4(mn)^2}.
	\end{align*}
	
	For $B\le1$, we have
	\begin{align*}
	\Pr\left(K\left|\sum_{i\in[m]} c_{i}\hx_{i}- \sum_{i\in [m]}c_{i} x_{i} \right|\ge \eps\norm{c}_{\infty} K\right)&\le2\exp\left(-\frac{\eps^2\norm{c}_{\infty}^2 K^2}{2K\norm{c}_{\infty}^2B+ 2\norm{c}_{\infty}^2BK\eps/3}\right) \\&\le \frac{1}{4(mn)^2}.
	\end{align*}
	
	Altogether this implies the third conclusion, and the second conclusion follows as a special case when specifically picking $c=\1_n$. Each holds with probability $\ge 1-\tfrac{1}{4}(mn)^{-2}$.
	
	Finally, for all $i\in [m]$, let $Y_{i}=1$ if $\hat{x}_{i}\neq 0$ and $Y_{i} = 0$ otherwise. It is straightforward to see that $Y_{i}=1$ with probability 
	\begin{align*}
	1-\Par{1-\max_{j\in[n]}\frac{|\ma_{ij}|x_{i}}{B_j}}^{K} & \le K\max_{j\in[n]}\frac{|\ma_{ij}|x_{i}}{B_j}\\
	\implies \sum_{i \in [m]} \E[Y_{i}] \le K \sum_{i \in [m]} x_i \max_{j\in[n]}\frac{|\ma_{ij}|}{B_j} & \le K \max\left( \sum_{i \in [m]} x_i \max_{j\in[n]}\frac{|\ma_{ij}|}{B_j},1\right).
	\end{align*}
	
	Thus, by applying the one-sided Chernoff bound in Proposition~\ref{prop:cher-down-oneside} (where we consider the cases where $\sum_{i \in [m]} x_i \max_{j\in[n]}\frac{|\ma_{ij}|}{B_j}\ge1$ and $\sum_{i \in [m]} x_i \max_{j\in[n]}\frac{|\ma_{ij}|}{B_j}\le1$ separately), 
	\[
	\Pr\left( \sum_{i \in [m]} Y_{i} \ge 2K+2K\Par{\sum_{i \in [m]} x_i \max_{j\in[n]}\frac{|\ma_{ij}|}{B_j}}\right) \le \frac{1}{4(mn)^2}.
	\]
	Finally, applying a union bound, all desired events hold with probability $ 1-\tfrac{1}{mn}$.
\end{proof}

Note that in the semi-streaming model, one can choose $B_j \defeq \Brack{|\ma|^\top x}_j$ and compute all such values in one pass when $x$ is of the form in Lemma~\ref{lem:implicitx}; however, occasionally we will choose larger values of $B_j$ to obtain improved sparsity guarantees. Thus, we provide the following one-sided guarantee, which holds when $\ma_{ij}\ge 0$, $\forall i,j$ and $B_j\ge [\ma^\top x]_j$.

\begin{corollary}\label{cor:gen-rounding}
	Consider an instance of problem \eqref{eq:boxsimplex} parameterized by $\ma$ with $\ma_{ij}\ge0$ $\forall$ i,j, and $b$, $c$. For some $x \in \Delta^m$ whose coordinates can be computed in streaming fashion, suppose we have bounds for all $j \in [n]$,
	\[B_j \ge \Brack{\ma^\top x}_j.\] 
	Let $\hat{x}$ be the output of Algorithm~\ref{alg:randsample} on input $x$ with 
	\[M_{i} = \min_{j\in [n]}\frac{B_j}{|\ma_{ij}|}, \text{ and } K= \frac{12\log(mn)}{\eps^2}.\]
	With probability at least $1 - (mn)^{-1}$, $\hat{x}$ satisfies the following properties for $B \defeq \max_{i\in[m]} M_i$:
	\begin{align*}
	\Brack{\ma^\top \hat{x}-\ma^\top x}_j\le \eps B_j \text{ for all } j \in [n],& \quad |\norm{\hat{x}}_1 - \norm{x}_1| \le \eps \max(1,B), \\
	\left|c^\top\hat{x}-c^\top x\right| \le \eps\norm{c}_\infty \max(1,B),
	\quad \|\hat{x}\|_0 & =  O\left(\left(1 + \sum_{i \in [m]} x_i \max_{j\in[n]}\frac{|\ma_{ij}|}{B_j}\right)\cdot\frac{\log(mn)}{\eps^2}\right).
	\end{align*}
\end{corollary}

Note that the first and fourth properties of Corollary~\ref{cor:gen-rounding} are one-sided in that we only wish to upper bound a property of $\hx$. The proof is identical to Lemma~\ref{lem:gen-rounding}, except that we use Proposition~\ref{prop:cher-down-oneside} with $\delta = \eps\cdot\tfrac{B_j}{[|\ma|^\top x]_j}>0$ instead of Proposition~\ref{prop:cher-down} in cases where $\delta \ge 1$, which suffices for the one-sided bound of the first property. Similarly, when $\delta\ge1$ as in the proof of the fourth property, the one-sided Chernoff bound only helps concentration. Next, we give an end-to-end guarantee on turning Algorithm~\ref{alg:sherman} into a solver for the fractional problem, via applying Algorithm~\ref{alg:randsample} on the output.

\begin{lemma}\label{lem:gen-rounding-back}
	Given $(x,y)$, an $\tfrac \eps 2$-approximate saddle point of~\eqref{eq:boxsimplex}, let \[B_j =[|\ma|^\top x]_j,\quad B = \max_{i\in[m]}\min_{j\in [n]}\frac{B_j}{|\ma_{ij}|}\] and $\hx = \RandomSample(x, K, \{M_i\}_{i \in [m]})$ with
	\[M_i = \min_{j \in [n]} \frac{B_j }{|\ma_{ij}|},\text{ and } K =O\Par{ \frac{\log(mn)\Par{\Par{\norm{c}_\infty^2 + \norm{\ma}_\infty^2}(1+B)^2+\Par{\sum_{j\in[n]} B_j}^2}}{\eps^2}}.\]
	With probability $1-(mn)^{-1}$, $(\tfrac{\hx}{\|\hx\|_1},y)$ is an $\eps$-approximate saddle point to~\eqref{eq:boxsimplex}. Moreover, the total space complexity of Algorithm~\ref{alg:sherman} and Algorithm~\ref{alg:randsample} to compute and store the output $(\tfrac{\hx}{\|\hx\|_1}, y)$ is
	\[O\Par{\log n\Par{\sum_{i \in [m]} x_i \max_{j \in [n]} \frac{|\ma_{ij}|}{B_j}} \cdot \frac{\Par{\Par{\norm{c}_\infty^2 + \norm{\ma}_\infty^2}(1+B)^2+\Par{\sum_{j\in[n]} B_j}^2}}{\eps^2} + \frac{n\norm{\ma}_\infty}{\eps} }.\]
\end{lemma}

\begin{proof}
	Our choice of $K$ is with respect to an accuracy parameter on the order of  
	\[\frac{\eps}{(\norm{c}_\infty+\norm{\ma}_\infty) (1+B) + \sum_{j\in[n]} B_j}.\] Recall that an $\eps$-approximate saddle point to a convex-concave function $f$ satisfies
	\[
	\max_{\by \in \yset} f(x, \by) - \min_{\bx \in \xset} f(\bx, y) \le \eps.
	\]	
	For our output $(\tfrac{\hx}{\|\hx\|_1}, y)$, since we keep the same box variable $y$, it suffices to show that the side of the duality gap due to $x$ and $\tfrac{\hx}{\|\hx\|_1}$ does not change significantly. Namely, we wish to show
	\[ \max_{\by \in \yset} f\Par{\frac{\hx}{\|\hx\|_1}, \by} - \max_{\by \in \yset} f(x, \by) = \Par{c^\top \frac{\hx}{\|\hx\|_1} + \norm{\ma^\top \frac{\hx}{\|\hx\|_1} - b}_1} - \Par{c^\top x + \norm{\ma^\top x - b}_1} \le \frac{\eps}{2}.\]
	By the triangle inequality, it equivalently suffices to show that
	\begin{align*}c^\top \hx - c^\top x \le \frac{\eps}{8},\; c^\top \hx \Par{\frac{1}{\norm{\hx}_1 } - 1} \le \frac{\eps}{8},\\
	\norm{\ma^\top\hx - \ma^\top x}_1 \le \frac{\eps}{8},\; \norm{\ma^\top \hx}_1 \left|\frac{1}{\norm{\hx}_1 } - 1 \right| \le \frac{\eps}{8}.\end{align*}
	The first and third conclusions hold by applying Lemma~\ref{lem:gen-rounding} for the choice of $K$. The second and fourth hold by $|c^\top \hx|\le \norm{c}_\infty\norm{\hx}_1$ and $\norm{\ma^\top\hx}_1\le \norm{\ma}_\infty\norm{\hx}_1$ and then applying Lemma~\ref{lem:gen-rounding} for the choice of $K$. Finally, the desired sparsity follows by combining the space bound of the output (via Lemma~\ref{lem:gen-rounding}) and the space complexity of implicitly representing the average iterate (via Corollary~\ref{corr:barxrep}).
\end{proof}

\subsection{Application: rounding MCM solutions}

We give a simple application of our random sampling framework to computing an explicit low-space approximate MCM. While the space complexity does not match our strongest results based on a low-space cycle cancelling implementation, we hope it is a useful example of how to more generally sparsify fractional solutions of box-simplex games without explicit combinatorial structure.

Following the reductions of Section~\ref{sec:flowround}, we assume in this section that we have a simplex variable $x \in \Delta^m$ and a matching size $\bM \in [1, n]$ such that
\[\mb^\top (\bM x) \le \1_v,\; \bM\ge (1 - \eps)M^*,\]
where $M^*$ is the maximum matching size. We now show how to apply our random sampling procedure, Algorithm~\ref{alg:randsample}, to sparsify the support of the matching without significant loss.

\begin{corollary}\label{lem:MCM-rounding}
	Suppose for an MCM problem, $x \in \Delta^{E}$ satisfies $\mb^\top x \le \tfrac{1}{\bM}\1_V$, and $\eps \in (0, 1)$. Let $\hat{x}$ be the output of Algorithm~\ref{alg:randsample} on input $x$ with $m_{e} = \tfrac 1 \bM$ for all $e \in E$, and $K\defeq \tfrac{12\log n}{\eps^2}$.
	Then with probability at least $1 - n^{-2}$, the output $\hat{x}$ satisfies the following properties:
	\[
	\mb^\top \hat{x}\le 
	\frac{1+\eps}{\bM}\1_V, \quad |\norm{\hat{x}}_1 - 1| \le \eps,\quad \|\hat{x}\|_0 =  O\left(\frac{\bM\log n}{\eps^2}\right).
	\]
\end{corollary}

\begin{proof}
	It is straightforward to see that by the assumptions, we can take $\ma = \bM \mb$ and $B_v = 1$ for all $v \in V$, and choose the accuracy level to be $\eps$ in Lemma~\ref{lem:gen-rounding}. Thus, by the conclusions of Lemma~\ref{lem:gen-rounding}, noting that $M_e = \tfrac{1}{\bM} \le 1$ holds for all $e\in E$, we conclude that with probability at least $1 - n^{-2}$, all desired guarantees hold:
	\begin{align*}
	& \bM\mb^\top \hat{x}\le \bM\mb^\top x + \eps\1_V\le(1+\eps)\1_V,\\
	& |\norm{\hat{x}}_1 - 1| \le \eps,\\
	& \|\hat{x}\|_0 =  O\left(\left(\sum_{e \in E} x_e \max_{v\in V} \frac{|\ma_{ev}|}{B_v}\right)\cdot\frac{\log(mn)}{\eps^2}\right) = O\left(\frac{\bM\log n}{\eps^2}\right).
	\end{align*}
\end{proof} %
\section{Approximate MCM via box-constrained Newton's method}
\label{sec:box}

\newcommand{\arunsopt}{M^*}
\newcommand{\xopt}{x^{\star}}
\newcommand{\vopt}{v^{\star}}
\newcommand{\xhat}{\hat{x}}
\newcommand{\veps}{\varepsilon}
\newcommand{\xtilde}{\widetilde{x}}
\newcommand{\vtilde}{\widetilde{v}}
\newcommand{\mlap}{\mathcal{L}}
The goal of this section is to prove Theorem~\ref{thm:boxnewton_alg}. In particular, we give an alternate second-order method to compute $(1-\eps)$-approximate maximum matchings in unweighted graphs using techniques developed for solving matrix scaling and balancing problems.

\begin{restatable}{theorem}{boxconstrained}\label{thm:boxnewton_alg}
For a MCM problem on bipartite $G = (V,E)$ with $|V| =n$, $|E| = m$ and optimal value $M^*$, Algorithm~\ref{alg:boxconstrained} with parameter\footnote{This lower bound on $\eps$ is without loss of generality as otherwise we can use the algorithm in Section~\ref{ssec:exactMCM} which computes an exact MCM with a larger value of $\eps$.} $\eps \in [ \Theta\left(\frac{\log(mn)}{n}\right), \frac{1}{2}]$ obtains a matching of size at least $(1-\eps)M^*$ using $\tO(n)$ space, $\tO(\eps^{-1})$ passes, and $\tO(m)$ work per pass. 
\end{restatable}

We first recall the dual (vertex cover) formulation of the standard bipartite matching linear program
\begin{equation}
\label{eqn:exactLP}
\min_{v \geq 0, \\ \mb v \geq \1} \1^\top v,
\end{equation}
where $\mb$ is the unoriented incidence matrix. Our method (which is based on the box-constrained Newton's method of \cite{CMTV17}) will use a relaxation of this linear program which makes use of oriented incidence matrices and Laplacians, which we now define.

\begin{definition}
Let $G = (V,E,w)$ be a weighted undirected bipartite graph with bipartition $L, R$ and nonnegative edge weights $w$. We define the \emph{oriented} incidence matrix $\tmb  \in \R^{E \times V}$ as the matrix where for any edge $e = (i,j), i \in L, j \in R$ the row corresponding to $e$ in $\tmb$ has $\tmb_{ei} = 1$, $\tmb_{ej} = -1$, and all other entries set to $0$. We additionally define the graph Laplacian $\mlap_G =  \tmb^\top \diag{w} \tmb$. 
\end{definition}

When the graph is obvious we drop the subscript from $\mlap_G$. The orientation of the edges in $\tmb$ are typically chosen to be arbitrary in the literature, but we specify edge orientation as our algorithm will distinguish the vertices in $L$ and $R$, and denote a vector $x$ on these vertices as $x_L$ and $x_R$ (or sometimes $[x]_L$, $[x]_R$ for clarity). We will first describe our regularization scheme for this LP and prove that approximate minimizers for the regularized objective yield approximate fractional vertex covers. We then prove some stability properties on the Hessian of our regularized objective and show that a second-order method can be implemented in $\tO(n)$ space and $\tO(\eps^{-1})$ passes.

\begin{lemma}[Properties of regularized vertex cover]
\label{lem:VC-relaxed}
Let $G$ be an unweighted bipartite graph with $n$ nodes and $m$ edges, and let $L,R$ be the vertices on either side of its bipartition. Let $\arunsopt$ be the size of the maximum matching in $G$. Let $\veps \geq \frac{8 \log(mn)}{n}$ be a parameter, and set $\mu = \frac{\veps}{4 \log(mn)}$. Consider 
\begin{equation}
    \label{eqn:approxLP}
    f_{\mu} (x) \defeq \1_L^\top x_L - \1_R^\top x_R + \mu \left( \sum_{i \in [m]} e^{\frac{1}{\mu} (1 - [\tmb x]_i )} + \sum_{j \in L} e^{-\frac{1}{\mu} (x_j +\tfrac \veps 4)} + \sum_{j \in R} e^{\frac{1}{\mu} (x_j - \tfrac \veps 4)} \right).
\end{equation}
$f_{\mu}$ has the following properties.
\begin{itemize}
    \item For any $x$ with $f_{\mu}(x) <m$ we have $\tmb x \geq (1-\tfrac \veps 2)\1$, $x_L \geq - \frac{\veps}{2}$, and $x_R \leq \frac{\veps}{2}$. 
    \item If $\vopt$ is a feasible minimizer to~\ref{eqn:exactLP}, $x \defeq (1+\tfrac \veps 2) [\vopt_L, -\vopt_R]$ has $f_\mu(x) \leq (1+ \tfrac{2\veps}{3}) \arunsopt$.
    \item $f_\mu(x) \geq \arunsopt - \veps n$ for any $x$.
    \item For any $x$ with $f_{\mu}(x)< m$, let $x' = [x'_L, x'_R]$ where $x'_L = \min \{x_L, 2 \cdot \1_L \}$ and $x'_R = \max \{ x_R, -2 \cdot\1_R \}$. Then $f_\mu(x') \leq f_\mu(x)$.%
    
\end{itemize}
\end{lemma}
\begin{proof}
We prove the claims in order. For the first claim, let $v = [x_L, -x_R]$ and let $i$ be the index of the smallest entry of $v$. If $v_i \leq -\tfrac \veps 2$, then since $v_i + \tfrac \veps 4 \le \tfrac{v_i}{2}$ and $\exp$ is always nonnegative,
\[
f_{\mu}(x) \geq \1_L^\top x_L - \1_R^\top x_R +\mu \exp \left(- \frac{v_i}{2\mu}\right) \geq n v_i + \mu \exp \left(- \frac{v_i}{2\mu}\right) \geq m.
\]
The second inequality followed from the definition of $v_i$ and the third from monotonicity of the expression in $v_i$. Thus $v \ge -\tfrac \veps 2 \1$, implying the bounds on $x_L$ and $x_R$. This also implies that $\1_L^\top x_L - \1_R^\top x_R \geq -\tfrac{\veps n}{2}$, so it is straightforward to verify that if some $1 - [\tmb x]_i \ge \tfrac \veps 2$, then we would have $f_\mu(x) \ge -\tfrac{\veps n}{2} + (mn)^2 \mu \ge  m$, completing the first claim.

Next, since $\vopt$ is a feasible minimizer to \eqref{eqn:exactLP}, $\1^\top \vopt = \arunsopt$, $\vopt \geq 0$, and $\mb \vopt \geq 1$. Therefore by construction, $\1_L^\top x_L - \1_R^\top x_R = (1+\tfrac \veps 2) \arunsopt$, $x_L \geq 0$, $x_R \leq 0$, and $\tmb x = (1+\tfrac \veps 2)  \mb \vopt \geq (1+ \tfrac \veps 2)$. Plugging these into $f_\mu(x)$ yields (since $M^* \ge 1$)
\begin{align*}
f_{\mu}(x) &\leq \Par{1 + \frac \veps 2} \arunsopt + \mu \left( m \exp{\left(\frac{-\veps}{2 \mu} \right)} + n \exp{\left( \frac{-\veps}{4 \mu} \right)} \right) \\
&\leq \Par{1 + \frac \veps 2} \arunsopt + \mu \frac{m + n}{mn} \le \Par{1 + \frac{2\veps}{3}} M^*.
\end{align*}

For the third claim, suppose $x$ satisfied $f_\mu(x) \leq \arunsopt - \veps n$. As $\arunsopt \leq n$ we may apply the first claim so $x_L \geq -\tfrac \veps 2$, $x_R \leq \tfrac \veps 2$, and $\tmb x \geq \1 - \tfrac \veps 2$. Let $v = [x_L + \tfrac \veps 2 \1, -x_R+ \tfrac \veps 2 \1]$, so $v \geq 0$ and $\mb v \geq \1$. This is a contradiction, as $v$ is feasible for \eqref{eqn:exactLP}, so $\arunsopt \leq \1^\top v = \1_L^\top x_L - \1_R^\top x_R + \frac{\veps n}{2} < f_{\mu}(x) + \frac{\veps n}{2} < \arunsopt$.

Finally, note that $f_{\mu}(x) < m$ combined with the first property of $f_{\mu}$ implies that $\tmb x \geq (1-\frac{\veps}{2}) \1$, $x_L \geq -\frac{\veps}{2}$, and $x_R \leq \frac{\veps}{2}$. Assume $x$ has $[x_L]_i = \alpha > 2$ for some $i$: we will show that decreasing this coordinate to $2$ decreases $f_\mu$. The only terms affected by changing $[x_L]_i$ are $\1^\top x_L$ (which decreases by $\alpha -2$), $\mu \sum_{j \in L} e^{-\frac{1}{\mu} (x_j + \frac{\veps}{4} )}$, and $\mu \sum_{i \in [m]} e^{\frac{1}{\mu} (1 - [\tmb x]_i)}$. The first of these sums increases by at most 
\[
\mu \exp\left( -\frac{2 + \frac \veps 4}{\mu} \right) - \mu \exp\left( -\frac{\alpha + \frac \veps 4}{\mu} \right) \leq (\alpha - 2) e^{\frac 2\mu} \leq \frac{\alpha - 2}{mn},
\]
while each of the $m$ terms in the second sum increases by at most
\[
\mu \exp\left( -\frac{-1 - [x_R]_j}{\mu} \right) - \mu \exp\left( -\frac{1-\alpha -[x_R]_j }{\mu} \right) \leq (\alpha - 2) e^{\frac{1-\veps}{\mu}} \leq \frac{\alpha - 2}{mn},
\]
for some $[x_R]_j$: we use our upper bound on $x_R$ here. Thus the total change in $f_{\mu}(x)$ is at most $-(\alpha - 2) \left(1 - \frac{1}{mn} - \frac{1}{n} \right) < 0$. A similar analysis gives the analogous claim for coordinates of $x_R$.

\end{proof}

We next show that derivatives of $f$ satisfy some useful stability properties, which we now define.

\begin{definition}
\label{def:SOR}
Let $f$ be a convex function. We say that $f$ is $r$-second order robust if for any vectors $x,y$ where $\norm{x-y}_\infty \leq r$ we have (where $\preceq$ is the Loewner order on the positive semidefinite cone)
\[
e^{-1} \nabla^2 f(x) \preceq \nabla^2 f(y) \preceq e \nabla^2 f(x).
\]
\end{definition}

\begin{lemma}
\label{lem:VC-derivatives}
The gradient and Hessian of $f_{\mu}$ defined in \eqref{eqn:approxLP} are
\[
\nabla f_{\mu}(x) = \begin{bmatrix}
\1_L - z_L\\ 
-\1_R + z_R
\end{bmatrix} - \tmb^\top y,\;
\nabla^2 f_{\mu}(x) = \frac{1}{\mu} \left( \diag{z} + \tmb^\top~\diag{y} \tmb \right)
\]
where $y \defeq \exp\left(\frac{1}{\mu} \left( \1 - \tmb x \right) \right)$, $z_L \defeq \exp\left(-\frac{1}{\mu} (x_L + \tfrac \veps 4\1) \right)$, $z_R \defeq  \exp\left(\frac{1}{\mu} (x_R- \tfrac \veps 4\1) \right)$, and $z \defeq [z_L ,z_R]$. For all $x$,  $\nabla^2 f_{\mu}(x)$ is a nonnegative diagonal matrix plus a weighted graph Laplacian matrix. Further, $f_\mu$ is convex and $\frac{\mu}{2}$-second order robust. Finally, for any $x,\delta$ we have
\[
\delta^\top \nabla^2 f_{\mu}(x) \delta \leq \frac{2}{\mu} \cdot  \delta^{\top} \left( \diag{\mb^\top y} + \diag{z} \right) \delta.
\]
\end{lemma}
\begin{proof}
The gradient and Hessian claims may be verified directly, and it is clear that $\tfrac z \mu$ is nonnegative and $\tmb^\top \diag{\tfrac y \mu}\tmb$ is a weighted graph Laplacian. Convexity of $f_\mu$ follows from $\nabla^2 f_\mu(x) \succeq 0$ everywhere. To prove second order robustness, for any $x, x'$ with $\norm{x-x'}_\infty \leq \frac{\mu}{2}$,
\[\norm{\tmb x - \tmb x'}_\infty \leq \norm{\tmb}_{\infty} \norm{x-x'}_{\infty} \leq 2 \cdot \frac{\mu}{2} \leq \mu.\] Thus every matrix forming the Hessian of $f_{\mu}$ changes by at most a factor of $e$ multiplicatively, and so $f_{\mu}$ is second order robust. For the final claim, let $|\delta|$, $\delta^2$ be applied entrywise. Then,
\begin{align*}
\delta^\top \tmb^\top \diag{y} \tmb \delta &= \sum_{i \in [m]} y_i [\tmb \delta]_i^2 \leq  \sum_{i \in [m]} y_i [\mb |\delta|]_i^2 \\
&\leq \sum_{i \in [m]} y_i [\mb \delta^2]_i [\mb \1]_i =  2 \sum_{i \in [m]} y_i [\mb \delta^2]_i = 2 \delta^\top \diag{\mb^\top y} \delta.
\end{align*}
The second line used Cauchy-Schwarz and $\mb \1 = 2 \cdot\1$. Since $\diag{z} \succeq 0$, this yields the claim.
\end{proof}

It remains to implement a low-pass and low-space optimization procedure to minimize $f_{\mu}$. We make use of a variant of the box-constrained Newton's method of \cite{CMTV17} stated below.
\begin{definition}
We say that a procedure $\oracle$ is an $(r,k)$-oracle for a class of matrices $\mathcal{M}$ if, for any $\ma \in \mathcal{M}$ and vector $b$, $\oracle(\ma,b)$ returns a vector $z$ where $\norm{z}_\infty \leq rk$ and
\[
z^\top b + \frac{1}{2} z^\top \ma z \leq \frac{1}{2} \left( \min_{\norm{y}_\infty \leq r} y^\top b + \frac{1}{2} y^\top \ma y \right).
\]
\end{definition}

The proof of the following is patterned from \cite{CMTV17}, but tolerates error in the Hessian.

\begin{lemma}
Let $f$ be $r$-second order robust with minimizer $\xopt$. For some $x$, let $\mm$ be such that $\frac{1}{2} \nabla^2 f(x)\preceq \mm \preceq 2 \nabla^2 f(x)$. Let $\oracle$ be an $(r,k)$-oracle for a class of matrices $\mathcal{M}$ where $\mm\in\mathcal{M}$. Then if $\max \{ \norm{\xopt}_\infty , \norm{x}_\infty \} \leq R$ and $R \geq r$, for any vector $x$, $x' = x + \frac{1}{k} \oracle\left(\frac{1}{k} \nabla f(x), \frac{2e}{k^2} \mm\right)$ satisfies
\[
f(x') - f(\xopt) \leq \left(1 - \frac{r}{160 k R}\right) (f(x) - f(\xopt)).
\]

\label{lem:boxnewton}
\end{lemma}
\begin{proof}
We note that for any $y$ where $\norm{y-x}_\infty \leq r$ we have
\[
f(y) = f(x) + \nabla f(x)^\top (y-x) + \frac{1}{2} \int_{0}^1 (y-x)^\top \nabla^2 f(\alpha x + (1-\alpha)y) (y-x) d \alpha
\]
and thus by the second order robustness of $f$,
\begin{equation}\label{eq:ulbounds}\begin{aligned}
f(x) + \nabla f(x)^\top (y-x) + \frac{1}{4e} (y-x)^\top \mm (y-x) \le f(y),\\ f(y) \leq f(x) + \nabla f(x)^\top (y-x) + e (y-x)^\top \mm (y-x).\end{aligned}\end{equation}
Let
\[
\hat{\delta} = \argmin_{\norm{\delta}_\infty \leq r} \frac{1}{k} \nabla f(x)^\top \delta + \frac{e}{k^2} \delta^\top \mm \delta.
\]
Also, let $\Delta = \oracle(\frac{1}{k} \nabla f(x), \frac{2 e}{k^2} \mm)$. By definition of $\oracle$ we have $\norm{\Delta}_{\infty} \leq rk$ and 
\[
\frac{1}{k}\nabla f(x)^\top \Delta + \frac{e }{k^2} \Delta^\top \mm \Delta \leq \frac{1}{2} \left( \frac{1}{k} \nabla f(x)^\top \hat{\delta} + \frac{e }{k^2} \hat{\delta}^\top \mm \hat{\delta} \right).
\]
Thus since $x' = x + \frac{1}{k} \Delta$ we see $\norm{x' - x}_\infty \leq r$, and hence
\begin{equation}\label{eq:xprogress}
f(x') \leq f(x) + \frac{1}{k} \nabla f(x)^\top \Delta + \frac{e}{k^2} \Delta^\top \mm \Delta \leq f(x) + \frac{1}{2} \left( \frac{1}{k} \nabla f(x)^\top \hat{\delta} + \frac{e}{k^2} \hat{\delta}^\top \mm \hat{\delta} \right).
\end{equation}
Define $\tilde{x}=  \frac{r}{2R} (\xopt - x)$, we note that $\norm{\tilde{x}}_{\infty} \leq \frac{r}{2R} \left( \norm{\xopt}_\infty + \norm{x}_\infty \right) \leq r$. By the minimality of $\hat{\delta}$ we observe that $\hat{\delta}$ achieves a smaller value than $c \widetilde{x}$ for any $c \leq 1$. For the choice $c = \frac{1}{4e^2}$ this implies
\begin{align*}
\frac{1}{2} \left( \frac{1}{k} \nabla f(x)^\top \hat{\delta} + \frac{e}{k^2} \hat{\delta}^\top \mm \hat{\delta} \right) &\leq \frac{1}{2} \left( \frac{c}{k} \nabla f(x)^\top \tilde{x} + \frac{c^2 e}{k^2} \tilde{x}^\top \mm \tilde{x} \right) \\
&= \frac{1}{8e^2} \left( \frac{1}{k}\nabla f(x)^\top \tilde{x} + \frac{1}{4ek^2}  \tilde{x}^\top \mm \tilde{x} \right) \\
&\leq \frac{1}{8e^2} \left( f\left(x + \frac{1}{k}\tilde{x}\right) - f(x) \right) \\
&\leq -\frac{r}{16k e^2 R} \left( f(x) - f(\xopt) \right),
\end{align*}
where the third line used \eqref{eq:ulbounds} and the last used convexity of $f$. Plugging this into \eqref{eq:xprogress} yields
\[
f(x') - f(\xopt) \leq \left( 1 - \frac{r}{16k e^2 R} \right) \left( f(x) - f(\xopt) \right).
\]
\end{proof}

Next, the class of matrices which are Laplacians plus a nonnegative diagonal admit an efficient $\oracle$.

\begin{lemma}[Theorem 5.11, \cite{CMTV17}]
Let $\mathcal{M}$ be the class of matrices which consist of a Laplacian matrix plus a nonnegative diagonal. Let $\ma \in \mathcal{M}$ be a matrix with $m$ nonzero entries. There is an algorithm which runs in $\tO(m)$ time and space and implements an $(r,O(\log n))$-oracle for $\ma$.%
\end{lemma}

 We complement this with a known semi-streaming spectral approximation for Laplacians.
\begin{lemma}[Section 2.2, \cite{McGregor14}]\label{lem:lapsparsify}
Let $G = (V,E,w)$ be a weighted undirected graph, given as an insertion-only stream. There is an algorithm which for any $\veps \in (0,1)$ takes one pass and returns a graph $H$ with $O(\veps^{-2} n \log^3(n))$ edges such that $\mlap_{H} \preceq \mlap_{G} \preceq (1+\veps) \mlap_{H}$ using $\tO(m)$ work.
\end{lemma}

\newcommand{\BoxConstrainedVC}{\mathsf{BoxConstrainedVC}}
\begin{algorithm}
	\caption{$\BoxConstrainedVC(G, \veps, \oracle)$}
	\begin{algorithmic}[1]\label{alg:boxconstrained}
		\STATE \textbf{Input:} Bipartite graph $G = (V,E)$ with vertex partition $V = L \cup R$ and oriented edge-incidence matrix $\tmb$ given as a stream, $\veps > 0$, $(O(\frac{\veps}{\log(n)}),k)$-oracle $\oracle$ for symmetric diagonal dominant matrices with nonpositive off diagonal
		\STATE \textbf{Output:} $v$ fractional vertex cover for $G$
		\STATE $x_0 \gets [\1_L, -\1_R]$, $\mu = \frac{\veps}{4 \log(mn)}$, $T = O(\frac{\log \mu^{-1}}{\mu k})$
		\STATE $f_\mu(x) \defeq \1_L^\top x_L - \1_R^\top x_R + \mu \left( \sum_{i \in [m]} e^{\frac{1}{\mu} (1 - [\tmb x]_i )} + \sum_{j \in L} e^{-\frac{1}{\mu} (x_j + \frac \veps 4)} + \sum_{j \in R} e^{\frac{1}{\mu} (x_j - \frac \veps 4)} \right)$
		\FOR{$0 \le t < T$ }
		\STATE Compute $\nabla f_\mu(x_t)$ and $\mm$, a $2$-approximation of $\nabla^2 f_{\mu}(x_t)$ with Lemma~\ref{lem:lapsparsify}
		\STATE $x'_{t+1} = x_t + \frac{1}{k}\oracle\left(\frac{1}{k} \nabla f_\mu(x_t), \frac{2e}{k^2} \mm \right)$.
		\STATE $x_{t+1} = \max \left\{ \min \{ x'_{t+1}, 2\cdot \1 \}, -2\cdot \1 \right\}$ entrywise
		\ENDFOR
		\RETURN{$x_T$}.
	\end{algorithmic}
	\label{alg:boxconstrained}
\end{algorithm}

We now assemble these claims to prove the correctness of our algorithm.

\begin{proposition}
\label{prop:boxnewton}
Let $G = (V,E)$ be an unweighted bipartite graph with bipartition $L, R$. Given access to a $(O(\frac{\veps}{\log n}), k)$-oracle $\oracle$ for $\nabla^2 f_\mu$, Algorithm~\ref{alg:boxconstrained} computes $x$ so $v = [x_L + \frac{\veps}{2} \1_L, -x_R + \frac{\veps}{2} \1_R]$ is a feasible vertex cover of size $\arunsopt + \veps n$. For $y \defeq \exp\left( \frac{1}{\mu} (1 - \tmb x) \right)$, $\exists w$ with $\norm{w}_1 \leq \mu n$ so $y-w$ is a feasible matching with $\1^\top y \geq \arunsopt - 5 \veps n$. Algorithm~\ref{alg:boxconstrained} requires $\tO(n)$ auxiliary space and $O(\frac{k \log\mu^{-1}}{\mu})$ passes, each requiring $\tO(m)$ work, plus the work and space required by one $\oracle$ call.
\end{proposition}

As the space used by $\oracle$ can be reused between runs, the space overhead will be $\tO(n)$ throughout.
\begin{proof}
By Lemmas~\ref{lem:VC-relaxed} and \ref{lem:VC-derivatives}, $\oracle$ applies to a matrix family containing the Hessian $\nabla^2 f_\mu$. In addition, we see that $\nabla f_\mu$  may be computed in a single pass since it equals
\[
\nabla f_{\mu}(x) = \begin{bmatrix}
\1_L - z_L\\ 
-\1_R + z_R
\end{bmatrix} - \tmb^\top y
\]
for $z_L, z_R, y$ as in Lemma~\ref{lem:VC-derivatives}. The first term may be computed directly, while the second can be computed analogously to Lemma~\ref{lem:implicitx}. Further, we can obtain a $2$-spectral approximation of $\nabla^2 f_\mu(x)$ in semi-streaming fashion: we compute $z$ in one pass, and sparsify the Laplacian using Lemma~\ref{lem:lapsparsify} while computing $y$ coordinatewise in one pass. Thus in each iteration we perform one pass and one call to $\oracle$: as there are $O(\frac{k \log n}{\mu })$ iterations, the claimed space, pass, and work bounds follow. As the graph in Lemma~\ref{lem:lapsparsify} is sparse, its Laplacian can be explicitly stored in $\tO(n)$ space.

We now prove the correctness of the algorithm. Let $\xopt$ be the minimizer of $f_{\mu}$. By construction, $f_\mu(x_0) \le 2n \le m$ and function progress is monotone. We note that the point $x_t$ has $\norm{x_t}_\infty \leq 2$ by construction for all $t$, and that $\norm{\xopt}_\infty \leq 2$ by the fourth condition of Lemma~\ref{lem:VC-relaxed}. Further, by Lemma~\ref{lem:VC-derivatives} we see that is $r$-second order robust with $r = \frac{\mu}{2}$. By Lemma~\ref{lem:boxnewton} we obtain
\[
f_{\mu}(x'_{t+1}) - f_{\mu}(\xopt) \leq \left(1 - \frac{\mu}{640 k} \right) \left( f_{\mu}(x_{t}) - f_{\mu}(\xopt)\right)
\]
for any $t$. As the fourth condition of Lemma~\ref{lem:VC-relaxed} implies $f_\mu(x_{t+1}) \leq f_{\mu}(x'_{t+1})$, the final $x_T$ has
\[
f_\mu(x_T) - f_\mu(\xopt) \leq \left(1 - \frac{\mu}{640 k} \right)^T \left( f_\mu(x_0) - f_\mu(\xopt) \right) \leq \frac{\mu^3 n}{9e},
\]
where the second inequality uses  $f_\mu(x_0) \leq 2n$ and $f_\mu(\xopt) \geq 0$. Thus if $\arunsopt$ is the size of the minimum vertex cover of $G$, we obtain by Lemma~\ref{lem:VC-relaxed} that $f_{\mu}(x_T) \leq (1+\frac{2\veps}{3})\arunsopt + \frac{\mu^3 n}{9e}\leq \arunsopt + \frac{\veps}{2} n$, since $\arunsopt \leq \frac{n}{2}$ and $\mu \leq \veps$. To complete the proof, the first condition of Lemma~\ref{lem:VC-relaxed} holds implies $v = [[x_T]_L + \frac{\veps}{2} 1, -[x_T]_R  + \frac{\veps}{2} 1]$ is nonnegative with $\mb v \geq 1$ and $\1^\top v \leq \arunsopt + \frac{\veps}{2} n+ \frac{\veps}{2} n = M^* + \veps n$.

For the second claim, we observe that for any $\delta$ with $\norm{\delta}_\infty \leq \frac \mu 2$ and $\xhat = x_T + \delta$,
\begin{align*}
\frac{\mu^3 n}{9e}&\geq f_\mu(x_T) - f_\mu(\xopt) \geq f_{\mu}(x_T) -  f_{\mu}(\hat{x}) \\
&\geq - \nabla f_{\mu}(x_T)^\top \delta - e \delta^\top \nabla^2 f_\mu(x_T) \delta \\
&\geq - \nabla f_{\mu}(x_T)^\top \delta - \frac{2e}{\mu} \delta^\top \left(\diag{\mb^\top y} + \diag{z} \right) \delta,
\end{align*}

where $z = [z_L, z_R]$; the first line used optimality of $\xopt$, the second used second order robustness of $f_\mu$, and the third used Lemma~\ref{lem:VC-derivatives}. We choose $\delta = - \frac{\mu^2}{4 e} \textbf{sign}\left(\nabla f_\mu (x_T) \right)$: this satisfies $\norm{\delta}_\infty \le \tfrac{\mu}{2}$ so
\[
\frac{\mu^2}{4e} \norm{\nabla f_\mu(x_T) }_1  - \frac{\mu^3}{8e} \norm{ \mb^\top y + z}_1 \leq \frac{\mu^3 n}{9e}.
\]
Now note 
\begin{align*}
\norm{ \mb^\top y + z}_1 &= \norm{ [\tmb^\top y]_L + z_L}_1 + \norm{ -[\tmb^\top y]_R + z_R}_1 \\
&\leq \norm{ [\tmb^\top y]_L + z_L -  \1_L}_1 + \norm{ -[\tmb^\top y]_R + z_R - \1_R}_1 + n = \norm{\nabla f_{\mu}(x_T)}_1 + n. 
\end{align*}
Plugging this in to the above expression and rearranging, we obtain
\[
\frac{\mu^2}{4e} \norm{\nabla f_{\mu}(x_T)}_1 \leq  \frac{\mu^3 n}{9 e} + \frac{\mu^3 n}{8e} + \frac{\mu^3}{8e} \norm{\nabla f_{\mu}(x_T)}_1,
\]
so $\norm{\nabla f_{\mu}(x_T)}_1 \leq \mu n$. This implies $\exists w$ with $\norm{w}_1 \leq \mu n$ where $\1 + w=  \mb^\top y + z$. Since $z$ is nonnegative, $\mb^\top y \leq \1 + w$ with $\norm{w}_1 \leq \mu n$ as desired. We finally lower bound $\1^\top y$. Let $v' = [x_L, -x_R] - \mu \1$. Taking the inner product of $\1 + w=  \mb^\top y + z$ with $v'$ and rearranging gives
\[
\1^\top x_L - \1^\top x_R - \mu n - \1^\top y + w^\top v' = y^\top (\tmb x - 1) - \mu \1^\top y + x_L^\top z_L - x_R^\top z_R - \mu \1^\top z,
\]
or 
\[
f_{\mu}(x) - \mu n + w^\top v' = y^\top (\tmb x - 1) + \1^\top y + x_L^\top z_L - x_R^\top z_R.
\]
Next, $w^\top v' \geq - \norm{w}_1 \norm{v'}_\infty \geq - 2 \mu n$. Further,  
since $y = \exp\left(\frac{1}{\mu} (1 - \tmb x)\right)$, for any $i$ either $[\tmb x]_i -1 \leq \frac{\veps}{2}$ or $y_i \leq \exp\left(-\frac{1}{\mu} \frac{\veps}{2}\right) = \frac{\veps}{2mn}$. Thus  $y^\top (\tmb x - 1) \leq \veps \1^\top y + \frac{\veps}{2} \arunsopt$. Similarly, we obtain $x_L^\top z_L - x_R^\top z_R \leq \veps n$. Plugging these in, and using $\mu \le \tfrac \veps 4$ and $M^* \le n$ yields
\[
f_{\mu}(x) - 3 \mu n -\veps n - \frac{\veps}{2} \arunsopt \leq (1+\veps) \1^\top y \implies \1^\top y \geq \frac{1}{1+\veps} f_{\mu}(x) - 3 \veps n.
\]
Here we used $\mu \leq \frac{\veps}{4}$ and $\arunsopt \leq n$. The claim follows from $f_{\mu}(x) \geq \arunsopt - \veps n$ via Lemma~\ref{lem:VC-relaxed}.
\end{proof}

\newcommand{\ytilde}{\widetilde{y}}
\begin{proof}[Proof of Theorem~\ref{thm:boxnewton_alg}]
Let $\arunsopt$ be the size of the maximum matching in $G$. We first preprocess the graph $G$ using Proposition~\ref{prop:vertexreduce} within the pass, space, and work budgets to reduce to $n = O(M^*\log(\eps^{-1}))$. Applying Algorithm~\ref{alg:boxconstrained} to $\widetilde{G}$ with $\veps = \frac{\eps}{12 \log(\eps^{-1})}$, by the second claim in Proposition~\ref{prop:boxnewton} we obtain $y  \in \R^{E}_{\ge 0}$ with $\1^\top y \geq  (1-\frac{\eps}{2})\arunsopt - 5 \veps n\geq  (1 - O(\eps)) M^*$, and
\[
\sum_{j \in V} \max \Par{ [\mb^\top y]_j - 1, 0 } \leq \mu n \le \frac{\eps}{12} \arunsopt.
\]
The conclusion follows upon applying Propositions~\ref{prop:cc} and \ref{lem:feasiblecorrect} to round the approximate matching $y$ to be sparse and feasible, and using the implementation of $\oracle$ from Lemma~\ref{lem:oracleexists}.
\end{proof}
 }
\arxiv{%
\section{Deferred proofs from Section~\ref{sec:value}}\label{app:deferredsherman}

In this section, we give a proof of Proposition~\ref{prop:shermancorrect}, which we restate for convenience. We recall the following definition: for convex-concave $f:\xset\times \yset \rightarrow \R$, we say a pair $(x, y)$ is an $\eps$-approximate saddle point if $f(x, y') - f(x', y) \le \eps$ $\forall x' \in \xset$, $y' \in \yset$. Before embarking on the proof, we require the following helper lemma for demonstrating correctness of Line 2 of Algorithm~\ref{alg:sherman}.

\begin{restatable}{lemma}{restatetruncateb}\label{lem:truncateb}
Consider \eqref{eq:boxsimplex} under the transformation described by Lines 1 and 2 of Algorithm~\ref{alg:sherman}. The approximate optimality of any point in $\Delta^m$ for the corresponding problem \eqref{eq:boxsimplexprimal} is unaffected by this transformation (i.e.\ if some point is $\eps$-approximately suboptimal for the new problem, it also has $\eps$-approximate optimality for the old problem). Moreover, using $O(n)$ space and one pass we can implement this transformation for the remainder of any semi-streaming algorithm under the model of Definition~\ref{def:semistream_matrix}.
\end{restatable}
\begin{proof} We discuss correctness of each line separately.

\textit{Correctness of Line 1.} Let $i^* = \argmin_{i \in [m]} c_i$. We claim that the minimizer of $\inprod{c}{x} + \norm{\ma^\top x - b}_1$ only puts nonzero values on coordinates $i \in [m]$ where $c_i \le c_{i^*} + 2\norm{\ma}_\infty$. To see this, consider any $x \in \Delta^m$ where $x_i \neq 0$ for some $c_i \le c_{i^*} + 2\norm{\ma}_\infty$, and consider taking $\tx = x$ in every coordinate, except $\tx_{i^*} = x_{i^*} + x_i$ and $\tx_i = 0$ (i.e.\ we zero out the coordinate $i$ and shift its mass to $i^*$). Then,
\begin{align*}
\inprod{c}{x - \tx} > 2\norm{\ma}_\infty x_i = \norm{\ma}_\infty \norm{\tx - x}_1 \ge \norm{\ma^\top (\tx - x)}_1 \ge \norm{\ma^\top \tx - b}_1 - \norm{\ma^\top x - b}_1.
\end{align*}
Rearranging implies $\tx$ obtains better objective value for the problem \eqref{eq:boxsimplexprimal}. Hence, by ignoring all large coordinates, any $\eps$-approximately suboptimal point in the new problem is also $\eps$-approximately suboptimal for the old problem, because the objective values are the same and any feasible point in the new problem is also feasible for the old problem. Finally, shifting $c$ by a multiple of the all-ones vector affects the problem globally by a constant, since $x$ is in the simplex.

\textit{Correctness of Line 2.} Because $c^\top x$ is unaffected, it suffices to discuss the effect on
	\[\norm{\ma^\top x - b}_1 = \sum_{j \in [n]} \left|\Brack{\ma^\top x}_j - b_j\right|.\]
	For disambiguation call $b'$ the result of the truncation, and $b$ the vector in the original problem. Consider a coordinate where $b_j \neq b'_j$; suppose without loss that $b_j > \norm{\ma}_\infty$ (the other case is handled symmetrically). Then, since $\Brack{\ma^\top x}_j \le \norm{\ma}_\infty$ by definition,
	\[\left|\Brack{\ma^\top x}_j - b_j\right| = b_j - \Brack{\ma^\top x}_j = \Par{b_j - \norm{\ma}_\infty} + \Par{\norm{\ma}_\infty - \Brack{\ma^\top x}_j} = \Par{b_j - \norm{\ma}_\infty} + \left|\Brack{\ma^\top x}_j - b'_j\right|.\]
	Hence the effect is a constant scalar shift, as desired. Finally, regarding the semi-streaming implementation, $b$ is $n$-dimensional and hence we can store it in $O(n)$ space. We can also store $\norm{\ma}_\infty$ and $\cmax$ in one pass and constant space, and simulate the post-processed vector $c$ by modifying any coordinate $c_i$ encountered in a stream using these precomputed values.
\end{proof}

\restateshermancorrect*
\begin{proof}
	First, the effects of Lines 1 and 2 of Algorithm~\ref{alg:sherman} does not affect the problem~\eqref{eq:boxsimplexprimal} other than changing the objective value everywhere by a constant, by Lemma 5 of \cite{CohenST21} and Lemma~\ref{lem:truncateb} respectively. Next, the claim on the number of iterations required by Algorithm~\ref{alg:sherman} follows immediately from Theorem 2 of \cite{CohenST21}. It remains to show that the given value of $K$ suffices for Algorithm~\ref{alg:altmin} to solve the subproblems to $\delta$ additive accuracy.
	
	For simplicity, we will prove this bound on $K$ suffices for the computation of $w_t$; the bound for $z_{t + 1}$ follows analogously. Lemma 7 of \cite{JambulapatiST19} demonstrates that in solving the minimization subproblem of Algorithm~\ref{alg:altmin}, every iteration of alternating minimizaton decreases the suboptimality gap for the subproblem by a constant factor, so we only require a bound on the initial error of each subproblem in Line 5 of Algorithm~\ref{alg:sherman}. The subproblem is of the form: minimize over $(x, y) \in \Delta^m \times [-1, 1]^n$,
	\begin{equation}\label{eq:minwt}\begin{aligned}\inprod{\frac 1 3(\ma y_t + c) + |\ma|(y_t^2)}{x} + \inprod{\frac 1 3(-\ma^\top x_t + b) + 2\diag{y_t}|\ma|^\top x_t}{y} \\
	+ x^\top |\ma| (y^2) + 10\norm{\ma}_\infty \sum_{i \in [m]} x_i \log \frac{x_i}{[x_t]_i}. \end{aligned}\end{equation}
	It is immediate that the range of the first three summands in the objective \eqref{eq:minwt} over the domain $\Delta^m \times [-1, 1]^n$ is $O(\max(\norm{\ma}_\infty, \norm{c}_\infty, \norm{b}_1)) = O(\max(\norm{\ma}_\infty, \norm{b}_1))$. Here we used $\norm{c}_\infty \le 2\norm{\ma}_\infty$. Moreover, Theorem 2 of \cite{CohenST21} shows that the $x$ block of the minimizer of~\eqref{eq:minwt} is entrywise within a multiplicative factor $2$ from $x_t$; in other words, for all $i \in [m]$, $[x'_t]_i \in [\thalf [x_t]_i, 2[x_t]_i]$. Hence, the absolute value of the fourth summand in \eqref{eq:minwt} at the minimizer is bounded by 
	\[10\norm{\ma}_\infty \sum_{i \in [m]} x_i \log 2 = O\Par{\norm{\ma}_\infty}.\] 
	and at the initial point $x^{(0)} = x_t$ it has value $0$. Combining these bounds, and using that the desired accuracy is $\tfrac \eps 2$, yields the required bound on $K$. 
	
	To give the final result, we claim that any $\eps$-approximate saddle point $(\bx, \by)$ to a convex-concave function $f$ on the domain $\xset, \yset$ has that $f_{\xset}(\bx)$ approximates the value of the primal-only problem
	\[\min_{x \in \xset} f_{\xset}(x),\text{ where } f_{\xset}(x) \defeq \max_{y \in \yset} f(x, y), \]
	within an additive $\eps$. To see this, let $(x^\star, y^\star)$ be the exact saddle point of $f$ so that $f(x^\star, y^\star) = \opt$. By definition of $(\bx, \by)$,
	\[\Par{\max_{y \in \yset} f(\bx, y) - \opt} + \Par{\opt - \min_{x \in \xset} f(x, \by)} \le \eps.\]
	Moreover, both quantities on the left hand side are nonnegative by definition of $\opt$, i.e.
	\[\opt = f(x^\star, y^\star) = \min_{x \in \xset} \Par{\max_{y \in \yset} f(x, y)} = \max_{y \in \yset} \Par{\min_{x \in \xset} f(x, y)}.\] 
	Thus, by nonnegativity of the second quantity, we conclude that
	\[f_{\xset}(\bx) = \max_{y \in \yset} f(\bx, y) \le \opt + \eps = f_{\xset}(x^\star) + \eps.\]

\end{proof} 
\section{Matching tools from the literature}
\label{sec:vertex-size-redx}

In this section, we prove Lemma~\ref{lem:greedy} and Proposition~\ref{prop:vertexreduce}. We begin with the former.

\restategreedy*
\begin{proof}
	It is folklore that the greedy algorithm (adding edges one at a time to a matching, if and only if neither endpoint is already matched) is a $2$-approximation to the MCM, proving the approximation guarantee. Regarding the implementation, we can keep marking an array of vertex indices and check if a current edge violates a matched vertex, in $O(1)$ time per edge.
\end{proof}

 We now proceed to Proposition~\ref{prop:vertexreduce}. We first state our algorithm, a greedy iterative procedure which repeatedly adds a maximal matching from a maintained vertex set to the remaining vertices.

\begin{algorithm}
\DontPrintSemicolon
		\KwInput{Bipartite graph $G = (V, E)$ with MCM size $M^*$}
		\KwOutput{Vertex subset $V' \subseteq V$ with $|V'| = O(M^*\log(\eps^{-1}))$ and $G[V']$ that has MCM size $\ge (1 - \eps) M^*$}
		$V_0 \gets$ vertices incident to some maximal matching in $G$\;
		$N \gets \lceil 3\log(\eps^{-1}) \rceil$\;
		\For{$i \in [N]$}{
		$S \gets$ vertices in $V \setminus V_{i - 1}$ incident to some maximal matching between $V_{i - 1}$ and $V \setminus V_{i - 1}$\;
		$V_i \gets V_{i - 1} \cup S$\;
		}
		\Return{$V' \gets V_N$}
	\caption{$\RedSize(G, \eps)$}\label{alg:red-MCM-Vsize}
\end{algorithm}

We require one key technical claim for analyzing Algorithm~\ref{alg:red-MCM-Vsize}. 

\begin{lemma}\label{lem:progressreduce}
Fix an iteration $0 \le i \le N - 1$ of Algorithm~\ref{alg:red-MCM-Vsize}, and suppose the MCM size of $G[V_i]$ is $M_i$. Then, the MCM size of $G[V_{i + 1}]$ increases by at least $\tfrac 1 3 (M^* - M_i)$.
\end{lemma}
\begin{proof}
The symmetric difference between the maximum matchings of $G[V_i]$ and $G$ can be decomposed into even cycles and paths (to see this, every vertex in the symmetric difference has degree at most $2$). Hence, ignoring all even cycles, there is a vertex-disjoint augmenting path decomposition of size $k \defeq M^* - M_i$; call these paths $P_1, P_2, \ldots P_k$. For each path $P_j$, $j \in [k]$, all of its edges except the first and last are contained in $G[V_i]$; call the first and last edges $(x_j, x'_j)$ and $(y_j, y'_j)$, where $x_j, y_j \in G[V_i]$. We consider a few cases, ``processing'' paths sequentially.
\begin{enumerate}
	\item If both $(x_j, x'_j)$ and $(y_j, y'_j)$ are included in the maximal matching between $V_i$ and $V \setminus V_i$, this increases the maximum matching size in $V_{i + 1}$ by one since we can use the augmenting path.
	\item If both $x_j$ and $y_j$ are matched to other vertices $x''_j \neq x'_j$ and $y''_j \neq y'_j$, we can still use the resulting augmenting path in $V_{i + 1}$. The other vertices $x''_j$, $y''_j$ can only remove two other later paths from consideration in the process, where we consider paths sequentially. Note that if $x''_j$ or $y''_j$ is not an endpoint from a later path, this can only help the argument here.
	\item If one of $x_j$ and $y_j$ is matched to its corresponding endpoint in $P_j$ but the other is not, we can use the augmenting path and remove one other path from consideration.
\end{enumerate}
Because we took a maximal matching between $V_i$ and $S \setminus V_i$, at least one of the above cases must happen, so we increase the MCM size by at least $\tfrac 1 3 k$.
\end{proof}

We conclude with our analysis of Algorithm~\ref{alg:red-MCM-Vsize}.

\restatevertexreduce*
\begin{proof}
The correctness follows immediately by applying Lemma~\ref{lem:progressreduce}, and the fact that our original maximal matching (contained in $G[V_0]$) is of size at least $\thalf M^*$. We now prove correctness. For the first pass, the complexity follows directly from Lemma~\ref{lem:greedy}. For all additional iterations of the algorithm, note that we can only add $O(M^*)$ additional vertices to the current set since it is a valid matching in the original graph, and the iteration can be implemented in a single pass analogously to the proof of Lemma~\ref{lem:greedy}. The space overhead is only the maintenance of the current set $V_i$.
\end{proof}

\section{Cycle cancelling in low space}\label{sec:cc}

In this section, we provide implementation of a data structure which proves the following claim.

\restatecc*

Specifically, we reduce its proof to demonstrating the existence of a \emph{bipartite cycle-cancelling oracle}.

\begin{definition}[Bipartite cycle-cancelling oracle]
We call $\oracle$ a \emph{bipartite cycle-cancelling oracle} (BCCO) if it is associated with a (weighted) bipartite graph $G = (V, E, w)$, and given any vector $x \in \R_{\ge 0}^E$ supported on $L \le 2n$ edges, $\oracle$ outputs a vector $\tx \in \R_{\ge 0}^E$ such that $\mb^\top x = \mb^\top \tx$ and $\inprod{w}{\tx} \ge \inprod{w}{x}$, so that $\tx$ is supported on at most $n$ edges. \end{definition}

\begin{lemma}\label{lem:oracleexists}
A BCCO $\oracle$ is implementable in $O(L)$ space and $O(L \log n)$ work.
\end{lemma}

\begin{proof}[Proof of Proposition~\ref{prop:cc}]
If the stream length is $L \le 2n$, applying $\oracle$ obtains the desired results.
	
Otherwise, it suffices to divide the stream into chunks of $n$ edges, and repeatedly call $\oracle$. Specifically, we can input the first $2n$ edges of the stream into $\oracle$ and produce a set of $n$ edges, then input these $n$ edges with the next $n$ edges of the stream, and so on. Inducting on the properties of $\oracle$, we obtain the value and feasibility guarantees on the final output. Regarding the space and work guarantees, it suffices to repeatedly apply Lemma~\ref{lem:oracleexists}, and note that we can reuse the space for $\oracle$.
\end{proof}

The remainder of this section is devoted to proving Lemma~\ref{lem:oracleexists}. Specifically, we demonstrate how to use the link/cut tree data structure of \cite{SleatorT83, Tarjan83}, with a few small modifications, to provide the required oracle. In Section~\ref{sssec:lctree} we state some preliminaries on link/cut trees which are known from prior work (namely a set of supported operations), and we put these components together in Section~\ref{sssec:oracleproof} to give an implementation which proves Lemma~\ref{lem:oracleexists}.

\subsection{Link/cut tree description}\label{sssec:lctree}

The description in this section primarily follows the implementation of link/cut trees given in Chapter 5 of \cite{Tarjan83}, with a few terminology changes following the later lecture notes of \cite{Dem12}. It is primarily a summary of prior work, where we state the supported operations that we will use.

\textbf{Data structure representation.} The link/cut tree implementation of \cite{SleatorT83, Tarjan83} we will use is maintained as a forest, where each (undirected, rooted) tree in the forest is decomposed into ``preferred paths'' and each path is stored as a splay tree, along with a parent pointer to the neighbor of the preferred path endpoint closer to the root. Every edge in the forest has an associated value (which can change). We refer to the tree representations of preferred paths as ``auxiliary trees'' and the tree representation of all connected vertices as a ``main tree.'' The link/cut tree supports queries or modifications to main trees, including the following operations.
\begin{enumerate}
	\item $\Link(v, w, C)$: add an edge of value $C$ between $v$ and $w$ in the main tree (if they are in different trees, join them so that $w$ is the parent of $v$).
	\item $\Cut(v)$: remove the edge between $v$ and its parent in the main tree.
	\item $\ChangeRoot(r)$: change the root of the main tree to the vertex $r$.
	\item $\LCA(v, w)$: return the least common ancestor of $v$ and $w$ in the main tree.
	\item $\Min(v)$: return the smallest edge value between $v$ and its root in the main tree.
	\item $\Add(v, C)$: add $C$ to all edge values between $v$ and its root in the main tree.
	\item $\Sum(v)$: return the sum of all edge values between $v$ and its root in the main tree.
\end{enumerate}

For a more complete and formal description, we refer the reader to \cite{Tarjan83, Dem12}.

\textbf{Space complexity of link/cut trees.} Based on the representation described above, it is clear that as long as the amount of additional information we need to store in the auxiliary trees to implement path aggregation operations is a constant per vertex, the overall space complexity is $O(n)$ where $n$ is an upper bound on the number of vertices of all main trees. The only remaining space overhead is storing all auxiliary trees, and the parent pointers of preferred paths in main trees. We include this discussion because it is not explicitly stated in the source material.

\textbf{Work complexity of link/cut trees.} We state the guarantees of the link/cut tree in the following claim, which follows by the analysis of the works \cite{SleatorT83, Tarjan83}.

\begin{proposition}[Main result of \cite{SleatorT83, Tarjan83}]\label{prop:reducetopath}
The amoritized cost of $L$ calls to any of the operations $\Link$, $\Cut$, $\ChangeRoot$, $\LCA$, $\Min$, $\Add$, or $\Sum$ is bounded by $O(L\log n)$.
\end{proposition}

\subsection{Implementation of BCCO}\label{sssec:oracleproof}

We now prove Lemma~\ref{lem:oracleexists}. We first outline our approach in constructing the oracle, and then describe a concrete implementation by using link/cut trees.

\textbf{Approach outline.} The main goal is to show how to preserve an acyclic matching $\tx$ so that $\mb^\top \tx = \mb^\top x$ and $\inprod{w}{\tx} \ge \inprod{w}{x}$ after processing every edge, in a total of $O(n\log n)$ work using the link/cut tree. Inductively it is clear that the output will satisfy all the requirements of $\oracle$.

We now describe how to process a single input edge $e = (u, v)$ where $u \in L$ and $v \in R$ are on opposite sides of the bipartition. First, if $u$ and $v$ are both not in the data structure, create a new main tree consisting of just this edge (with say $u$ as the root). If only one is in the data structure then create a new edge with value $x_e$ appropriately using $\Link$. If $u$ and $v$ are in different main trees, we can call $\Link(u, v, x_e)$. In these cases, no cycles are created and no edge values are changed.

The last case is when $u$ and $v$ are in the same main tree. The path between $u$ and $v$ in their main tree and the new edge $(u, v)$ forms a cycle. In order to preserve $\mb^\top x$, it suffices to alternate adding and subtracting some amount $C$ from edges along the cycle. In order to make sure $\inprod w x$ is monotone, it suffices to compute the alternating sum of weights along the cycle to pick a parity, i.e.\ whether we add from every odd edge in the cycle and subtract from every even edge, or vice versa (since all cycles are even length, one of these will increase the weight). Once we pick a parity, we compute the minimum value amongst all edges of that parity in the cycle and alternate adding and subtracting this value; this will result in the tree becoming acyclic and preserve all invariants.

We finally remark that it suffices to divide the support of the input $x$ into connected components, and remove cycles from each connected component. By designating a ``root vertex'' in each connected component and processing the edges in that component in BFS order from the root, it is clear that the tree corresponding to that component will never become disconnected (edges are only removed when a cycle is formed, and removing any edge on the cycle will not disconnect the tree). Thus, it suffices to discuss how to implement cycle cancelling for a single connected component.

\textbf{Implementation.} We discuss the implementation of cycle cancelling for a single connected component. It will use three link/cut trees, called $W$, $T_+$, and $T_-$; the topology of these three trees will always be the same (i.e.\ they are the same tree up to the edge values). Roughly speaking, $W$ will contain edge weights with alternating signs (e.g.\ each edge's weight will be signed $(-1)^{\text{depth of edge}}$), and will be used to determine the direction of cycle cancelling to preserve $\inprod{w}{\tx} \ge \inprod{w}{x}$. Further, $T_+$ and $T_-$ will each correctly contain half of the edge values of the current fractional matching (depending on their sign in $W$), and the other edge values will be set to a large quantity. They are used to compute the minimum (signed) edge value in cycles and to modify the fractional matching.

Consider processing a new edge $e = (u, v)$ with weight $w_e$ and value $x_e$. If this edge is in any of the non-cycle-creating cases described in the outline, i.e.\ its endpoints do not belong in the same tree with a path between them, then we add a copy to each of $W$, $T_+$, and $T_-$. Depending on the depth of the edge in $W$ (which we can check by looking at the sign of its parent edge), we either give it value $w_e$ or $-w_e$. If it is positive in $W$, then we give it value $x_e$ in $T_+$ and value $n^2 \norm{x}_\infty$ in $T_-$; otherwise, we swap these values. We choose the value $n^2 \norm{x}_\infty$ for one copy of the edge, so that it never becomes the minimum value edge returned by $\Min$ throughout the algorithm.

Finally, we handle the case where we need to remove a cycle (when $u$ and $v$ are in the same tree, and the edge is not in the tree). We start by determining which half of edges in the cycle we want to remove value from, and which we should add to. To do this, we query $r = \LCA(u, v)$ and $\ChangeRoot(r)$ on the tree in $W$. Then, we compute the alternating sums of weights in each half of edges in the cycle. This allows us to determine a direction to preserve $\inprod w \tx \ge \inprod w x$. We also call $\LCA(u, v)$ and $\ChangeRoot(r)$ for both copies of the main tree in $T$.

We next determine the amount we wish to alternatingly add and subtract along the $u$-to-$r$ and $v$-to-$r$ paths by calling $\Min$ on the appropriate tree, $T_-$ or $T_+$ (corresponding to the direction obtained from querying $W$). We then call $\Add$ on both copies of the main tree in $T$ for these paths, with the value obtained by the $\Min$ queries or its negation appropriately. This zeros one edge, which we $\Cut$ from all three copies of the tree. Finally, we revert to the original root in all three trees to maintain correctness of signs. Overall, the number of link/cut tree operations per edge in the stream is a constant, so Proposition~\ref{prop:reducetopath} bounds the total work by $O(n\log n)$.

\section{Sampling for rounding linear programming solutions}\label{sec:samplelp}

In this section, we give a general procedure for rounding a fractional solution which is returned by our algorithm for box-simplex games in Section~\ref{sec:value}. This procedure applies to general box-simplex problems, beyond those with combinatorial structure, so we include this section for completeness. In particular, directly applying the sparsity bounds of this section to our matching problems directly imply $\tO(n \cdot \text{poly}(\eps^{-1}))$ space bounds for the various matching-related applications in the paper, up to width parameters, which roughly match our strongest results up to the $\eps^{-1}$ dependence. As an example, we give an application of this technique to MCM at the end of this section.

 Given streaming access to an approximate fractional solution $x$ on the simplex (via an implicit representation), a natural way of constructing a low-space approximate solution is to randomly sample each entry of $x_i$ and reweight to preserve expected objective value; this is the rounding strategy we analyze. We use the following Algorithm~\ref{alg:randsample}, parameterized by some prescribed $\{M_i\}_{i \in [m]}$.

\begin{algorithm}
\DontPrintSemicolon
\KwInput{Coordinates of $x \in \R_{\ge 0}^m$ in streaming fashion, sample count $K$, parameters $\{M_i\}_{i \in [m]}$}
		\For{$i\in[m]$}{
		Draw $K$ independent random variables $\{X_i^k\}_{k \in [K]}$, where 
		\[X_i^k = \begin{cases}M_i & \text{with probability } \tfrac{x_i}{M_i}\\ 0 &\text{otherwise}\end{cases}.\]\;
		$\hat{x}_i \gets \tfrac{1}{K}\sum_{k\in[K]}X_i^k$\;
		}
		\Return{$\hat{x}$}
	\caption{$\RandomSample(x, K, \{M_i\}_{i \in [m]})$}\label{alg:randsample}
\end{algorithm}

\subsection{Concentration bounds}

For proofs in this section, as well as later, we crucially rely on well-known concentration properties of bounded random variables. We state here a few facts used repeatedly throughout.

\begin{proposition}[Chernoff bound]\label{prop:cher-down}
	For $K$ independent scaled Bernoulli random variables $\{X_k\}_{k\in[K]}$ satisfying $X_k = N_k$ with probability $p_k$, $0<N_k \le 1$ for all $k \in [K]$, and all $0<\delta<1$,
	\[
	\Pr\left(\left|\sum_{k\in[K]}X_k-\sum_{k\in[K]}\E X_k\right| \ge \delta \sum_{k\in[K]}\E X_k\right)\le 2\exp\left(-\frac{\delta^2\sum_{k\in[K]}\E X_k}{3}\right)
	\]
\end{proposition}

We give a simple generalization of Proposition~\ref{prop:cher-down} to the case where the scaled Bernoulli variables are allowed to take on negative values.

\begin{corollary}[Generalized Chernoff bound]\label{cor:cher-down}
	For $K$ independent scaled Bernoulli random variables $\{X_k\}_{k\in[K]}$	satisfying $X_k= N_k$ with probability $p_k$, $0< |N_k| \le 1$ for all $k \in [K]$, and all $0<\delta<1$,
	\[
	\Pr\left(\left|\sum_{k\in[K]}X_k-\sum_{k\in[K]}\E X_k\right| \ge \delta \sum_{k\in[K]}\E |X_k|\right)\le 4\exp\left(-\frac{\delta^2\sum_{k\in[K]}\E |X_k|}{3}\right)
	\]
\end{corollary}
\begin{proof}
	Divide the set $[K] = \mathcal{K}^+\cup\mathcal{K}^-$, where we define $\mathcal{K}^+\defeq\{k\in[K]| N_k\ge0\}$ and $\mathcal{K}^+\defeq\{k\in[K]| N_k<0\}$. Applying Proposition~\ref{prop:cher-down} to $\sum_{k\in\mathcal{K}^+}X_k$ 	 and $\sum_{k\in\mathcal{K}^-}-X_k$ yields the result.
\end{proof}

Furthermore, the following one-sided Chernoff bound holds when $0<N_k\le 1$ and $\delta\ge 1$.

\begin{proposition}[One-sided Chernoff bound]\label{prop:cher-down-oneside}
	For $K$ independent scaled Bernoulli random variables $\{X_k\}_{k\in[K]}$ satisfying $X_k = N_k$ with probability $p_k$, $0<N_k \le 1$ for all $k \in [K]$, and all $\delta>0$,
	\[
	\Pr\left(\sum_{k\in[K]}X_k-\sum_{k\in[K]}\E X_k \ge \delta \sum_{k\in[K]}\E X_k\right)\le\exp\left(-\frac{\delta^2\sum_{k\in[K]}\E X_k}{2+\delta}\right)
	\]
\end{proposition}

\begin{proposition}[Bernstein's inequality]\label{prop:bernstein}
	For $K$ independent random variables $\{X_k\}_{k \in [K]}$	satisfying $|X_k|\le C$ with probability one, let $V = \sum_{k\in[K]} \Var[X_k]$. Then for all $t \ge 0$,
	\[
	\Pr\left(\left|\sum_{k\in[K]}X_k-\sum_{k\in[K]}\E X_k\right| \ge t\right) \le 2\exp\left(-\frac{t^2}{2V+2Ct/3}\right).
	\]
\end{proposition}

\subsection{Random sampling guarantees}

We first give a general guarantee on the approximation error incurred by random sampling via Algorithm~\ref{alg:randsample}. While the guarantees are a bit cumbersome to state, they become significantly simpler in applications. For instance, in all our applications, all binary random variables are scaled with $M_i$ such that  $\max_{i \in [n]} M_i \le 1$ (see Lemma~\ref{lem:gen-rounding} for definition), and the bounds on $\ma^\top \hx - \ma^\top x$ become standard multiplicative error approximations when the matrix $\ma$ is all-positive.

\begin{lemma}\label{lem:gen-rounding}
	Consider an instance of problem \eqref{eq:boxsimplex} parameterized by $\ma$, $b$, $c$. For some $x \in \Delta^m$ whose coordinates can be computed in streaming fashion, define for all $j \in [n]$,
	\[B_j = \Brack{\left|\ma\right|^\top x}_j.\] 
	Let $\hat{x}$ be the output of Algorithm~\ref{alg:randsample} on input $x$ with 
	\[M_{i} = \min_{j\in [n]}\frac{B_j}{|\ma_{ij}|}\text{ and } K= \frac{12\log(mn)}{\eps^2}.\]
	With probability at least $1 - (mn)^{-1}$, $\hat{x}$ satisfies the following properties for $B \defeq \max_{i\in[m]} M_i$:
	\begin{align*}
	\left|\Brack{\ma^\top \hat{x}-\ma^\top x}_j\right|& \le \eps B_j \text{ for all } j \in [n], \quad |\norm{\hat{x}}_1 - \norm{x}_1| \le \eps \max(1,B), \\
	\left|c^\top\hat{x}-c^\top x\right| & \le \eps\norm{c}_\infty \max(1,B),\\
	\|\hat{x}\|_0 & =  O\left(\left(1 + \sum_{i \in [m]} x_i \max_{j\in[n]}\frac{|\ma_{ij}|}{B_j}\right)\cdot\frac{\log(mn)}{\eps^2}\right).
	\end{align*}
\end{lemma}

\begin{proof}
	First note that it is immediate by definition of the $B_j$ to see that all $\tfrac{x_i}{M_i}=\max_{j\in[n]}\tfrac{|\ma_{ij}|x_i}{B_j}\le 1$ are valid sampling probabilities for all $i\in[m]$. Also, recall that entrywise
	\[\hat{x}_i = \frac{1}{K}\sum_{k \in [K]} X_i^k.\]	
	We show the first property; fix some $j \in [n]$ and consider $[\ma^\top \hx]_j - [\ma^\top x]_j $. Applying Corollary~\ref{cor:cher-down} to the random variables $\{\tfrac{1}{B_j}\ma_{ij}X^k_{i}\}_{i \in [m], k \in [K]}$ with $\delta = \eps$, we see by definition of $K$ that
	\begin{align*}
	\sum_{i \in \mathcal[m], k \in [K]} \frac{1}{B_j}\ma_{ij}X^k_{i} = \frac{K}{B_j}  \Brack{\ma^\top \hx}_j, & \\
	\sum_{i \in \mathcal[m], k \in [K]} \E\Brack{\frac{1}{B_j} \ma_{ij}X_{i}^k} = \frac{K}{B_j} [\ma^\top x]_j, & \;\\ \sum_{i \in \mathcal[m], k \in [K]} \E\left|\frac{1}{B_j} \ma_{ij}X_{i}^k\right| = \frac{K}{B_j} [|\ma|^\top x]_j,& \\
	\implies \Pr\left(\frac{K}{B_j}\left|[\ma^\top \hx]_j-[\ma^\top x]_j\right| \ge \delta \frac{K}{B_j}[|\ma|^\top x]_j\right)
	& \le 4\exp\left(-\frac{\delta^2 K[|\ma|^\top x]_j}{3B_j}\right)\le \frac{4}{(nm)^4},
	\end{align*}
	where for the last inequality we use definitions of $\delta$, $K$, and that $B_j= [|\ma|^\top x]_j$, for all $j\in[n]$.

	This conclusion for a coordinate $j \in [n]$ is equivalent to $\left|[\ma^\top \hx]_j - [\ma^\top x]_j \right| \le \eps B_j$. Union bounding over all $j \in [n]$, we thus have with probability at least $1-\tfrac{4}{n^3m^4}$ 
	\[\left|\left[\ma^\top \hx - \ma^\top x\right]_j\right|\le \eps B_j,\quad\forall j\in[n].\]
	Now we show the second and third properties. 
	Given $B\defeq\max_{i\in[m]}M_i$, we first apply Proposition~\ref{prop:bernstein} on the sum $\sum_{i\in[m]}K c_{i}\hat{x}_{i} = \sum_{i \in [m], k \in [K]} c_i X_i^{k}$, using the bound
	\begin{align*}\sum_{i \in [m], k \in [K]} \Var\Brack{c_{i} X_{i}^k} 
	&\le K\norm{c}_{\infty}^2 \sum_{i \in [m]} \Var\Brack{X_i^k} \\
	&\le K\norm{c}_{\infty}^2\sum_{i \in [m]} x_i M_i \le K\norm{c}_\infty^2 B.\end{align*}
	Thus, we can choose parameters $C = \norm{c}_{\infty}B$, $V = K\norm{c}_{\infty}^2B$ to obtain the following bounds for $K\ge \tfrac{12\log(mn)}{\eps^2}$ depending on whether $B>1$ or $B\le1$. For $B>1$, we have
	\begin{align*}
	\Pr\left(K\left|\sum_{i \in [m]} c_{i}\hx_{i}- \sum_{i \in [m]}c_{i} x_{i} \right|\ge \eps B\norm{c}_{\infty} K\right)&\le2\exp\left(-\frac{\eps^2B^2\norm{c}_{\infty}^2 K^2}{2K\norm{c}_{\infty}^2B+ 2\norm{c}_{\infty}^2B^2K\eps/3}\right) \\&\le \frac{1}{4(mn)^2}.
	\end{align*}
	
	For $B\le1$, we have
	\begin{align*}
	\Pr\left(K\left|\sum_{i\in[m]} c_{i}\hx_{i}- \sum_{i\in [m]}c_{i} x_{i} \right|\ge \eps\norm{c}_{\infty} K\right)&\le2\exp\left(-\frac{\eps^2\norm{c}_{\infty}^2 K^2}{2K\norm{c}_{\infty}^2B+ 2\norm{c}_{\infty}^2BK\eps/3}\right) \\&\le \frac{1}{4(mn)^2}.
	\end{align*}
	
	Altogether this implies the third conclusion, and the second conclusion follows as a special case when specifically picking $c=\1_n$. Each holds with probability $\ge 1-\tfrac{1}{4}(mn)^{-2}$.
	
	Finally, for all $i\in [m]$, let $Y_{i}=1$ if $\hat{x}_{i}\neq 0$ and $Y_{i} = 0$ otherwise. It is straightforward to see that $Y_{i}=1$ with probability 
	\begin{align*}
	1-\Par{1-\max_{j\in[n]}\frac{|\ma_{ij}|x_{i}}{B_j}}^{K} & \le K\max_{j\in[n]}\frac{|\ma_{ij}|x_{i}}{B_j}\\
	\implies \sum_{i \in [m]} \E[Y_{i}] \le K \sum_{i \in [m]} x_i \max_{j\in[n]}\frac{|\ma_{ij}|}{B_j} & \le K \max\left( \sum_{i \in [m]} x_i \max_{j\in[n]}\frac{|\ma_{ij}|}{B_j},1\right).
	\end{align*}
	
	Thus, by applying the one-sided Chernoff bound in Proposition~\ref{prop:cher-down-oneside} (where we consider the cases where $\sum_{i \in [m]} x_i \max_{j\in[n]}\frac{|\ma_{ij}|}{B_j}\ge1$ and $\sum_{i \in [m]} x_i \max_{j\in[n]}\frac{|\ma_{ij}|}{B_j}\le1$ separately), 
	\[
	\Pr\left( \sum_{i \in [m]} Y_{i} \ge 2K+2K\Par{\sum_{i \in [m]} x_i \max_{j\in[n]}\frac{|\ma_{ij}|}{B_j}}\right) \le \frac{1}{4(mn)^2}.
	\]
	Finally, applying a union bound, all desired events hold with probability $ 1-\tfrac{1}{mn}$.
\end{proof}

Note that in the semi-streaming model, one can choose $B_j \defeq \Brack{|\ma|^\top x}_j$ and compute all such values in one pass when $x$ is of the form in Lemma~\ref{lem:implicitx}; however, occasionally we will choose larger values of $B_j$ to obtain improved sparsity guarantees. Thus, we provide the following one-sided guarantee, which holds when $\ma_{ij}\ge 0$, $\forall i,j$ and $B_j\ge [\ma^\top x]_j$.

\begin{corollary}\label{cor:gen-rounding}
	Consider an instance of problem \eqref{eq:boxsimplex} parameterized by $\ma$ with $\ma_{ij}\ge0$ $\forall$ i,j, and $b$, $c$. For some $x \in \Delta^m$ whose coordinates can be computed in streaming fashion, suppose we have bounds for all $j \in [n]$,
	\[B_j \ge \Brack{\ma^\top x}_j.\] 
	Let $\hat{x}$ be the output of Algorithm~\ref{alg:randsample} on input $x$ with 
	\[M_{i} = \min_{j\in [n]}\frac{B_j}{|\ma_{ij}|}, \text{ and } K= \frac{12\log(mn)}{\eps^2}.\]
	With probability at least $1 - (mn)^{-1}$, $\hat{x}$ satisfies the following properties for $B \defeq \max_{i\in[m]} M_i$:
	\begin{align*}
	\Brack{\ma^\top \hat{x}-\ma^\top x}_j\le \eps B_j \text{ for all } j \in [n],& \quad |\norm{\hat{x}}_1 - \norm{x}_1| \le \eps \max(1,B), \\
	\left|c^\top\hat{x}-c^\top x\right| \le \eps\norm{c}_\infty \max(1,B),
	\quad \|\hat{x}\|_0 & =  O\left(\left(1 + \sum_{i \in [m]} x_i \max_{j\in[n]}\frac{|\ma_{ij}|}{B_j}\right)\cdot\frac{\log(mn)}{\eps^2}\right).
	\end{align*}
\end{corollary}

Note that the first and fourth properties of Corollary~\ref{cor:gen-rounding} are one-sided in that we only wish to upper bound a property of $\hx$. The proof is identical to Lemma~\ref{lem:gen-rounding}, except that we use Proposition~\ref{prop:cher-down-oneside} with $\delta = \eps\cdot\tfrac{B_j}{[|\ma|^\top x]_j}>0$ instead of Proposition~\ref{prop:cher-down} in cases where $\delta \ge 1$, which suffices for the one-sided bound of the first property. Similarly, when $\delta\ge1$ as in the proof of the fourth property, the one-sided Chernoff bound only helps concentration. Next, we give an end-to-end guarantee on turning Algorithm~\ref{alg:sherman} into a solver for the fractional problem, via applying Algorithm~\ref{alg:randsample} on the output.

\begin{lemma}\label{lem:gen-rounding-back}
	Given $(x,y)$, an $\tfrac \eps 2$-approximate saddle point of~\eqref{eq:boxsimplex}, let \[B_j =[|\ma|^\top x]_j,\quad B = \max_{i\in[m]}\min_{j\in [n]}\frac{B_j}{|\ma_{ij}|}\] and $\hx = \RandomSample(x, K, \{M_i\}_{i \in [m]})$ with
	\[M_i = \min_{j \in [n]} \frac{B_j }{|\ma_{ij}|},\text{ and } K =O\Par{ \frac{\log(mn)\Par{\Par{\norm{c}_\infty^2 + \norm{\ma}_\infty^2}(1+B)^2+\Par{\sum_{j\in[n]} B_j}^2}}{\eps^2}}.\]
	With probability $1-(mn)^{-1}$, $(\tfrac{\hx}{\|\hx\|_1},y)$ is an $\eps$-approximate saddle point to~\eqref{eq:boxsimplex}. Moreover, the total space complexity of Algorithm~\ref{alg:sherman} and Algorithm~\ref{alg:randsample} to compute and store the output $(\tfrac{\hx}{\|\hx\|_1}, y)$ is
	\[O\Par{\log n\Par{\sum_{i \in [m]} x_i \max_{j \in [n]} \frac{|\ma_{ij}|}{B_j}} \cdot \frac{\Par{\Par{\norm{c}_\infty^2 + \norm{\ma}_\infty^2}(1+B)^2+\Par{\sum_{j\in[n]} B_j}^2}}{\eps^2} + \frac{n\norm{\ma}_\infty}{\eps} }.\]
\end{lemma}

\begin{proof}
	Our choice of $K$ is with respect to an accuracy parameter on the order of  
	\[\frac{\eps}{(\norm{c}_\infty+\norm{\ma}_\infty) (1+B) + \sum_{j\in[n]} B_j}.\] Recall that an $\eps$-approximate saddle point to a convex-concave function $f$ satisfies
	\[
	\max_{\by \in \yset} f(x, \by) - \min_{\bx \in \xset} f(\bx, y) \le \eps.
	\]	
	For our output $(\tfrac{\hx}{\|\hx\|_1}, y)$, since we keep the same box variable $y$, it suffices to show that the side of the duality gap due to $x$ and $\tfrac{\hx}{\|\hx\|_1}$ does not change significantly. Namely, we wish to show
	\[ \max_{\by \in \yset} f\Par{\frac{\hx}{\|\hx\|_1}, \by} - \max_{\by \in \yset} f(x, \by) = \Par{c^\top \frac{\hx}{\|\hx\|_1} + \norm{\ma^\top \frac{\hx}{\|\hx\|_1} - b}_1} - \Par{c^\top x + \norm{\ma^\top x - b}_1} \le \frac{\eps}{2}.\]
	By the triangle inequality, it equivalently suffices to show that
	\begin{align*}c^\top \hx - c^\top x \le \frac{\eps}{8},\; c^\top \hx \Par{\frac{1}{\norm{\hx}_1 } - 1} \le \frac{\eps}{8},\\
	\norm{\ma^\top\hx - \ma^\top x}_1 \le \frac{\eps}{8},\; \norm{\ma^\top \hx}_1 \left|\frac{1}{\norm{\hx}_1 } - 1 \right| \le \frac{\eps}{8}.\end{align*}
	The first and third conclusions hold by applying Lemma~\ref{lem:gen-rounding} for the choice of $K$. The second and fourth hold by $|c^\top \hx|\le \norm{c}_\infty\norm{\hx}_1$ and $\norm{\ma^\top\hx}_1\le \norm{\ma}_\infty\norm{\hx}_1$ and then applying Lemma~\ref{lem:gen-rounding} for the choice of $K$. Finally, the desired sparsity follows by combining the space bound of the output (via Lemma~\ref{lem:gen-rounding}) and the space complexity of implicitly representing the average iterate.
\end{proof}

\subsection{Application: rounding MCM solutions}

We give a simple application of our random sampling framework to computing an explicit low-space approximate MCM. While the space complexity does not match our strongest results based on a low-space cycle cancelling implementation, we hope it is a useful example of how to more generally sparsify fractional solutions of box-simplex games without explicit combinatorial structure.

Following the reductions of Section~\ref{sec:flowround}, we assume in this section that we have a simplex variable $x \in \Delta^m$ and a matching size $\bM \in [1, n]$ such that
\[\mb^\top (\bM x) \le \1_V,\; \bM\ge (1 - \eps)M^*,\]
where $M^*$ is the maximum matching size. We now show how to apply our random sampling procedure, Algorithm~\ref{alg:randsample}, to sparsify the support of the matching without significant loss.

\begin{corollary}\label{lem:MCM-rounding}
	Suppose for an MCM problem, $x \in \Delta^{E}$ satisfies $\mb^\top x \le \tfrac{1}{\bM}\1_V$, and $\eps \in (0, 1)$. Let $\hat{x}$ be the output of Algorithm~\ref{alg:randsample} on input $x$ with $m_{e} = \tfrac 1 \bM$ for all $e \in E$, and $K\defeq \tfrac{12\log n}{\eps^2}$.
	Then with probability at least $1 - n^{-2}$, the output $\hat{x}$ satisfies the following properties:
	\[
	\mb^\top \hat{x}\le 
	\frac{1+\eps}{\bM}\1_V, \quad |\norm{\hat{x}}_1 - 1| \le \eps,\quad \|\hat{x}\|_0 =  O\left(\frac{\bM\log n}{\eps^2}\right).
	\]
\end{corollary}

\begin{proof}
	It is straightforward to see that by the assumptions, we can take $\ma = \bM \mb$ and $B_v = 1$ for all $v \in V$, and choose the accuracy level to be $\eps$ in Lemma~\ref{lem:gen-rounding}. Thus, by the conclusions of Lemma~\ref{lem:gen-rounding}, noting that $M_e = \tfrac{1}{\bM} \le 1$ holds for all $e\in E$, we conclude that with probability at least $1 - n^{-2}$, all desired guarantees hold:
	\begin{align*}
	& \bM\mb^\top \hat{x}\le \bM\mb^\top x + \eps\1_V\le(1+\eps)\1_V,\\
	& |\norm{\hat{x}}_1 - 1| \le \eps,\\
	& \|\hat{x}\|_0 =  O\left(\left(\sum_{e \in E} x_e \max_{v\in V} \frac{|\ma_{ev}|}{B_v}\right)\cdot\frac{\log(mn)}{\eps^2}\right) = O\left(\frac{\bM\log n}{\eps^2}\right).
	\end{align*}
\end{proof}

\section{Approximate MCM via box-constrained Newton's method}
\label{sec:box}

\newcommand{\arunsopt}{M^*}
\newcommand{\xopt}{x^{\star}}
\newcommand{\xhat}{\hat{x}}
\newcommand{\veps}{\varepsilon}
\newcommand{\xtilde}{\widetilde{x}}
\newcommand{\vtilde}{\widetilde{v}}
\newcommand{\mlap}{\mathcal{L}}
The goal of this section is to prove the following Theorem~\ref{thm:boxnewton_alg}. In particular, we give an alternate second-order method to compute $(1-\eps)$-approximate maximum matchings in unweighted graphs using techniques developed by \cite{CMTV17} for solving matrix scaling and balancing problems. 

\begin{restatable}{theorem}{boxconstrained}\label{thm:boxnewton_alg}
For a MCM problem on bipartite $G = (V,E)$ with $|V| =n$, $|E| = m$ and optimal value $M^*$, Algorithm~\ref{alg:boxconstrained} with parameter\footnote{This lower bound on $\eps$ is without loss of generality as otherwise we can use the algorithm in Section~\ref{ssec:exactMCM} which computes an exact MCM with a larger value of $\eps$.} $\eps \in [ \Theta\left(\frac{\log(mn)}{n}\right), \frac{1}{2}]$ obtains a matching of size at least $(1-\eps)M^*$ using $\tO(n)$ space, $\tO(\eps^{-1})$ passes, and $\tO(m)$ work per pass. 
\end{restatable}

We first recall the dual (vertex cover) formulation of the standard bipartite matching linear program
\begin{equation}
\label{eqn:exactLP}
\min_{v \geq 0, \\ \mb v \geq \1} \1^\top v,
\end{equation}
where $\mb$ is the unoriented incidence matrix. Our method (which is based on the box-constrained Newton's method of \cite{CMTV17}) uses a relaxation of this linear program which makes use of oriented incidence matrices and Laplacians, which we now define.

\begin{definition}
Let $G = (V,E,w)$ be a weighted undirected bipartite graph with bipartition $L, R$
and nonnegative edge weights $w$. We define the unoriented incidence matrix $\mb  \in \R^{E \times V}$ as the matrix where for any edge $e = (i,j), i \in L, j \in R$ the row corresponding to $e$ in $\tmb$ has $\tmb_{ei} = \tmb_{ej} = 1$, and all other entries set to $0$. Additionally, we define the \emph{oriented} incidence matrix $\tmb \in \R^{E \times V}$ as \[
\tmb = \begin{bmatrix}
\mathbf{I}_L & 0\\ 
0   &  -\mathbf{I}_R
\end{bmatrix} \mb,
\]
and the graph Laplacian $\mlap_G =  \tmb^\top \diag{w} \tmb$. 
\end{definition}

When the graph is obvious we drop the subscript from $\mlap_G$. The orientation of the edges in $\tmb$ are typically chosen to be arbitrary in the literature, but we specify edge orientation as our algorithm will distinguish the vertices in $L$ and $R$, and denote a vector $x$ on these vertices as $x_L$ and $x_R$ (or sometimes $[x]_L$, $[x]_R$ for clarity). We will first describe our regularization scheme for this LP and prove that approximate minimizers for the regularized objective yield approximate fractional vertex covers. We then prove some stability properties on the Hessian of our regularized objective and show that a second-order method can be implemented in $\tO(n)$ space and $\tO(\eps^{-1})$ passes. Our choice of regularization and notion of stability are heavily based on \cite{CMTV17}, which employed a similar regularization scheme to solve matrix scaling and balancing problems. We nevertheless give a self-contained description of these details for completeness. 

\begin{lemma}[Properties of regularized vertex cover]
\label{lem:VC-relaxed}
Let $G$ be an unweighted bipartite graph with $n$ nodes and $m$ edges, and let $L,R$ be the vertices on either side of its bipartition. Let $\arunsopt$ be the size of the maximum matching in $G$. Let $\veps \geq \frac{8 \log(mn)}{n}$ be a parameter, and set $\mu = \frac{\veps}{4 \log(mn)}$. Consider
\begin{equation}
    \label{eqn:approxLP}
    f_{\mu} (x) \defeq \1^\top x + \mu \left( \sum_{e \in E} e^{\frac{1}{\mu} (1 - [\mb x]_e )} + \sum_{i \in V} e^{-\frac{1}{\mu} (x_i +\tfrac \veps 4)}  \right).
\end{equation}
$f_{\mu}$ has the following properties.
\begin{itemize}
    \item For any $x$ with $f_{\mu}(x) <m$ we have $\mb x \geq (1-\tfrac \veps 2)\1$ and $x \geq - \frac{\veps}{2}$
    \item If $\xopt$ is a feasible minimizer to~\ref{eqn:exactLP}, $x \defeq (1+\tfrac \veps 2) \xopt$ has $f_\mu(x) \leq (1+ \tfrac{2\veps}{3}) \arunsopt$.
    \item $f_\mu(x) \geq \arunsopt - \veps n$ for any $x$.
    \item For any $x$ with $f_{\mu}(x)< m$, let $x' = \min \{x, 2 \cdot \1 \}$  Then $f_\mu(x') \leq f_\mu(x)$.
    
\end{itemize}
\end{lemma}
\begin{proof}
We prove the claims in order. For the first claim, let $i$ be the index of the smallest entry of $x$. If $x_i \leq -\tfrac \veps 2$, then since $x_i + \tfrac \veps 4 \le \tfrac{x_i}{2}$ and $\exp$ is always nonnegative,
\[
f_{\mu}(x) \geq \1^\top x + \mu \exp \left(- \frac{x_i}{2\mu}\right) \geq n x_i + \mu \exp \left(- \frac{x_i}{2\mu}\right) \geq m.
\]
The second inequality followed from the definition of $x_i$ and the third from monotonicity of the expression in $x_i$. Thus $x \ge -\tfrac \veps 2 \1$ as claimed. This also implies that $\1^\top x \geq -\tfrac{\veps n}{2}$, so it is straightforward to verify that if some $1 - [\mb x]_i \ge \tfrac \veps 2$, then we would have $f_\mu(x) \ge -\tfrac{\veps n}{2} + (mn)^2 \mu \ge  m$, completing the first claim.

Next, since $\xopt$ is a feasible minimizer to \eqref{eqn:exactLP} we have $\1^\top \xopt = \arunsopt$, $\xopt \geq 0$, and $\mb \xopt \geq 1$. Therefore by construction, $\1^\top x = (1+\tfrac \veps 2) \arunsopt$, $x \geq 0$, and $\mb x = (1+\tfrac \veps 2)  \mb \xopt \geq (1+ \tfrac \veps 2)$. Plugging these into $f_\mu(x)$ yields (since $M^* \ge 1$)
\begin{align*}
f_{\mu}(x) &\leq \Par{1 + \frac \veps 2} \arunsopt + \mu \left( m \exp{\left(\frac{-\veps}{2 \mu} \right)} + n \exp{\left( \frac{-\veps}{4 \mu} \right)} \right) \\
&\leq \Par{1 + \frac \veps 2} \arunsopt + \mu \frac{m + n}{mn} \le \Par{1 + \frac{2\veps}{3}} M^*.
\end{align*}

For the third claim, suppose $x$ satisfied $f_\mu(x) \leq \arunsopt - \veps n$. As $\arunsopt \leq n$ we may apply the first claim and obtain $x \geq -\tfrac \veps 2$ and $\mb x \geq \1 - \tfrac \veps 2$. Let $v = x + \tfrac \veps 2 \1$, so $v \geq 0$ and $\mb v \geq \1$. This is a contradiction: as $v$ is feasible for \eqref{eqn:exactLP} we have $\arunsopt \leq \1^\top v = \1^\top x + \frac{\veps n}{2} < f_{\mu}(x) + \frac{\veps n}{2} < \arunsopt$.

Finally, note that $f_{\mu}(x) < m$ combined with the first property of $f_{\mu}$ implies that $\mb x \geq (1-\frac{\veps}{2}) \1$, $x \geq -\frac{\veps}{2}$. Assume $x$ has $x_i = \alpha > 2$ for some $i$: we will show that decreasing this coordinate to $2$ decreases $f_\mu$. The only terms affected by changing $x_i$ are $\1^\top x$ (which decreases by $\alpha -2$), $\mu \sum_{j \in V} e^{-\frac{1}{\mu} (x_j + \frac{\veps}{4} )}$, and $\mu \sum_{e \in E} e^{\frac{1}{\mu} (1 - [\mb x]_e)}$. The first of these sums increases by at most 
\[
\mu \exp\left( -\frac{2 + \frac \veps 4}{\mu} \right) - \mu \exp\left( -\frac{\alpha + \frac \veps 4}{\mu} \right) \leq (\alpha - 2) e^{\frac 2\mu} \leq \frac{\alpha - 2}{mn},
\]
while each of the $m$ terms in the second sum increases by at most
\[
\mu \exp\left( -\frac{-1 - [x_R]_j}{\mu} \right) - \mu \exp\left( -\frac{1-\alpha -[x_R]_j }{\mu} \right) \leq (\alpha - 2) e^{\frac{1-\veps}{\mu}} \leq \frac{\alpha - 2}{mn},
\]
for some $x_j$: we use our lower bound on $x$ here. Thus the total change in $f_{\mu}(x)$ is at most $-(\alpha - 2) \left(1 - \frac{1}{mn} - \frac{1}{n} \right) < 0$. 
\end{proof}

We next show that derivatives of $f$ satisfy some useful stability properties, which we now define.

\begin{definition}
\label{def:SOR}
Let $f$ be a convex function. We say that $f$ is $r$-second order robust if for any vectors $x,y$ where $\norm{x-y}_\infty \leq r$ we have (where $\preceq$ is the Loewner order on the positive semidefinite cone)
\[
e^{-1} \nabla^2 f(x) \preceq \nabla^2 f(y) \preceq e \nabla^2 f(x).
\]
\end{definition}

\begin{lemma}
\label{lem:VC-derivatives}
Let 
\[
g_\mu(x) = f_\mu\left(\begin{bmatrix}
\mathbf{I}_L & 0\\ 
0 & -\mathbf{I}_R
\end{bmatrix} x \right)
\]
where $f_\mu$ is defined in \eqref{eqn:approxLP}. The gradient and Hessian of $g_{\mu}$ are
\[
\nabla g_{\mu}(x) = \begin{bmatrix}
\1_L - z_L\\ 
-\1_R + z_R
\end{bmatrix} - \tmb^\top y,\;
\nabla^2 g_{\mu}(x) = \frac{1}{\mu} \left( \diag{z} + \tmb^\top~\diag{y} \tmb \right)
\]
where $y \defeq \exp\left(\frac{1}{\mu} \left( \1 - \tmb x \right) \right)$, $z_L \defeq \exp\left(-\frac{1}{\mu} (x_L + \tfrac \veps 4\1) \right)$, $z_R \defeq  \exp\left(\frac{1}{\mu} (x_R- \tfrac \veps 4\1) \right)$, and $z \defeq [z_L ,z_R]$\footnote{Here, we use the notation $x_L$ to refer to the vector $x$ restricted to the entries in $L \subseteq V$.}. For all $x$,  $\nabla^2 g_{\mu}(x)$ is a nonnegative diagonal matrix plus a weighted graph Laplacian matrix. Further, $g_\mu$ is convex and $\frac{\mu}{2}$-second order robust. Finally, for any $x,\delta$ we have
\[
\delta^\top \nabla^2 g_{\mu}(x) \delta \leq \frac{2}{\mu} \cdot  \delta^{\top} \left( \diag{\mb^\top y} + \diag{z} \right) \delta.
\]
\end{lemma}
\begin{proof}
The gradient and Hessian claims may be verified directly as we have via the chain rule 
\[
\nabla g_\mu (x) = \begin{bmatrix}
\mathbf{I}_L & 0\\ 
0 & -\mathbf{I}_R
\end{bmatrix} \nabla f_\mu\left( \begin{bmatrix}
\mathbf{I}_L & 0\\ 
0 & -\mathbf{I}_R
\end{bmatrix} x\right) \]
and
\[\nabla^2 g_\mu(x) = \begin{bmatrix}
\mathbf{I}_L & 0\\ 
0 & -\mathbf{I}_R
\end{bmatrix} \nabla^2 f_\mu\left( \begin{bmatrix}
\mathbf{I}_L & 0\\ 
0 & -\mathbf{I}_R
\end{bmatrix} x\right)
\begin{bmatrix}
\mathbf{I}_L & 0\\ 
0 & -\mathbf{I}_R
\end{bmatrix}.
\] 
It is clear that $\tfrac z \mu$ is nonnegative and $\tmb^\top \diag{\tfrac y \mu}\tmb$ is a weighted graph Laplacian. Convexity of $g_\mu$ follows from $\nabla^2 g_\mu(x) \succeq 0$ everywhere. To prove second order robustness, for any $x, x'$ with $\norm{x-x'}_\infty \leq \frac{\mu}{2}$,
\[\norm{\tmb x - \tmb x'}_\infty \leq \norm{\tmb}_{\infty} \norm{x-x'}_{\infty} \leq 2 \cdot \frac{\mu}{2} \leq \mu.\] Thus every matrix forming the Hessian of $g_{\mu}$ changes by at most a factor of $e$ multiplicatively, and so $g_{\mu}$ is second order robust. For the final claim, let $|\delta|$, $\delta^2$ be applied entrywise. Then,
\begin{align*}
\delta^\top \tmb^\top \diag{y} \tmb \delta &= \sum_{i \in [m]} y_i [\tmb \delta]_i^2 \leq  \sum_{i \in [m]} y_i [\mb |\delta|]_i^2 \\
&\leq \sum_{i \in [m]} y_i [\mb \delta^2]_i [\mb \1]_i =  2 \sum_{i \in [m]} y_i [\mb \delta^2]_i = 2 \delta^\top \diag{\mb^\top y} \delta.
\end{align*}
The second line used Cauchy-Schwarz and $\mb \1 = 2 \cdot\1$. Since $\diag{z} \succeq 0$, this yields the claim.
\end{proof}

It remains to implement a low-pass and low-space optimization procedure to minimize $g_{\mu}$. We make use of a variant of the box-constrained Newton's method of \cite{CMTV17} stated below.
\begin{definition}
We say that a procedure $\oracle$ is an $(r,k)$-oracle for a class of matrices $\mathcal{M}$ if, for any $\ma \in \mathcal{M}$ and vector $b$, $\oracle(\ma,b)$ returns a vector $z$ where $\norm{z}_\infty \leq rk$ and
\[
z^\top b + \frac{1}{2} z^\top \ma z \leq \frac{1}{2} \left( \min_{\norm{y}_\infty \leq r} y^\top b + \frac{1}{2} y^\top \ma y \right).
\]
\end{definition}

The proof of the following is patterned from \cite{CMTV17}, but tolerates multiplicative error in the Hessian.

\begin{lemma}
Let $f$ be $r$-second order robust with minimizer $\xopt$. For some $x$, let $\mm$ be such that $\frac{1}{2} \nabla^2 f(x)\preceq \mm \preceq 2 \nabla^2 f(x)$. Let $\oracle$ be an $(r,k)$-oracle for a class of matrices $\mathcal{M}$ where $\mm\in\mathcal{M}$. Then if $\max \{ \norm{\xopt}_\infty , \norm{x}_\infty \} \leq R$ and $R \geq r$, for any vector $x$, $x' = x + \frac{1}{k} \oracle\left(\frac{1}{k} \nabla f(x), \frac{2e}{k^2} \mm\right)$ satisfies
\[
f(x') - f(\xopt) \leq \left(1 - \frac{r}{160 k R}\right) (f(x) - f(\xopt)).
\]

\label{lem:boxnewton}
\end{lemma}
\begin{proof}
We note that for any $y$ where $\norm{y-x}_\infty \leq r$ we have
\[
f(y) = f(x) + \nabla f(x)^\top (y-x) + \int_0^1\int_{0}^\beta (y-x)^\top \nabla^2 f(\alpha x + (1-\alpha)y) (y-x) d \alpha d \beta.
\]
As the integral is over a triangle of area $\frac{1}{2}$, the second order robustness of $f$ yields
\begin{equation}\label{eq:ulbounds}\begin{aligned}
f(x) + \nabla f(x)^\top (y-x) + \frac{1}{4e} (y-x)^\top \mm (y-x) \le f(y),\\ f(y) \leq f(x) + \nabla f(x)^\top (y-x) + e (y-x)^\top \mm (y-x).\end{aligned}\end{equation}
Let
\[
\hat{\delta} = \argmin_{\norm{\delta}_\infty \leq r} \frac{1}{k} \nabla f(x)^\top \delta + \frac{e}{k^2} \delta^\top \mm \delta.
\]
Also, let $\Delta = \oracle(\frac{1}{k} \nabla f(x), \frac{2 e}{k^2} \mm)$. By definition of $\oracle$ we have $\norm{\Delta}_{\infty} \leq rk$ and 
\[
\frac{1}{k}\nabla f(x)^\top \Delta + \frac{e }{k^2} \Delta^\top \mm \Delta \leq \frac{1}{2} \left( \frac{1}{k} \nabla f(x)^\top \hat{\delta} + \frac{e }{k^2} \hat{\delta}^\top \mm \hat{\delta} \right).
\]
Thus since $x' = x + \frac{1}{k} \Delta$ we see $\norm{x' - x}_\infty \leq r$, and hence
\begin{equation}\label{eq:xprogress}
f(x') \leq f(x) + \frac{1}{k} \nabla f(x)^\top \Delta + \frac{e}{k^2} \Delta^\top \mm \Delta \leq f(x) + \frac{1}{2} \left( \frac{1}{k} \nabla f(x)^\top \hat{\delta} + \frac{e}{k^2} \hat{\delta}^\top \mm \hat{\delta} \right).
\end{equation}
Define $\tilde{x}=  \frac{r}{2R} (\xopt - x)$, we note that $\norm{\tilde{x}}_{\infty} \leq \frac{r}{2R} \left( \norm{\xopt}_\infty + \norm{x}_\infty \right) \leq r$. By the minimality of $\hat{\delta}$ we observe that $\hat{\delta}$ achieves a smaller value than $c \widetilde{x}$ for any $c \leq 1$. For the choice $c = \frac{1}{4e^2}$ this implies
\begin{align*}
\frac{1}{2} \left( \frac{1}{k} \nabla f(x)^\top \hat{\delta} + \frac{e}{k^2} \hat{\delta}^\top \mm \hat{\delta} \right) &\leq \frac{1}{2} \left( \frac{c}{k} \nabla f(x)^\top \tilde{x} + \frac{c^2 e}{k^2} \tilde{x}^\top \mm \tilde{x} \right) \\
&= \frac{1}{8e^2} \left( \frac{1}{k}\nabla f(x)^\top \tilde{x} + \frac{1}{4ek^2}  \tilde{x}^\top \mm \tilde{x} \right) \\
&\leq \frac{1}{8e^2} \left( f\left(x + \frac{1}{k}\tilde{x}\right) - f(x) \right) \\
&\leq -\frac{r}{16k e^2 R} \left( f(x) - f(\xopt) \right),
\end{align*}
where the third line used \eqref{eq:ulbounds} and the last used convexity of $f$. Plugging this into \eqref{eq:xprogress} yields
\[
f(x') - f(\xopt) \leq \left( 1 - \frac{r}{16k e^2 R} \right) \left( f(x) - f(\xopt) \right).
\]
\end{proof}

Next, the class of matrices which are Laplacians plus a nonnegative diagonal admit an efficient $\oracle$.

\begin{lemma}[Theorem 5.11, \cite{CMTV17}]
Let $\mathcal{M}$ be the class of matrices which consist of a Laplacian matrix plus a nonnegative diagonal. Let $\ma \in \mathcal{M}$ be a matrix with $m$ nonzero entries. There is an algorithm which runs in $\tO(m)$ time and space and implements an $(r,O(\log n))$-oracle for $\ma$.%
\end{lemma}

 We complement this with a known semi-streaming spectral approximation for Laplacians.
\begin{lemma}[Section 2.2, \cite{McGregor14}]\label{lem:lapsparsify}
Let $G = (V,E,w)$ be a weighted undirected graph, given as an insertion-only stream. There is an algorithm which for any $\veps \in (0,1)$ takes one pass and returns a graph $H$ with $O(\veps^{-2} n \log^3(n))$ edges such that $\mlap_{H} \preceq \mlap_{G} \preceq (1+\veps) \mlap_{H}$ using $\tO(m)$ work and $O(\veps^{-2} n \log^3 n)$ space.
\end{lemma}

\newcommand{\BoxConstrainedVC}{\mathsf{BoxConstrainedVC}}
\begin{algorithm}
\DontPrintSemicolon
		\KwInput{Bipartite graph $G = (V,E)$ with vertex partition $V = L \cup R$ and oriented edge-incidence matrix $\tmb$ given as a stream, $\veps > 0$, $(O(\frac{\veps}{\log(n)}),k)$-oracle $\oracle$ for symmetric diagonal dominant matrices with nonpositive off diagonal}
		\KwOutput{$v$ fractional vertex cover for $G$}
		$\mathbf{R} \defeq \begin{bmatrix}
		\mathbf{I}_L & 0 \\
		0 & -\mathbf{I}_R
		\end{bmatrix} $\;
		$x_0 \gets \mathbf{R} \1$, $\mu = \frac{\veps}{4 \log(mn)}$, $T = O(\frac{\log \mu^{-1}}{\mu k})$\;
		$f_\mu(x) \defeq \1^\top x + \mu \left( \sum_{e \in E} e^{\frac{1}{\mu} (1 - [\mb x]_e )} + \sum_{i \in V} e^{-\frac{1}{\mu} (x_i + \frac \veps 4)}  \right)$\;
		$g_{\mu}(x) \defeq f_{\mu}\left( \mathbf{R}
		x\right)$\;
		\For{$0 \le t < T$ }{
		Compute $\nabla g_\mu(x_t)$ and $\mm$, a $2$-approximation of $\nabla^2 g_{\mu}(x_t)$ with Lemma~\ref{lem:lapsparsify}\;
		$x'_{t+1} = x_t + \frac{1}{k}\oracle\left(\frac{1}{k} \nabla g_\mu(x_t), \frac{2e}{k^2} \mm \right)$\;
		$x_{t+1} = \max \{ \min \{ x'_{t+1}, 2\cdot \1 \}, -2 \cdot \1 \}$ entrywise\;
		}
		\Return{$v = \mathbf{R} x_T + \frac{\veps}{2}\1$}
	\caption{$\BoxConstrainedVC(G, \veps, \oracle)$}	\label{alg:boxconstrained}
\end{algorithm}

We now assemble these claims to prove the correctness of our algorithm.

\begin{proposition}
\label{prop:boxnewton}
Let $G = (V,E)$ be an unweighted bipartite graph with bipartition $L, R$. Given access to a $(O(\frac{\veps}{\log n}), k)$-oracle $\oracle$ for $\nabla^2 g_\mu$, Algorithm~\ref{alg:boxconstrained} computes $v$, a feasible vertex cover of size $\arunsopt + \veps n$. For $y \defeq \exp\left( \frac{1}{\mu} (1 - \mb x) \right)$, there exists $w$ with $\norm{w}_1 \leq \mu n$ so $y-w$ is a feasible matching with $\1^\top y \geq \arunsopt - 5 \veps n$. Algorithm~\ref{alg:boxconstrained} requires $\tO(n)$ auxiliary space and $O(\frac{k \log\mu^{-1}}{\mu})$ passes, each requiring $\tO(m)$ work, plus the work and space required by one $\oracle$ call.
\end{proposition}

As the space used by $\oracle$ can be reused between runs, the space overhead will be $\tO(n)$ throughout.
\begin{proof}
By Lemmas~\ref{lem:VC-relaxed} and \ref{lem:VC-derivatives}, $\oracle$ applies to a matrix family containing the Hessian $\nabla^2 g_\mu$. In addition, we see that $\nabla g_\mu$  may be computed in a single pass since it equals
\[
\nabla g_{\mu}(x) = \mathbf{R} \left( \1 - z \right) - \tmb^\top y
\]
for $z=  [z_L, z_R]$ and $z_L, z_R, y$ as in Lemma~\ref{lem:VC-derivatives}. The first term may be computed directly, while the second can be computed analogously to Lemma~\ref{lem:implicitx}. Further, we can obtain a $2$-spectral approximation of $\nabla^2 g_\mu(x)$ in semi-streaming fashion: we compute $z$ in one pass, and sparsify the Laplacian using Lemma~\ref{lem:lapsparsify} while computing $y$ coordinatewise in one pass. Thus in each iteration we perform one pass and one call to $\oracle$: as there are $O(\frac{k \log \mu^{-1}}{\mu })$ iterations: the claimed space, pass, and work bounds follow.

We now prove the correctness of the algorithm. Let $\xopt$ be the minimizer of $g_{\mu}$. By construction, $g_\mu(x_0)  = f_{\mu}(\mathbf{R} x_0) \le 2n \le m$ and the function is monotone decreasing. We note that the point $x_t$ has $\norm{x_t}_\infty \leq 2$ by construction for all $t$, and that $\norm{\xopt}_\infty \leq 2$ by the fourth condition of Lemma~\ref{lem:VC-relaxed}. Further, by Lemma~\ref{lem:VC-derivatives} we see that is $r$-second order robust with $r = \frac{\mu}{2}$. By Lemma~\ref{lem:boxnewton} we obtain
\[
g_{\mu}(x'_{t+1}) - g_{\mu}(\xopt) \leq \left(1 - \frac{\mu}{640 k} \right) \left( g_{\mu}(x_{t}) - g_{\mu}(\xopt)\right)
\]
for any $t$. As the fourth condition of Lemma~\ref{lem:VC-relaxed} implies $g_\mu(x_{t+1}) \leq g_{\mu}(x'_{t+1})$, the final $x_T$ has
\[
g_\mu(x_T) - g_\mu(\xopt) \leq \left(1 - \frac{\mu}{640 k} \right)^T \left( g_\mu(x_0) - g_\mu(\xopt) \right) \leq \frac{\mu^3 n}{9e},
\]
where the second inequality uses  $g_\mu(x_0) \leq 2n$ and $g_\mu(\xopt) \geq 0$. Thus if $\arunsopt$ is the size of the minimum vertex cover of $G$, we obtain by Lemma~\ref{lem:VC-relaxed} that $g_{\mu}(x_T) = f_{\mu}(\mathbf{R} x_T) \leq (1+\frac{2\veps}{3})\arunsopt + \frac{\mu^3 n}{9e}\leq \arunsopt + \frac{\veps}{2} n$, since $\arunsopt \leq \frac{n}{2}$ and $\mu \leq \veps$. To complete the proof, the first condition of Lemma~\ref{lem:VC-relaxed} implies $v = \mathbf{R} x + \frac{\veps}{2} \1 $ is nonnegative with $\mb v \geq 1$ and $\1^\top v \leq \arunsopt + \frac{\veps}{2} n+ \frac{\veps}{2} n = M^* + \veps n$: it is a feasible fractional matching as claimed. 

For the second claim, we observe that for any $\delta$ with $\norm{\delta}_\infty \leq \frac \mu 2$ and $\xhat = x_T + \delta$,
\begin{align*}
\frac{\mu^3 n}{9e}&\geq g_\mu(x_T) - g_\mu(\xopt) \geq g_{\mu}(x_T) -  g_{\mu}(\hat{x}) \\
&\geq - \nabla g_{\mu}(x_T)^\top \delta - e \delta^\top \nabla^2 g_\mu(x_T) \delta \\
&\geq - \nabla g_\mu(x_T)^\top \delta - \frac{2e}{\mu} \delta^\top \left(\diag{\mb^\top y} + \diag{z} \right) \delta;
\end{align*}
the first line used optimality of $\xopt$, the second used second order robustness of $g_\mu$, and the third used Lemma~\ref{lem:VC-derivatives}. We choose $\delta = - \frac{\mu^2}{4 e} \textbf{sign}\left(\nabla g_\mu (x_T) \right)$: this satisfies $\norm{\delta}_\infty \le \tfrac{\mu}{2}$ so
\[
\frac{\mu^2}{4e} \norm{\nabla g_\mu(x_T) }_1  - \frac{\mu^3}{8e} \norm{ \mb^\top y + z}_1 \leq \frac{\mu^3 n}{9e}.
\]
Now note 
\begin{align*}
\norm{ \mb^\top y + z}_1 &= \norm{ [\tmb^\top y]_L + z_L}_1 + \norm{ -[\tmb^\top y]_R + z_R}_1 \\
&\leq \norm{ [\tmb^\top y]_L + z_L -  \1_L}_1 + \norm{ -[\tmb^\top y]_R + z_R - \1_R}_1 + n = \norm{\nabla g_\mu(x_T)}_1 + n. 
\end{align*}
Plugging this in to the above expression and rearranging, we obtain
\[
\frac{\mu^2}{4e} \norm{\nabla g_\mu(x_T)}_1 \leq  \frac{\mu^3 n}{9 e} + \frac{\mu^3 n}{8e} + \frac{\mu^3}{8e} \norm{\nabla g_\mu(x_T)}_1,
\]
so $\norm{\nabla g_{\mu}(x_T)}_1 \leq \mu n$. This implies the existence of $w$ with $\norm{w}_1 \leq \mu n$ where $\1 + w=  \mb^\top y + z$. Since $z$ is nonnegative, $\mb^\top y \leq \1 + w$ with $\norm{w}_1 \leq \mu n$ as desired. We finally lower bound $\1^\top y$. Let $v' = \mathbf{R} x - \mu \1$. Taking the inner product of $\1 + w=  \mb^\top y + z$ with $v'$ and rearranging gives
\[
\1^\top x_L - \1^\top x_R - \mu n - \1^\top y + w^\top v' = y^\top (\tmb x - \1) - \mu \1^\top y + x^\top \mathbf{R} z - \mu \1^\top z,
\]
or 
\[
f_{\mu}(\mathbf{R}x) - \mu n + w^\top v' = y^\top (\tmb x - 1) + \1^\top y + x^\top \mathbf{R} z.
\]
Next, $w^\top v' \geq - \norm{w}_1 \norm{v'}_\infty \geq - 2 \mu n$. Further,  
since $y = \exp\left(\frac{1}{\mu} (1 - \tmb x)\right)$, for any $i$ either $[\tmb x]_i -1 \leq \frac{\veps}{2}$ or $y_i \leq \exp\left(-\frac{1}{\mu} \frac{\veps}{2}\right) = \frac{\veps}{2mn}$. Thus  $y^\top (\tmb x - 1) \leq \veps \1^\top y + \frac{\veps}{2} \arunsopt$. Similarly, we obtain $x^\top \mathbf{R} z \leq \veps n$. Plugging these in, and using $\mu \le \tfrac \veps 4$ and $M^* \le n$ yields
\[
f_{\mu}(\mathbf{R} x) - 3 \mu n -\veps n - \frac{\veps}{2} \arunsopt \leq (1+\veps) \1^\top y \implies \1^\top y \geq \frac{1}{1+\veps} f_{\mu}(\mathbf{R}x) - 3 \veps n.
\]
Here we used $\mu \leq \frac{\veps}{4}$ and $\arunsopt \leq n$. The claim follows from $f_{\mu}(\mathbf{R}x) \geq \arunsopt - \veps n$ via Lemma~\ref{lem:VC-relaxed}.
\end{proof}

\newcommand{\ytilde}{\widetilde{y}}
\begin{proof}[Proof of Theorem~\ref{thm:boxnewton_alg}]
Let $\arunsopt$ be the size of the maximum matching in $G$. We first preprocess the graph $G$ using Proposition~\ref{prop:vertexreduce} within the pass, space, and work budgets to reduce to $n = O(M^*\log(\eps^{-1}))$. Applying Algorithm~\ref{alg:boxconstrained} to $\widetilde{G}$ with $\veps = \frac{\eps}{12 \log(\eps^{-1})}$, by the second claim in Proposition~\ref{prop:boxnewton} we obtain $y  \in \R^{E}_{\ge 0}$ with $\1^\top y \geq  (1-\frac{\eps}{2})\arunsopt - 5 \veps n\geq  (1 - O(\eps)) M^*$, and
\[
\sum_{j \in V} \max \Par{ [\mb^\top y]_j - 1, 0 } \leq \mu n \le \frac{\eps}{12} \arunsopt.
\]
The conclusion follows upon applying Propositions~\ref{prop:cc} and \ref{lem:overflowremoval} to round the approximate matching $y$ to be sparse and feasible, and using the implementation of $\oracle$ from Lemma~\ref{lem:oracleexists}.
\end{proof}
}

\end{document}